\documentclass[letterpaper,11pt]{article}

\usepackage[margin=1in]{geometry}
\usepackage{amsfonts,amsmath,amssymb,amsthm,graphicx}
\usepackage{mathtools}
\usepackage{mathrsfs}
\usepackage{ifthen}
\usepackage{changepage}
\usepackage{enumerate}
\usepackage{scalerel}
\usepackage{accents}
\usepackage{enumitem}
\usepackage{subcaption} 
\usepackage{float}      

\usepackage[utf8]{inputenc} 
\usepackage[T1]{fontenc} 
\usepackage{csquotes}
\usepackage{comment}
\usepackage{bbm}
\usepackage{array}
\usepackage{makecell}
\usepackage{booktabs}
\usepackage{xfrac}

\usepackage[suppress]{color-edits}  
\addauthor{wt}{purple}
\addauthor{hz}{cyan}

\usepackage{tcolorbox}
\DeclareRobustCommand{\mybox}[2][gray!15]{
\begin{tcolorbox}[  
        left=0pt,
        right=0pt,
        top=0pt,
        bottom=0pt,
        colback=#1,
        colframe=#1,
        enlarge left by=0mm,
        boxsep=10pt,
        arc=2pt,outer arc=2pt,
        ]
        #2
\end{tcolorbox}
}

\usepackage{thm-restate}
\usepackage{tikz}
\usepackage{pgfplots}
\pgfplotsset{compat=1.17}
\usetikzlibrary{patterns}
\usepgfplotslibrary{fillbetween}
\usetikzlibrary{intersections}
\usetikzlibrary{positioning,fit,calc,arrows.meta,backgrounds}
\usepackage{fix-cm}
\usepackage[ruled,vlined,linesnumbered]{algorithm2e}
\allowdisplaybreaks
\usetikzlibrary{arrows.meta}

\allowdisplaybreaks
\theoremstyle{plain}
\newtheorem{theorem}{Theorem}[section]
\newtheorem{lemma}[theorem]{Lemma}

\newtheorem{proposition}[theorem]{Proposition}
\newtheorem{corollary}[theorem]{Corollary}

\theoremstyle{plain}
\newtheorem{definition}{Definition}[section] 
\newtheorem{example}[definition]{Example}

\newtheorem{remark}[definition]{Remark}

\theoremstyle{plain}

\usepackage[colorlinks=true, linkcolor=blue!70!black, citecolor=blue!70!black, urlcolor=blue, breaklinks=true]{hyperref}
\usepackage{cleveref}
\Crefname{algocf}{Algorithm}{Algorithms}
\crefname{claim}{claim}{claims}

\newcounter{relctr} 
\everydisplay\expandafter{\the\everydisplay\setcounter{relctr}{0}} 

\AtBeginDocument{} 

\newcommand{\squishlist}{
\begin{list}{{{\small{$\bullet$}}}}
{\setlength{\itemsep}{3pt}      \setlength{\parsep}{1pt}
\setlength{\topsep}{1pt}       \setlength{\partopsep}{0pt}
\setlength{\leftmargin}{1em} \setlength{\labelwidth}{1em}
\setlength{\labelsep}{0.5em} } }
\newcommand{\squishend}{  \end{list}}

\newcommand{\fea}{\omega}

\newcommand{\feanum}{n}
\newcommand{\feaspace}{\Omega}
\newcommand{\labelspace}{\mathcal{Y}}

\newcommand{\feaprob}{\lambda}

\newcommand{\bayesprob}{q}
\newcommand{\prediction}{p}

\newcommand{\labelrea}{Y}
\newcommand{\feasubset}{A}
\newcommand{\feasubsetprob}{\lambda}
\newcommand{\feasubsetlabel}{Y}
\newcommand{\feasubsetmean}{\mu}
\newcommand{\givenclassifier}{f}
\newcommand{\classifiernum}{k}
\newcommand{\baserate}{\bar{\mu}}

\newcommand{\cellnum}{m}

\newcommand{\blackwelldom}{\succ}
\newcommand{\blackwelldomeq}{\succeq}

\newcommand{\predrv}{p}
\newcommand{\reportset}{\mathcal{R}}
\newcommand{\reportvec}{\boldsymbol{\rho}}
\newcommand{\classifierfam}{\mathcal{G}}
\newcommand{\strictdomfam}{\classifierfam_{\succ}}

\newcommand{\strictdomfamdet}{\classifierfam_{\succ}^{\textsc{det}}}

\newcommand{\algo}{\mathcal{A}}
\newcommand{\detclassifierfam}{\classifierfam^{\textsc{det}}}

\newcommand{\predicrefine}{\textsc{Search-Aggregation}}
\newcommand{\predicrefineopt}{\textsc{OPT-Aggregation}}
\newcommand{\predicrefinelink}{\hyperref[prob:predictor-refine]{\predicrefine}}
\newcommand{\predicrefineoptlink}{\hyperref[prob:predictor-refine-opt]{\predicrefineopt}}

\newcommand{\detpredicrefine}{\textsc{Search-DetAgg}}
\newcommand{\detpredicrefineopt}{\textsc{OPT-DetAgg}}
\newcommand{\detpredicrefinelink}{\hyperref[prob:predictor-refine]{\detpredicrefine}}
\newcommand{\detpredicrefineoptlink}{\hyperref[prob:predictor-refine-opt]{\detpredicrefineopt}}

\newcommand{\outputpreind}{h}
\newcommand{\signalvec}{\vv}

\newcommand{\cellnew}{C}
\newcommand{\totalcells}{M}

\newcommand{\mymatrix}{\mathbf{A}}
\newcommand{\weightedbayes}{\mathbf{y}}
\newcommand{\weightedpredic}{\mathbf{b}}

\newcommand{\rowspace}{\textsc{Row}}
\newcommand{\myspan}{\textsc{span}}
\newcommand{\classifier}{g}
\newcommand{\classifiernew}{\classifier\primed}
\newcommand{\CDF}{F}
\newcommand{\SCDF}{I}

\newcommand{\obsspace}{\mathcal{S}}
\newcommand{\obscone}{\mathcal{K}}
\newcommand{\compindex}{\mathcal{C}}
\newcommand{\obslabel}{\widehat{Y}}

\newcommand{\subsetindi}{\mathbf{1}}
\newcommand{\classifierfamdet}{\classifierfam^{\textsc{det}}}

\newcommand{\findomfam}{\strictdomfam^\mathrm{ex}}
\newcommand{\finfam}{\classifierfam^\mathrm{fin}}
\newcommand{\rfinfam}{\classifierfam^\mathrm{ex}}
\newcommand{\predictions}{\mathcal{P}^\mathrm{fin}}
\newcommand{\numpredictions}{M^{\mathrm{fin}}}
\newcommand{\numrays}{R}
\newcommand{\preind}{u}
\newcommand{\rayind}{v}
\newcommand{\cellind}{a}
\newcommand{\claind}{j}
\newcommand{\extremerays}{\obscone^\mathrm{ex}}
\newcommand{\ray}{\mathbf{r}}
\newcommand{\weight}{\alpha}
\newcommand{\coeff}{\beta}
\newcommand{\preprob}{\feasubsetprob}
\newcommand{\targetclassifier}{{\givenclassifier_\targetindex}}
\newcommand{\maxint}{W}

\newcommand{\lexobj}{\text{\sf LexObj}}
\newcommand{\lexprog}{\text{\sc LexProg}}
\newcommand{\lvlind}{\ell}
\newcommand{\optpre}{\prediction^*}
\newcommand{\optprob}{\preprob^*}
\newcommand{\bound}{\bar{\prediction}}
\newcommand{\feasprog}{\text{\sc FeasProg}_{\ell, \bound}}
\newcommand{\suppsize}{\feanum + 2(\totalcells + 1)}
\newcommand{\polytope}{\mathcal{P}}
\newcommand{\erd}{\text{\sf ExRayDec}}
\newcommand{\finpre}{\prediction^\mathrm{fin}}
\newcommand{\absprog}{\text{\sc AbsProg}}
\newcommand{\exprog}{\text{\sc ExProg}}
\newcommand{\premaxint}{(\feanum\maxint)^{10\feanum^3}}
\newcommand{\algsearch}{Polynomial-time algorithm for \predicrefinelink.}
\newcommand{\integernew}{c}
\newcommand{\objval}{P}

\newcommand{\scalescalar}{\beta}
\newcommand{\targetindex}{\tau}
\newcommand{\targetmass}{\bar{\feaprob}}
\newcommand{\predicmass}{\bar{\feaprob}}
\newcommand{\loss}{L}
\newcommand{\exploss}{\loss}
\newcommand{\bayesrisk}{\loss}

\newcommand{\minorant}{\kappa}
\newcommand{\grid}{\mathcal{Q}}
\newcommand{\lossopt}{\mathrm{OPT}}
\newcommand{\lpprog}{\mathrm{LP}}
\newcommand{\algfptas}{{FPTAS for {\predicrefineoptlink}}}
\newcommand{\affslope}{\alpha}
\newcommand{\affintercept}{\beta}
\newcommand{\LPval}{P}
\newcommand{\eps}{\varepsilon}
\newcommand{\affnum}{S}
\newcommand{\affindex}{s}
\newcommand{\preindnew}{\preind}
\newcommand{\targetcompidx}{a}

\newcommand{\NP}{\mathrm{NP}}
\newcommand{\Poly}{\mathrm{P}}

\newcommand{\rss}{\textsc{RestrictedSubsetSum}}
\newcommand{\rsslink}{\hyperref[prob:det-subsetsum]{\rss}}
\newcommand{\hardground}{\feaspace_{\mathrm{hard}}}
\newcommand{\harditemnum}{m}
\newcommand{\hardblocksize}{L}
\newcommand{\hardtotal}{A}
\newcommand{\hardtarget}{B}
\newcommand{\hardposblock}{\feaspace_{+}}
\newcommand{\hardnegblock}{\feaspace_{-}}
\newcommand{\hardplus}{\mathsf{+}}
\newcommand{\hardminus}{\mathsf{-}}
\newcommand{\hardspace}{\mathcal{S}_{\mathrm{hard}}}
\newcommand{\hardbasis}{\mathbf{t}}
\newcommand{\hardeta}{\eta}
\newcommand{\harditemeps}{\delta}
\newcommand{\hardset}{T}
\newcommand{\hardcell}{C}
\newcommand{\yesinst}{\mathrm{YES}}
\newcommand{\noinst}{\mathrm{NO}}
\newcommand{\brierloss}{\loss_{\mathrm{Br}}}
\newcommand{\brierlossval}{\loss_{\mathrm{Br}}}
\newcommand{\detbrieropt}{\operatorname{OPT}^{\textsc{det}}_{\mathrm{Br}}}
\newcommand{\subsuminteger}{c}

\newcommand{\omt}[1]{}

\newcommand{\cc}[1]{\ensuremath{\mathsf{#1}}} 

\newcommand{\prob}[2][]{\mathbb{P}\ifthenelse{\not\equal{}{#1}}{_{#1}}{}\!\left[{\def\givenn{\middle|}#2}\right]}
\newcommand{\expect}[2][]{\mathbb{E}\ifthenelse{\not\equal{}{#1}}{_{#1}}{}\!\left[{\def\givenn{\middle|}#2}\right]}
\newcommand{\indicator}[2][]{\mathbf 1\ifthenelse{\not\equal{}{#1}}{_{#1}}{}\!\left[{\def\givenn{\middle|}#2}\right]}

\newcommand{\R}{\mathbb{R}}

\newcommand{\N}{\mathbb{N}}

\newcommand{\vv}{\mathbf{v}}

\newcommand{\supp}{\cc{supp}}

\newcommand{\vzero}{\mathbf{0}}
\newcommand{\vone}{\boldsymbol{1}}

\newcommand{\primed}{^{\dagger}}
\newcommand{\doubleprimed}{^{\ddagger}}

\newcommand{\dd}{{\mathrm d}}
\newcommand{\reals}{\R}
\newcommand{\integers}{\N}

\newcommand{\xhdr}[1]{\vspace{6pt}\noindent{\bf {#1.}}}
\usepackage[square,numbers]{natbib}
\newif\ifusealphabib
\usealphabibtrue

\ifusealphabib
    \usepackage[square,numbers]{natbib}
    \AtBeginDocument{\let\citet\citep}
    \AtBeginDocument{\let\citealp\citep}
\else
    \usepackage[round,authoryear]{natbib}
\fi

\begin{document}

\title{Algorithmic Expert Aggregation}
\author{
Wei Tang\thanks{Chinese University of Hong Kong. Email: {\tt weitang@cuhk.edu.hk}}
\and
Hanrui Zhang\thanks{Chinese University of Hong Kong. Email: {\tt hanrui@cse.cuhk.edu.hk}}
}
\date{}

\maketitle

\begin{abstract}

Forecast aggregation aims to combine information from multiple Bayesian experts' forecasts into an aggregate forecast. 
In much of this literature, however, the aggregate forecast is optimized for a particular loss or robustness criterion and need not itself be 
\wtedit{calibrated with respect to the outcome:}
the reported forecast need not equal the conditional probability of the outcome given the aggregate forecast itself. 
We introduce and study {\em expert aggregation}, where the goal is instead to aggregate Bayesian experts into a new expert that continues to provide calibrated forecasts. 
In particular, we consider a setting where each input expert reports calibrated predictions, and the aggregator observes the prior distribution over states\hzedit{,} and the input experts, but not the underlying Bayes probabilities of the states. 
We ask whether one can (i) construct a calibrated output expert that Blackwell refines a target expert and cannot be further Blackwell improved using the available information; and (ii) when a proper loss is specified, compute a nearly loss-optimal expert among all such refinements.

We formulate calibrated experts as reduced-form information structures and measure refinement by Blackwell dominance of the induced prediction distributions. 
We characterize the constructible output experts through observable linear information: the input experts generate a linear system whose row space determines which calibrated output predictions are identifiable, and a new expert is constructible exactly when its predictions lie in the associated observable nonnegative cone. 
We establish a sharp algorithmic picture. 
When randomized output experts are allowed, both the refinement-search question (i) and the proper-loss optimization question (ii) admit efficient algorithms.
In contrast, deterministic output experts are computationally intractable: deciding whether a deterministic calibrated refinement exists is $\NP$-hard even with two input experts where target expert is a constant base-rate expert, and deterministic proper-loss optimization admits no multiplicative PTAS unless $\Poly=\NP$, even for the Brier loss.

\end{abstract}

\newpage

\section{Introduction}







A decision maker is debating whether to take a certain action whose payoff depends on whether a certain event happens: For concreteness, the payoff is $1$ if the event happens, and $-1$ otherwise.  To assist, a team of experts make forecasts on the probability that the event happens.  Being professional, these experts always report their best-effort forecast, taking into consideration all the information they possess.  Nonetheless, different experts specialize in different aspects of the matter, leading to different forecasts all being honest from the respective expert's own point of view.  The decision maker then faces a problem: How to aggregate these different forecasts into a single forecast that best informs the decision?

This is the very research question addressed by the line of research on {\em forecast aggregation} (see, e.g., \citealp{S-61,BG-69}).  In an idealized world where the decision maker has all the relevant prior knowledge, the above boils down to Bayesian inference: Given the joint information structure over the event concerned and all experts' forecasts, it is (at least in principle) possible to infer the posterior probability of the event conditioning on all experts' forecasts.  On the other hand, oftentimes the decision maker cannot access the complete information structure (or form a belief about it in the Bayesian sense), in which case one may turn to {\em robust} forecast aggregation, which aims to optimize the aggregate forecast in the worst case over the uncertainty in the information structure \citep{ABS-18,LR-22,KWW-24,GHHKSY-25,FMNW-25}.  This has been a predominant methodology in forecast aggregation under practical considerations.

In this paper, we take another angle at forecast aggregation, and study a natural variant which we term {\em expert aggregation}.  In expert aggregation, instead of one-time forecasts, we directly deal with experts themselves, each as a possibly randomized mapping from possible observations to numerical forecasts.  The goal is to aggregate multiple input experts into a single output expert, which is superior as quantified by Blackwell's informativeness \citep{B-53}, a canonical notion that ranks experts by the ``overall usefulness''.  The distinguishing characteristic of expert aggregation, particularly in the context of robust forecast aggregation, is that we require the output expert to make {\em calibrated} forecasts, i.e., the forecast made by the expert is always exactly the posterior probability of the event conditioning on the forecast.\footnote{
    Note that expert aggregation does not necessarily ``depart'' from robust forecast aggregation --- technically, one may view the former as a more conservative approach to the latter, where by requiring the output expert to make calibrated forecasts, we guarantee exactly the same quality of forecasts even under uncertainty.  The robustness aspect of the problem therefore fades away.  This is discussed in more detail later.
}

Conceptually, expert aggregation corresponds to scenarios where we aim to build a {\em versatile} expert system that, instead of a specific one-time task, provides useful advice for indefinite downstream applications.  To be more concrete, imagine the following scenario: A hospital has collected profiles of years of medical cases, which in particular contain multiple medical experts' judgments as to the likelihood that each patient has a certain latent medical condition, such as subclinical diabetes.  In more technical terms, these profiles contain precisely the following information: the population distribution of patients' observable features, as well as the correlation structure between these observable features and each expert's forecast of the likelihood of the latent condition.\footnote{
    Presumably, different experts' forecasts are independent conditioning on the observable features of the patient.
}
Note that the latter is equivalent to the form of input expert systems discussed above, i.e., mappings from observations to numerical forecasts.  We want to aggregate these profiles into an expert system that assists future medical decisions.  Importantly, we would like not to overfit the expert system to any particular treatment, for the latent condition affects different potential treatments in different ways.  As we will soon discuss, our algorithms aggregate input profiles into a single expert system with the following desiderata, which essentially define the expert aggregation problem:
\begin{itemize}
    \item Calibration: The output expert always makes calibrated forecasts, which is essentially the only reasonable thing to do without specific knowledge of the downstream task, and also ensures that the forecasts are interpretable.  The latter is crucial in certain application domains, including medicine.
    \item Maximal informativeness: The output expert is (approximately) undominated, or Pareto optimal, in terms of Blackwell informativeness, among all experts that can be constructed based on information contained in the input.  In other words, the output expert maximally utilizes input information.
    \item Modularity: The output expert can be used for any downstream task, or even fed into subsequent aggregation algorithms in a blackbox manner.  Performance of the output expert is backed by calibration and maximal informativeness, which are largely agnostic to downstream applications.  Also, downstream applications do not require any access to the input experts of expert aggregation.  This provides a certain form of privacy guarantee, which is also of great value in medicine, as well as many other applications.
    \item Targeted refinement: In addition to the above, we can optionally require the output expert to refine a designated target input expert, meaning that the output expert is never less informative --- and usually in many ways more informative --- than the target input expert.  This is particularly helpful in applications where we wish to retain the full domain expertise of a certain input expert, while incorporating generalist knowledge from other input experts.
\end{itemize}

Summarizing the discussion above, our investigation of expert aggregation seeks to answer two cascading algorithmic questions, which reorganize the desiderata above into a form that makes more technical sense.  First:
\begin{displayquote} 
    {\em {\bf\em Refinement}: Is it possible to restructure the observable information contained in the input experts into a new expert that is more informative than a (possibly trivial) target expert in the Blackwell sense?}
\end{displayquote}
If such improvement is possible, we further ask how far one can push it: 
\begin{displayquote} 
    {\em {\bf\em Optimality}: Can the aggregator construct an {\em undominated} expert, and can it optimize over all constructible experts that dominate the target?}
\end{displayquote}

\subsection{Our Contributions}
We provide algorithmic answers to these questions.
Our results reveal a sharp computational dichotomy between randomized and deterministic expert aggregation: randomized output experts admit efficient search and optimization algorithms, whereas deterministic output experts lead to computational intractability even in highly restricted instances.

\xhdr{Algorithmic expert aggregation and Blackwell dominance}
We introduce and study the algorithmic expert aggregation problem. 
We consider a finite state space $\feaspace=\{\fea_i\}_{i\in[\feanum]}$, \wtedit{where each state represents a feature profile like patient condition, or more generally a profile of decision-relevant characteristics. }
A state $\fea_i$ is realized according to a publicly known prior distribution $\feaprob=(\feaprob_i)_{i\in[\feanum]}$ where $\feaprob_i$ denotes the prior probability for state $\fea_i$.
Conditional on the realized state, a binary outcome/event $\labelrea\in\{0,1\}$ is generated; for example, the outcome may indicate whether a patient has a disease, or whether a user clicks on an ad.
Each state $\fea_i$ has an associated Bayes probability $\bayesprob_i=\prob{\labelrea=1\mid \fea_i}$, namely the conditional probability that the outcome equals one in that state. 
These Bayes probabilities $(\bayesprob_i)_{i\in[\feanum]}$ are unknown to the aggregator.
We follow the literature \citep{ABS-18,DIL-21,GHHKSY-25} and model each expert as a {\em reduced-form information structure}. 
Formally, an expert $\givenclassifier$ is represented by a stochastic mapping $\givenclassifier:\feaspace\to\Delta([0,1])$: when the realized state is $\fea_i$, the expert reports a prediction $\predrv_{\givenclassifier}\sim \givenclassifier(\cdot\mid \fea_i)$. 
This reduced-form representation abstracts away from the expert's underlying signal space and records only the distribution of reports generated by the expert in each state. 
Each reported value is interpreted as the expert's posterior probability that the binary outcome equals one. 
\wtedit{Namely, the predictions are {\em posterior-consistent}, or {\em calibrated}: for every reported prediction value $\prediction$, we have
$\expect{\labelrea\mid \predrv_{\givenclassifier}=\prediction}=\prediction$}.\footnote{\wtedit{Throughout the paper, unless otherwise specified, ``posterior'' refers to the posterior probability of the outcome/event, rather than a posterior distribution over the underlying states.}}

The aggregator observes a collection of input experts $(\givenclassifier_\claind)_{\claind\in[\classifiernum]}$ and the prior distribution $\feaprob$, but does not observe the underlying Bayes probabilities $(\bayesprob_i)_{i\in[\feanum]}$. 
Thus, any output expert must be constructed only from the observable information contained in the input experts, and must remain posterior-consistent, equivalently, calibrated, for every latent environment compatible with those inputs.

To evaluate the informativeness of experts, we use Blackwell dominance \citep{B-53}: an expert is more informative than another if the distribution of its reported predictions is a {\em mean-preserving spread} of the other expert's reported predictions \citep{R-78}.\footnote{This order has the following decision-theoretic interpretation. 
Consider any downstream decision problem with action set $\mathcal A$ and utility $u(a,\labelrea)$, where the decision-maker observes the expert's report $\prediction$ and then chooses an action $a\in\mathcal A$. 
The induced utility from report $\prediction$ is 
$\max_{a\in\mathcal A}\{\prediction u(a,1)+(1-\prediction)u(a,0)\}$,
which is convex in $\prediction$. 
Thus, if expert $\givenclassifier$ Blackwell dominates expert $\givenclassifier\primed$, equivalently if the prediction distribution of $\givenclassifier$ is a mean-preserving spread of that of $\givenclassifier\primed$, then the decision-maker's ex ante optimal expected utility under $\givenclassifier$ is weakly higher than under $\givenclassifier\primed$.}
Formally, let $\CDF_{\givenclassifier}$ denote the CDF of the marginal distribution of the prediction reported by expert $\givenclassifier$, and define the integrated CDF by 
$\SCDF_{\givenclassifier}(t)\triangleq 
\int_0^t \CDF_{\givenclassifier}(\prediction)\,\dd \prediction$. 
We say that expert $\givenclassifier$ is  Blackwell more informative (or Blackwell dominates) expert  $\givenclassifier\primed$, denoted by  $\givenclassifier\blackwelldomeq \givenclassifier\primed$, iff $\SCDF_{\givenclassifier}(t)\ge \SCDF_{\givenclassifier\primed}(t)$ for every $t\in[0,1]$ (see \Cref{fig:blackwell-example} for an illustration). 
We write $\givenclassifier\blackwelldom \givenclassifier\primed$ if the dominance is strict for some $t\in(0,1)$. 

\begin{figure}[H]
\centering
\begin{tikzpicture}[
    x=8cm,
    y=7cm,
    >=stealth,
    axis/.style={->, thick},
    redline/.style={red, very thick},
    blackline/.style={black!60, line width=2.6pt},
    reddash/.style={red, thick, dashed},
    blackdash/.style={black, thick, dashed},
    tick/.style={thick}
]

\def\pOne{0.10}
\def\ptOne{0.30}
\def\pTwo{0.50}
\def\pThree{0.70}
\def\one{1.00}

\def\yRedAtTargetOne{0.0667}
\def\yAtTwo{0.1333}
\def\yAtThree{0.2667}
\def\yAtOne{0.5667}

\draw[axis] (-0.02,0) -- (1.08,0) node[right] {$\prediction$};
\draw[axis] (0,-0.02) -- (0,0.64);

\node[below left] at (0,0) {$0$};
\node[below] at (\one,0) {$1$};

\node[anchor=south west] at (0.02,0.645)
{\textcolor{red}{$\SCDF_{\classifier}(\prediction)$} and $\SCDF_{\givenclassifier_1}(\prediction)$};

\fill[red!8]
    (\pOne,0)
    -- (\ptOne,\yRedAtTargetOne)
    -- (\pTwo,\yAtTwo)
    -- (\ptOne,0)
    -- cycle;

\draw[reddash] (\pOne,0) -- (\pOne,0.60);
\draw[blackdash] (\ptOne,0) -- (\ptOne,0.60);
\draw[reddash] (\pTwo,0) -- (\pTwo,0.60);
\draw[blackdash] (\pThree,0) -- (\pThree,0.60);
\draw[blackdash] (\one,0) -- (\one,0.60);

\draw[blackline]
    (0,0)
    -- (\ptOne,0)
    -- (\pTwo,\yAtTwo)
    -- (\pThree,\yAtThree)
    -- (\one,\yAtOne);

\draw[redline]
    (0,0)
    -- (\pOne,0)
    -- (\ptOne,\yRedAtTargetOne)
    -- (\pTwo,\yAtTwo)
    -- (\pThree,\yAtThree)
    -- (\one,\yAtOne);

\draw[red, thick] (\pOne,-0.01) -- (\pOne,0.01);
\draw[black, thick] (\ptOne,-0.01) -- (\ptOne,0.01);
\draw[red, thick] (\pTwo,-0.01) -- (\pTwo,0.01);
\draw[black, thick] (\pThree,-0.01) -- (\pThree,0.01);

\node[red, below] at (\pOne,-0.025) {$0.1$};
\node[below] at (\ptOne,-0.025) {$0.3$};
\node[red, below] at (\pTwo,-0.025) {$0.5$};
\node[below] at (\pThree,-0.025) {$0.7$};

\end{tikzpicture}
\caption{Graphical illustration of Blackwell dominance in a three equally-realized states example. 
The expert $\classifier$ reports predictions $0.1, 0.5, 0.7$ on the three states, respectively. 
The expert $\givenclassifier_1$ reports prediction $0.3$ on the pooled states $\{\fea_1,\fea_2\}$ and $0.7$ on state $\fea_3$. 
The red curve $\SCDF_\classifier$ lies weakly above the gray curve $\SCDF_{\givenclassifier_1}$ pointwise, and strictly above it on a nonempty interval, so $\classifier\blackwelldom\givenclassifier_1$.}
\label{fig:blackwell-example}
\end{figure}
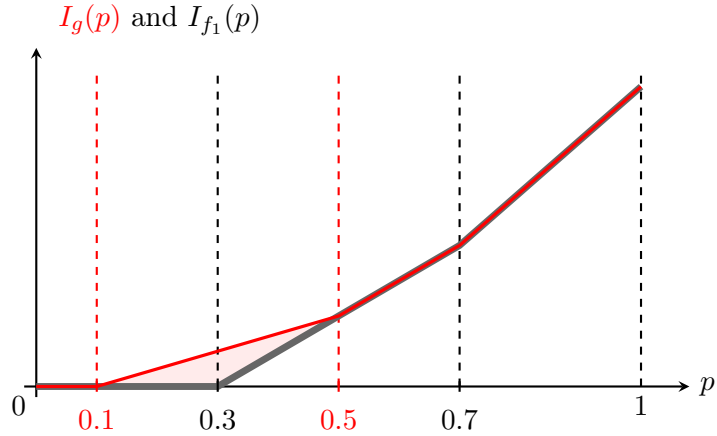

Under this comparison, we focus on a target expert $\targetclassifier$, which is the benchmark expert that the aggregator aims to improve using the information revealed by all input experts. 
In our main formulation, the target expert is one of the input experts, although the same framework also applies to an externally specified benchmark expert. 
We write $\classifierfam$ for the class of constructible experts. 
We call a constructible expert undominated if there is no other expert in $\classifierfam$ that strictly Blackwell dominates it. 
In other words, an undominated constructible expert is one that cannot be further improved, in the Blackwell sense, using the observable information available to the aggregator. 
We study two aggregation problems:
\begin{itemize}
    \item The search problem (see \predicrefinelink) outputs an undominated constructible expert that is strictly more informative than the target expert in the Blackwell sense. 
    If no such strict improvement exists, the problem returns the target expert unchanged.
    
    \item The optimization problem (see \predicrefineoptlink) asks, among all constructible experts that Blackwell dominate the target expert, which one minimizes expected proper loss.
\end{itemize}
We illustrate our problem using following example:
\begin{example}
Consider three equally realized states $\{\fea_1,\fea_2,\fea_3\}$. 
Consider the following two deterministic input experts:
$\givenclassifier_1:
    \{\{\fea_1,\fea_2\}\mapsto 0.3,\ \{\fea_3\}\mapsto 0.7\}$
and 
$\givenclassifier_2:
\{\{\fea_1\}\mapsto 0.1,\ \{\fea_2,\fea_3\}\mapsto 0.6\}$.
Together, these reports reveal enough information to construct an output expert $\classifier$ that reports predictions $0.1$, $0.5$, and $0.7$ on states $\fea_1$, $\fea_2$, and $\fea_3$, respectively; we explain this constructibility calculation in the next paragraph. 

Indeed, this output expert $\classifier$ strictly Blackwell dominates every input expert. 
For example, we know that $\givenclassifier_1$ pools states $\fea_1$ and $\fea_2$ and reports their average prediction $0.3$, while the output expert $\classifier$ separates them into predictions $0.1$ and $0.5$; both experts report $0.7$ on state $\fea_3$. 
Thus, the prediction distribution of $\classifier$ is a mean-preserving spread of that of $\givenclassifier_1$, as illustrated in \Cref{fig:blackwell-example}. 
In this example, $\classifier$ is also undominated, because it already reports the state-level Bayes probabilities.
\end{example}

\xhdr{Constructibility via observable linear information}
We next formalize what it means for an output expert to be constructible from the input experts. 
The basic idea is to record all information revealed by the input experts in a matrix $\mymatrix$. 
Each row of this matrix corresponds to one possible reported prediction of one input expert, and each column corresponds to one state. 
The entry in a row-column pair is the conditional probability that the corresponding expert reports that row's prediction when the realized state is the corresponding column's state. 
Thus, each row is a state-wise reporting-probability vector.

Posterior consistency turns these reporting-probability vectors into {\em linear information} about the unknown Bayes probabilities. 
To see this, suppose a row corresponds to an expert reporting prediction $\prediction$. 
The probability mass of this report is known from the prior distribution and the expert's reporting rule. 
Since the expert is posterior-consistent, the positive-label mass associated with this report must equal $\prediction$ times this probability mass.\footnote{Here the positive-label mass means the joint probability mass that this report is generated and the binary outcome satisfies $\labelrea=1$. For example, if a reporting component is represented by a nonnegative vector $\signalvec$, where $\signalvec_i$ is the probability that this component is selected in state $\fea_i$, then its probability mass is $\sum_i \feaprob_i\signalvec_i$, while its positive-label mass is $\sum_i \feaprob_i\bayesprob_i\signalvec_i$.} 
Equivalently, if $\weightedbayes_i=\feaprob_i\bayesprob_i$ denotes the unknown prior-weighted Bayes probability of state $\fea_i$, then all input experts together induce a linear system $\mymatrix\weightedbayes=\weightedpredic$, where both the matrix $\mymatrix$ and the right-hand side $\weightedpredic$ are determined by the observed input experts and the prior distribution.

The row space of this matrix is the key object in our characterization. 
We call it the {\em observable linear space $\obsspace$}. 
Intuitively, a state-wise vector belongs to $\obsspace$ exactly when its positive-label mass can be inferred from the posterior-consistency equations of the input experts. 
Since an output expert must assign nonnegative reporting probabilities in every state, the relevant feasible components lie in the observable cone $\obscone\triangleq\obsspace\cap\R_+^{\feanum}$. 
Our characterization shows that a constructible expert is precisely one that can be assembled from nonzero components in this cone: the output expert decomposes the all-one vector into components in $\obscone$, and each component reports the posterior prediction given by its identifiable positive-label mass divided by its probability mass.
Thus, constructibility can be checked and optimized over using observable linear information alone. 
This characterization is the foundation for all of our algorithmic results.


\begin{example}
To see how this representation constructs an output expert, consider an instance with three equally realized states $\{\fea_1,\fea_2,\fea_3\}$. 
Consider the instance in \Cref{fig:blackwell-example}.
The four possible reports of two input experts $\givenclassifier_1, \givenclassifier_2$ correspond to the four state-wise reporting-probability vectors
$(1,1,0), (0,0,1), (1,0,0), (0,1,1).$
Thus, the matrix $\mymatrix$ that records the observable reporting information is
\begin{align*}
    \mymatrix
    =
    \begin{pmatrix}
        1 & 1 & 0 \\
        0 & 0 & 1 \\
        1 & 0 & 0 \\
        0 & 1 & 1
    \end{pmatrix}.
\end{align*}
Let $\weightedbayes_i=\feaprob_i\bayesprob_i$ be the prior-weighted Bayes probability of state $\fea_i$. 
Since the states are equally likely, each state has prior mass $1/3$. 
Posterior consistency implies that the right-hand side vector $\weightedpredic$ in the observable linear system $\mymatrix\weightedbayes=\weightedpredic$ is
\begin{align*}
    \weightedpredic
    =
    \begin{pmatrix}
        0.3\cdot (2/3) \\
        0.7\cdot (1/3) \\
        0.1\cdot (1/3) \\
        0.6\cdot (2/3)
    \end{pmatrix}
    =
    \begin{pmatrix}
        1/5 \\
        7/30 \\
        1/30 \\
        2/5
    \end{pmatrix}~.
\end{align*}
For example, the first row corresponds to the report $\{\fea_1,\fea_2\} \mapsto 0.3$ in expert $\givenclassifier_1$. 
The associated component is $\signalvec=(1,1,0)$, whose probability mass, denoted by $\feaprob(\signalvec)$, is $2/3$, and whose observable positive-label mass, denoted by $\obslabel(\signalvec)$, is $\obslabel(\signalvec)=0.3\cdot \feaprob(\signalvec)=1/5$.

The observable linear space is the row space of $\mymatrix$:
$\obsspace
=
\myspan\{(1,1,0),(0,0,1),(1,0,0),(0,1,1)\}$.

The output expert $\classifier$ in \Cref{fig:blackwell-example} decomposes the all-one vector into the three observable components
$(1,0,0),
(0,1,0),
(0,0,1)$,
and reports predictions $0.1$, $0.5$, and $0.7$ on states $\fea_1$, $\fea_2$, and $\fea_3$, respectively. 
This output expert is constructible because each of its reporting components lies in the observable cone $\obscone=\obsspace\cap\R_+^3$. 
\end{example}

\xhdr{Efficient algorithms when aggregating to randomized expert}
We first study the computation complexity of the aggregator's problem when the output expert is allowed to report possibly randomized predictions.
We show that, under this natural randomized-output model, both {\predicrefinelink} and {\predicrefineoptlink} admit efficient algorithmic solutions.

\begin{theorem}[Informal]
With an arbitrary number $\classifiernum$ of input experts and any target expert,
there is a polynomial-time algorithm for solving {\predicrefinelink}. 
Moreover, with an arbitrary number $\classifiernum$ of input experts and any target expert, for every regular proper loss, there is an additive FPTAS for solving {\predicrefineoptlink}, namely, the algorithm outputs a constructible expert that Blackwell dominates the target expert, and its expected proper loss is within the desired additive error of the optimum.
\end{theorem}
The first part of the theorem shows that, for the search problem, randomized aggregation can efficiently find a maximally informative posterior-consistent improvement of a target expert whenever such an improvement exists.
The second part shows that the same tractability extends to optimization problem: among all constructible experts that Blackwell dominates the target expert, one can efficiently find an expert whose expected proper loss is nearly optimal\hzedit{, for any given proper loss}. 
\hzedit{Conceptually, the second part of the result solves a better specified problem, where we still aim to construct a calibrated expert system that can provide persistent help in downstream tasks, but we {\em do} know the loss function that matters in these tasks.} 
Importantly, in both problems, the algorithms do not require access to the latent Bayes probabilities. 
They use only the prior distribution and the input experts, and the output expert remains exactly constructible and exactly \wtedit{calibrated}.

\xhdr{Hardness of deterministic aggregation} 
We then turn to deterministic output experts. 
Deterministic experts are natural in applications where each state, feature profile, or patient profile should receive a stable prediction. 
However, our results show that deterministic aggregation is computationally much harder than randomized aggregation.

\begin{theorem}[Informal]
The deterministic output of problem {\predicrefinelink} is $\NP$-hard, even with only two input experts, a uniform prior distribution, and one of the two input experts being a constant expert that reports the base rate and serves as the target expert.
Moreover, unless $\Poly=\NP$, deterministic output of problem {\predicrefineoptlink} admits no multiplicative PTAS, even for the Brier loss, even with only two input experts, a uniform prior distribution, and a constant target expert.
\end{theorem}
This result shows that the tractability of {\predicrefinelink} relies essentially on allowing randomized output experts. 
The hardness is not driven by a large number of experts, a complicated prior distribution, or a rich target expert: 
it already arises with only two input experts, one of which is the simplest possible expert: a constant base-rate expert. 
Thus, even deciding whether such a constant expert can be strictly improved by a deterministic constructible expert is computationally hard. 
The optimization hardness further shows that deterministic aggregation is not only hard for the search problem, but also hard for loss minimization {\predicrefineoptlink}. 
Even for the standard Brier loss, no efficient multiplicative approximation scheme exists unless $\Poly=\NP$.

\subsection{Our Techniques}
\label{subsec:tech}
\xhdr{Efficient algorithm for solving  {\predicrefinelink}}
We first explain the main idea behind our polynomial-time algorithm for \predicrefinelink. 
The \hzedit{first, straightforward} difficulty is that the class of constructible experts is infinite\hzedit{-dimensional}. 
A constructible expert may use an arbitrary nonzero observable component $\signalvec\in\obscone$, and the prediction attached to this component is the ratio between its observable positive-label mass and its probability mass,
\wtedit{namely $\obslabel(\signalvec)/\feaprob(\signalvec)$.}
Thus, even after constructibility is reduced to the observable cone $\obscone$, it is not clear how to search over all possible output experts, or how to certify that the output expert is undominated under Blackwell dominance.

Our first step is to reduce this infinite\hzedit{-dimensional} search space to a finite\hzedit{-dimensional} one. 
Since the observable cone $\obscone=\obsspace\cap\R_+^{\feanum}$ is a polyhedral cone, every observable component can be decomposed into a nonnegative combination of extreme rays of $\obscone$. 
Let $\extremerays$ denote the normalized set of extreme rays of $\obscone$, and define the finite prediction set
\begin{align*}
    \predictions
    \triangleq
    \left\{
        \frac{\obslabel(\ray)}{\feaprob(\ray)}:
        \ray\in\extremerays
    \right\}~.
\end{align*}
We show that every constructible expert can be weakly Blackwell-improved by an expert whose reported predictions all lie in $\predictions$. 
Therefore, when searching for an undominated expert, it is without loss to restrict attention to experts supported on this finite prediction set.

The second step is to \hzedit{reduce the selection of an undominated expert to a specific lexicographic optimization problem, which is at least superficially more tractable}. 
Order the finite prediction set as $\predictions=\{\finpre_1,\ldots,\finpre_{\numpredictions}\}$, where $\finpre_\preind<\finpre_{\preind+1}$. 
For an expert $\classifier$ supported on $\predictions$, let $\preprob_\preind(\classifier)$ denote the probability that $\classifier$ reports $\finpre_\preind$. 
We consider the lexicographic objective
\begin{align*}
    \lexobj(\classifier)
    =
    \left(
        \preprob_1(\classifier),
        \preprob_2(\classifier),
        \ldots,
        \preprob_{\numpredictions}(\classifier)
    \right),
\end{align*}
and maximize this objective over constructible experts that weakly Blackwell dominate the target expert $\targetclassifier$. 
The reason this objective is useful is that, once the prediction values are ordered, the probability masses, the CDF, and the integrated CDF determine one another. 
Thus, any strict Blackwell improvement strictly improves the lexicographic objective. 
Consequently, any lexicographically optimal feasible expert must be undominated.

\hzedit{Now that we have reduced our task to a well-formulated, finite-dimensional problem, t}he remaining challenge \hzedit{concerns the {\em size} of the search space}: although $\predictions$ is finite, it may be exponentially large. 
We overcome this challenge by proving two structural properties of the lexicographic optimum. 
First, there exists an optimal solution with at most $\feanum+2(M+1)$ positive-probability prediction values, where $M$ is the total number of prediction components across all input experts. 
This small-support property is proved through an auxiliary extreme-ray linear program, which is \hzedit{exponential-size and} used only as a proof device, together with a \hzedit{binding-constraint}-counting argument. 
Second, under the standard rational-input encoding, every prediction value in $\predictions$ has bounded rational complexity: if $W$ upper bounds the numerator and denominator of every rational number in the input, then each prediction value in $\predictions$ has numerator and denominator bounded by $(\feanum W)^{10\feanum^3}$. 
This numerical bound allows us to search for the relevant prediction values exactly.

Combining these properties yields the efficient algorithm. 
The algorithm identifies the positive-probability prediction values of an optimal lexicographic solution one by one. 
Suppose the first $\lvlind$ positive prediction values $\optpre_1,\ldots,\optpre_\lvlind$ and their corresponding probabilities have already been identified. 
For a candidate upper bound $\bound$, we solve a polynomial-size feasibility linear program (see \ref{eq:feasibility-program}). 
This program fixes the already identified lexicographic prefix and asks whether an optimal feasible expert can assign positive additional probability to some prediction value at most $\bound$, while preserving constructibility and the Blackwell dominance constraint $\classifier\blackwelldomeq\targetclassifier$.\footnote{\hzedit{We will discuss in more detail below how these constraints are implemented algorithmically.}} 
Its optimal value is positive if and only if the next positive-probability prediction value satisfies $\optpre_{\lvlind+1}\le \bound$. 
Thus, \ref{eq:feasibility-program} serves as a \hzedit{polynomial-time} comparison oracle for the unknown rational number $\optpre_{\lvlind+1}$.
\hzedit{Note that we cannot easily solve for $\optpre_{\lvlind+1}$ directly, because that would introduce nonlinearity.}

Since the unknown prediction value has bounded rational complexity, we can recover it exactly using {\em accelerated search in the Stern-Brocot tree} \hzedit{(conceptually similar to standard binary search over bounded integers)}. 
After the next prediction value is found, we solve the feasibility program again to obtain its optimal probability. 
Repeating this procedure for at most $\suppsize$ iterations identifies all positive prediction values in an optimal lexicographic solution. 
A final linear program then reconstructs the corresponding constructible expert. 
By the lexicographic optimality argument above, the recovered expert weakly Blackwell dominates the target expert and is undominated within the constructible class, which proves the polynomial-time solvability of \predicrefinelink.

\xhdr{FPTAS for solving {\predicrefineoptlink}}
We first explain the main idea behind our additive FPTAS for {\predicrefineoptlink}. 
The difficulty is that, although constructible experts admit a linear representation through observable components, the proper-loss objective is not directly convex in the natural decision variables. 
Indeed, a natural formulation would choose output components \wtedit{$\signalvec\in\obscone$} (where $\obscone$ is the nonnegative cone of observable state-wise components induced by the input experts), where each component induces its own reported prediction through the ratio between its observable outcome mass and its probability mass. 
The loss contribution of such a component is the probability mass of the component times the Bayes risk evaluated at this induced prediction. 
Since the Bayes-risk function of a proper loss is concave, this gives a concave-perspective objective, and minimizing it is generally non-convex. 
Equivalently, if we introduce the induced prediction as an explicit variable, the constraint tying this prediction to the observable component becomes bilinear. 
Thus, even though constructibility and Blackwell dominance can be expressed through linear constraints, the proper-loss minimization problem is still non-convex.

Our first step is to exploit the structure of proper losses. The Bayes-risk function of a proper loss is concave, and for the class of regular proper losses, it admits a polynomial-size piecewise-linear upper approximation. Replacing the Bayes-risk function by this upper approximation turns the nonlinear loss contribution of each prediction into a linear expression. This gives a linear objective that upper bounds the true expected proper loss, with only an additive approximation error.
The second step is to encode Blackwell dominance through a source-labelled decomposition. Instead of directly optimizing over an arbitrary output expert, we label each output prediction by the prediction component of the target expert from which it is split. For each target prediction value, the source-labelled predictions must preserve both its probability mass and its posterior mean. These two constraints exactly express a martingale coupling from the target expert’s prediction distribution to the output expert’s prediction distribution. By the martingale characterization of Blackwell dominance, this guarantees that the output expert weakly Blackwell dominates the target expert.

Combining these two ingredients gives us a polynomial-size linear program (see program \ref{eq:discretized-LP}): 
its variables are source-labelled observable components; its constraints enforce constructibility and Blackwell dominance exactly; and its objective minimizes the piecewise-linear upper approximation of the proper loss. 
After solving the LP, the algorithm converts each positive-mass component into an output prediction by assigning it the posterior value induced by that component. 
The resulting expert is exactly constructible, exactly posterior-consistent, and exactly Blackwell dominates the target expert; only the objective value is approximated. The soundness and completeness of the LP then show that the output loss is within the desired additive error of the optimum.

\xhdr{Hardness via observable binary vectors}
We next explain the main idea behind our deterministic hardness results. 
The key observation is that deterministic output experts impose an integrality constraint that is absent in the randomized model. 
A deterministic expert partitions the state space into prediction cells. If such an expert is constructible, then each cell of this partition must be {\em observable}: its binary indicator vector must lie in the observable linear space generated by the input experts. Thus, deterministic aggregation is equivalent to asking whether the observable linear space contains useful nontrivial binary vectors.

Our hardness construction reduces this binary-vector requirement from the SubsetSum problem. 
Starting from a restricted subset-sum instance, we construct an aggregation instance with only two input experts: one expert is the constant base-rate expert, which also serves as the target; 
the other expert is a carefully designed randomized auxiliary expert. 
Together, these two experts generate an observable linear space with the following property: every observable binary vector must be constant on two amplified blocks of states, and once the values on these two blocks are fixed, the remaining feasibility condition becomes exactly a subset-sum equation.

This construction creates a direct equivalence. A subset-sum solution exists if and only if the observable linear space contains a nontrivial binary vector. Such a binary vector defines a nonconstant deterministic constructible expert. Since the target expert is the constant base-rate expert, any nonconstant calibrated expert with the same mean strictly Blackwell dominates it. Therefore, deciding whether the constant target expert admits a deterministic constructible improvement is already $\NP$-hard, even with only two input experts and a uniform prior.

For the optimization hardness, we use the same construction but exploit the amplified blocks to create a constant Brier-loss gap. 
On NO instances, the observable linear space contains no nontrivial binary vector, so every deterministic constructible expert must be constant and has Brier loss equal to the base-rate loss. 
On YES instances, the subset-sum certificate yields a two-cell deterministic expert whose predictions are separated away from the base rate, and the amplified blocks guarantee a strictly smaller Brier loss. 
The gap is large enough that any multiplicative PTAS for deterministic {\predicrefineoptlink} would distinguish YES instances from NO instances, implying $\Poly = \NP$ Thus, deterministic aggregation is hard not only for search, but also for loss minimization, even for the Brier loss and even in a highly restricted two-expert instance.

\subsection{Further Related Work}
\label{subsec:related}

\xhdr{Forecast and information aggregation}
Forecast aggregation, and its closely related formulations have been studied across several communities, including machine learning, statistics, economics, and theoretical computer science.
Early foundational work includes opinion pooling and forecast combination \citep{S-61,BG-69}, as well as classical reviews of forecast combination \citep{C-89}.
The broader literature has since developed along different methodological directions.

In machine learning, aggregation is often studied through ensemble methods, where multiple predictors/classifiers are combined to improve predictive performance. 
Classical examples include stacking, bagging, boosting, and random forests \citep{W-92,B-96,S-90,F-95,FS-97,FHT-00,B-01,D-00}. 
A related online-learning perspective studies prediction with expert advice, where multiplicative-weights-type algorithms combine experts' predictions while achieving low regret relative to the best fixed expert
\citep{LW-94,CL-06}.

In statistics, forecast aggregation has been studied both axiomatically and probabilistically.
The axiomatic literature imposes desiderata on aggregation rules, such as unanimity preservation and variants of independence, and characterizes rules such as linear or externally Bayesian pooling under corresponding assumptions \citep{AW-80,G-84,DL-16}.
The probabilistic and Bayesian literature instead models how
forecasts are generated from underlying signals and then derives aggregation rules by Bayesian updating or parametric estimation; see, for example, \citet{SBFMTU-14,FCK-15,EPSU-16,SPU-16,RG-10}.

A common feature of much of this literature is that the reported forecasts are interpreted as probabilistic beliefs induced by the forecasters' information \citep{B-82, GBR-07, SPU-16,ABS-18,GHHKSY-25}. Under this interpretation, a forecast is calibrated with respect to the forecaster's signal: it represents the posterior mean of the target conditional on the information available to that forecaster. 
In our setting, we also require that the constructed forecasts are calibrated which also aligns with this standard belief-updating view. 


Another related line studies agreement-based information aggregation: whether agents who exchange beliefs or predictions and reach agreement aggregate the information held across agents, and under what conditions the resulting agreement coincides with, or approximates, the posterior belief based on pooled information \citep{Aum-76,A-05,KS-23,FNW-23}. 
More recent work also studies computationally tractable agreement and collaborative-prediction protocols, where agents or models iteratively exchange predictions or feedback to improve accuracy while keeping their underlying information private \citep{CGGR-25,CGGGRS-26}. 
Our work is complementary. Rather than modeling an interactive belief-exchange or prediction-exchange process among agents, we take calibrated experts and their reporting rules as inputs, and ask which posterior-consistent output experts can be constructed and optimized using only the observable information contained in those inputs.

\xhdr{Robust forecast aggregation}
Our work concerns the problem of information aggregation.
First introduced in \citet{ABS-18}, 
a recent line of literature takes a robust perspective and asks for aggregators that perform well across large classes of possible
information structures, especially when the correlation structure among experts' signals is
unknown or misspecified.
Recent works study a variety of extensions and variants of this robust information aggregation problem \citep{DIL-21,LR-22,NR-22,KWW-24,GHHKSY-25,GK-25,FMNW-25,CPT-26}.
A complementary strand leverages second-order or higher-order information, such as agents' beliefs about others' answers, to improve aggregation in finite populations \citep{P-04,PSM-17,PS-19,WLC-21,WMH-22,PCK-24,APSTX-25}.

\wtedit{Our work differs from these lines in both its objective and its information constraint. 
Rather than designing a rule that maps several reported forecasts into a single aggregate forecast, or evaluating such a rule under a loss or regret criterion, we study when and how forecasts can be refined using only the observable information contained in the input experts' forecasts, without access to the underlying Bayes probabilities. 
Another important difference is that the output forecast of a robust aggregation rule need not itself be a Bayesian posterior belief. 
By contrast, every forecast constructed in our framework is required to be posterior-consistent: each reported prediction is a valid Bayesian belief under every latent Bayes-probability vector consistent with the input forecasts. 
This leads to a distinct constructibility problem: the goal is to move from forecast aggregation to information aggregation, and to characterize and compute maximally informative constructible output experts under a Blackwell dominance relation.
We discuss the connections in more detail in \Cref{rmk:relation-to-robust}.}

\xhdr{Optimization over information structures}
Our work is also connected to information design and Bayesian persuasion, where a designer chooses an information structure to optimize an objective subject to Bayes plausibility constraint \citep{KG-11,DX-16,BM-19}. 
A growing algorithmic literature studies computational aspects of optimizing over information structures in persuasion and related problems. 
Similar to this literature, we view forecasts as information structures and optimize over feasible posterior distributions.
However, the constraint in our problem is essentially a reverse analogue of Bayesian plausibility: instead of optimizing over all Bayes-plausible information structures, we start from observed forecasts and ask which more informative forecasts are constructible from their observable linear information alone.

Conceptually, our work is also related to recent studies on the comparison and interaction of information structures. \citet{CW-16} study informational substitutes and complements, focusing on how the marginal value of one signal depends on the availability of other signals, while \citet{BFK-22,BFK-24} study when information hierarchies or signal comparisons are robustly valid across decision environments and auxiliary information. 
These works compare given information sources or characterize robust orderings among them; in contrast, our problem is constructive: given several observed forecasts, we ask which more informative forecast can be generated and certified using only their observable implications. 
Recent work \cite{CHJL-26} on calibeating also studies how to post-process external forecasts to obtain calibration, informativeness, and proper-loss guarantees, but in an online learning setting.

\section{Preliminary}
\label{sec:prelim}

\newcommand{\signal}{\sigma}

We consider a stochastic environment that randomly generates a \emph{state} from a finite state space $\feaspace = \{\fea_i\}_{i\in[\feanum]}$.
We denote by $\feaprob_i \in [0, 1]$ the prior probability of the state $\fea_i\in\feaspace$, satisfying $\sum_{i \in [\feanum]} \feaprob_i = 1$. 
Conditional on the realized state, a binary outcome is generated: for each state $\fea_i$, the outcome $\labelrea \in \labelspace \triangleq \{0, 1\}$ is drawn from a Bernoulli distribution with mean $\bayesprob_i \in [0, 1]$.  
We often refer to $\bayesprob_i$ as the Bayes probability of state $\fea_i$.

A (possibly randomized) expert $\givenclassifier$ is identified with the reduced-form information structure\footnote{\wtedit{This reduced-form representation follows the standard signal-based formulation of Bayesian experts (see, e.g., \citealp{ABS-18,GHHKSY-25}).
In a signal-based model, an expert observes a signal $\signal$ generated from an information structure and reports the posterior probability of the positive outcome, $\prediction(\signal) = \prob{\labelrea = 1\mid \signal}$. Conditional on each state $\fea$, the random signal $\signal$ induces a distribution over reported posterior probabilities $\prediction(\signal)\in[0,1]$. The mapping $\givenclassifier(\cdot\mid \fea)$ records exactly this conditional distribution. 
Thus, instead of explicitly modeling the expert's signal space, we work directly with the induced distribution of the expert's reported probability.}} induced by that expert, represented by a stochastic mapping $\givenclassifier:\feaspace\to\Delta([0,1])$.\footnote{With slight abuse of terminology, throughout the paper we use ``expert'' to refer to the reduced-form information structure associated with this expert. Formally, expert $\givenclassifier$ is represented by the stochastic mapping $\givenclassifier:\feaspace\rightarrow\Delta([0, 1])$, which specifies, for each realized state $\fea$, the distribution of the expert's reported posterior probability.}
That is, given a state realization $\fea$, the expert reports a prediction $\predrv_\givenclassifier\sim\givenclassifier(\cdot\mid \fea)$.
This prediction is interpreted as the expert's reported posterior probability that the binary outcome equals one, or equivalently, the expert's reported posterior mean of outcome $\labelrea$.
Under this Bayesian interpretation, the report is calibrated: conditional on the expert reporting value $\prediction$, the average realized outcome is $\prediction$. That is, for every reported prediction value $\prediction$,
\begin{align*}
    \expect{\labelrea\mid \predrv_\givenclassifier = \prediction }
    =
    \prediction~.
\end{align*}
A deterministic expert is the special case in which $\givenclassifier(\cdot\mid \fea)$ is a point mass for every state $\fea\in\feaspace$; in this case, we identify $\givenclassifier$ with a function $\givenclassifier:\feaspace\to[0,1]$.
Together with the prior distribution $\feaprob = (\feaprob_i)_{i\in[\feanum]}$ and the corresponding Bayes probabilities $(\bayesprob_i)_{i\in[\feanum]}$, 
the stochastic mapping associated with expert $\givenclassifier$ induces a joint distribution over predictions and outcomes, denoted by $\Gamma_\givenclassifier\in\Delta([0,1]\times\labelspace)$.
Finally, let $\CDF_\givenclassifier$ denote the cumulative distribution function (CDF) of the marginal distribution of the expert's reported prediction $\predrv_\givenclassifier$ under $\Gamma_\givenclassifier$.

We next define the dominance relation used to compare experts.

\xhdr{Blackwell dominance}
Blackwell dominance compares experts through the marginal distributions of their reported predictions.
Given any expert $\givenclassifier$, we define the integrated CDF $\SCDF_\givenclassifier:[0,1]\to\reals_+$ by
\begin{align*}
    \SCDF_\givenclassifier(t)
    \triangleq
    \expect[\prediction\sim \CDF_\givenclassifier]{(t-\prediction)_+}
    =
    \int_0^t \CDF_\givenclassifier(s)\,\dd s,
    \quad
    t\in[0,1]~.
\end{align*}

\begin{definition}[Blackwell dominance]
We say that an expert $\givenclassifier$ weakly Blackwell dominates another expert $\givenclassifier\primed$, denoted by $\givenclassifier\blackwelldomeq\givenclassifier\primed$, if
\begin{align*}
    \SCDF_\givenclassifier(t)
    \ge
    \SCDF_{\givenclassifier\primed}(t)
    \quad
    \text{for all }t\in[0,1]~.
\end{align*}
We say that $\givenclassifier$ strictly Blackwell dominates $\givenclassifier\primed$, denoted by $\givenclassifier\blackwelldom\givenclassifier\primed$, if $\givenclassifier\blackwelldomeq\givenclassifier\primed$ and there exists $t\in(0,1)$ such that $\SCDF_\givenclassifier(t)
>
\SCDF_{\givenclassifier\primed}(t)$.
\end{definition}

Since every prediction generated by an expert is a Bayesian posterior mean, all experts have the same mean prediction, equal to the base rate, $\prob{\labelrea=1} 
= \expect{\labelrea}
=
\sum\nolimits_{i\in[\feanum]}\feaprob_i\bayesprob_i$.
Therefore, $\givenclassifier\blackwelldomeq\givenclassifier\primed$ is equivalent to saying that the prediction distribution of $\CDF_\givenclassifier$ is a mean-preserving spread of the prediction distribution of $\CDF_{\givenclassifier\primed}$.
In this sense, we say that expert $\givenclassifier$ is Blackwell more informative than expert $\givenclassifier\primed$.

\xhdr{Constructible experts}
Fix the state marginal distribution $\feaprob\in\Delta(\feaspace)$ and fix $\classifiernum$ input experts $(\givenclassifier_j)_{j\in[\classifiernum]}$.
The aggregator observes only the prior distribution $\feaprob$ and the input experts $(\givenclassifier_j)_{j\in[\classifiernum]}$, and does not observe the underlying Bayes probabilities $(\bayesprob_i)_{i\in[\feanum]}$.
Thus, any output expert produced by the aggregator must be justified solely by the information contained in $\feaprob$ and $(\givenclassifier_j)_{j\in[\classifiernum]}$.

Formally, an expert $\classifier$ is constructible from $\feaprob$ and $(\givenclassifier_j)_{j\in[\classifiernum]}$ if there exists a mapping $\algo$ such that
$\classifier=\algo\bigl(\feaprob,(\givenclassifier_j)_{j\in[\classifiernum]}\bigr)$,
and this construction is valid uniformly over all Bayes-probability vectors consistent with the input experts.
That is, for every vector $(\bayesprob_i)_{i\in[\feanum]}$ under which the predictions generated by the input experts $(\givenclassifier_j)_{j\in[\classifiernum]}$ are Bayesian posterior means, the output expert
$\classifier=\algo\bigl(\feaprob,(\givenclassifier_j)_{j\in[\classifiernum]}\bigr)$
must also generate predictions that are Bayesian posterior means under the same data distribution.
We write $\classifierfam$ for the set of all such constructible experts.
When there is no ambiguity, we suppress the dependence of $\classifierfam$ on $\feaprob$ and $(\givenclassifier_j)_{j\in[\classifiernum]}$ in the notation.
We also denote by $\detclassifierfam\subseteq\classifierfam$ the set of all constructible deterministic experts.

We next formalize two versions of the expert-aggregation problem: one as a search problem and another as the optimization problem.

\xhdr{The expert-aggregation search problem}
We first introduce the notion of undominated experts:
\begin{definition}[Undominated experts]
Fix a prior distribution $\feaprob=(\feaprob_i)_{i\in[\feanum]}$ and fix $\classifiernum$ input experts $(\givenclassifier_j)_{j\in[\classifiernum]}$.
Let $\strictdomfam\subseteq\classifierfam$ be the set of strictly undominated constructible experts:
\begin{align*}
    \strictdomfam
    \triangleq
    \{\classifier\in\classifierfam:\ 
    \text{there does not exist }\classifier\primed\in\classifierfam
    \text{ such that }\classifier\primed\blackwelldom\classifier\}~.
\end{align*}
\end{definition}

In words, $\strictdomfam$ consists of constructible  experts that cannot be strictly improved under Blackwell dominance by another constructible expert.

\phantomsection
\label{prob:predictor-refine}
\mybox{
$\predicrefine$:\\
\textbf{Input:} a prior distribution $\feaprob=(\feaprob_i)_{i\in[\feanum]}$ and $\classifiernum$ input experts $(\givenclassifier_j)_{j\in[\classifiernum]}$, and target index $\targetindex\in[\classifiernum]$. \\
\textbf{Question:} 
output an expert $\classifier$ such that 
\begin{itemize}
    \item 
    if $\{\classifier\primed\in \strictdomfam: \classifier\primed \blackwelldom\givenclassifier_\targetindex\} \neq \emptyset$, then $\classifier \in \{\classifier\primed\in \strictdomfam: \classifier\primed \blackwelldom\givenclassifier_\targetindex\}$;
    \item
    otherwise $\classifier=\givenclassifier_\targetindex$.
\end{itemize}
}

Equivalently, the problem {\predicrefinelink} is a target-wise aggregation problem. Given a target expert $\givenclassifier_\targetindex$, it asks the aggregator to return a constructible  expert that is both strictly Blackwell more informative than $\givenclassifier_\targetindex$ and undominated within the constructible class $\classifierfam$, whenever such an expert exists. If no such expert exists, the aggregator returns the original target expert $\givenclassifier_\targetindex$ unchanged. Importantly, the aggregator observes only the prior distribution $\feaprob$ and the input experts $(\givenclassifier_j)_{j\in[\classifiernum]}$, and does not observe the underlying Bayes probabilities $(\bayesprob_i)_{i\in[\feanum]}$. Thus, any returned expert must be constructed and certified using only the observable information contained in $\feaprob$ and $(\givenclassifier_j)_{j\in[\classifiernum]}$.

\xhdr{The expert-aggregation optimization problem}
Let $\loss:[0,1]\times\{0,1\}\to\R$ be a proper loss.
For a report $\prediction\in[0,1]$ and a true label probability $\bayesprob\in[0,1]$, we slightly overload notation and extend $\loss$ linearly in the second argument by defining
\begin{align*}
    \exploss(\prediction,\bayesprob)
    \triangleq
    \bayesprob\loss(\prediction,1)
    +
    (1-\bayesprob)\loss(\prediction,0)~.
\end{align*}
The associated Bayes risk is the univariate function
$\bayesrisk(\bayesprob)
\triangleq
\exploss(\bayesprob,\bayesprob)
=
\bayesprob\loss(\bayesprob,1)
+
(1-\bayesprob)\loss(\bayesprob,0)$.
The loss function is proper if $\exploss(\bayesprob,\bayesprob)\le \exploss(\prediction,\bayesprob)$ for every $\prediction,\bayesprob\in[0,1]$.
Thus, the Bayes-risk function $\bayesrisk$ is concave as the pointwise infimum of affine functions of $\bayesprob$.
For an expert $\classifier$, we define its expected proper loss as 
$\expect{\loss(\predrv_\classifier,\labelrea)} = \expect[\prediction\sim \CDF_\classifier]{\bayesrisk(\prediction)}$.
More formally, we study the following expert-aggregation optimization problem:
\phantomsection
\label{prob:predictor-refine-opt}
\mybox{
$\predicrefineopt$:\\
\textbf{Input:} a proper loss $\loss$, a prior distribution $\feaprob=(\feaprob_i)_{i\in[\feanum]}$ and $\classifiernum$ input calibrated experts $(\givenclassifier_j)_{j\in[\classifiernum]}$, and target index $\targetindex\in[\classifiernum]$. \\
\textbf{Question:} output an expert
$\classifier\in \arg\inf
\left\{
\expect[\prediction\sim \CDF_\classifier]{\bayesrisk(\prediction)}:
\classifier\in\classifierfam
\text{ and }
\classifier\blackwelldomeq\givenclassifier_\targetindex
\right\}$.
}
When the output expert $\classifier$ is required to be deterministic, we refer to the corresponding deterministic-output variants of the search problem {\predicrefinelink} and the optimization problem {\predicrefineoptlink}  as {\detpredicrefinelink} and {\detpredicrefineoptlink}, respectively.
In other words, 
in {\detpredicrefinelink}, we replace the class of undominated experts  $\strictdomfam$ with
\begin{align*}
    \strictdomfamdet
    \triangleq
    \left\{
        \classifier\in\classifierfamdet:
        \text{there does not exist }\classifiernew\in\classifierfamdet
        \text{ such that }\classifiernew\blackwelldom\classifier
    \right\}
\end{align*}
which includes the experts that are strictly undominated within the deterministic constructible class.
In  {\detpredicrefineoptlink}, we replace the constructible class $\classifierfam$ by the deterministic constructible class $\classifierfamdet$.





\wtedit{
We formulate the aggregation problems with a single target expert $\givenclassifier_{\targetindex}$, where $\targetindex\in[\classifiernum]$, only for presentation simplicity.
All of our results and analysis extend directly to the setting in which the output expert is required to weakly Blackwell dominate a collection of target experts (see \Cref{rmk:multiple-dominance} for more details).
In addition, the target expert need not be one of the input experts used to define the feasibility set of the constructible experts, 
it can be any finite-support benchmark expert whose prediction distribution is given as part of the instance. 
The assumption that the target is one of the input experts is made only to streamline notation.\footnote{This target formulation also subsumes the unconstrained refinement version, in which the aggregator simply seeks a maximally informative constructible expert without requiring improvement over a particular input expert. 
Indeed, let $\baserate\triangleq\prob{\labelrea=1}$ denote the base rate. 
Although the Bayes probabilities are unknown, this base rate can be inferred from any input expert.
If we set the target expert to be the trivial base-rate expert $\givenclassifier^{\mathrm{base}}\equiv\baserate$, then the constraint $\classifier\blackwelldomeq f^{\mathrm{base}}$ is vacuous as every constructible expert is a mean-preserving spread of the point mass at $\baserate$. 
Thus, choosing the trivial base-rate expert as the target recovers the version in which the aggregator searches or optimizes over all constructible experts.}
}

\xhdr{Additional notations}
For any nonzero vector $\signalvec\in\R_+^{\feanum}$, we slightly abuse notation and define
\begin{align*}
    \feasubsetprob(\signalvec)
    \triangleq
    \sum\nolimits_{i\in[\feanum]}\feaprob_i\signalvec_i~,
    \quad
    \feasubsetlabel(\signalvec)
    \triangleq
    \sum\nolimits_{i\in[\feanum]}\feaprob_i\bayesprob_i\signalvec_i~,
    \quad
    \feasubsetmean(\signalvec)
    \triangleq
    \frac{\feasubsetlabel(\signalvec)}{\feasubsetprob(\signalvec)}~.
\end{align*}
For a nonempty subset $\feasubset\subseteq[\feanum]$, we write $\subsetindi_{\feasubset}$ for its indicator vector.
We will often use the following finite-support representation.
For an input expert $\givenclassifier_j$ for some $j\in[\classifiernum]$, let its finite support be $\reportset_j = \{\prediction_{j,1},\dots,\prediction_{j,\cellnum_j}\}$ with the prediction values sorted in an increasing order, and $\cellnum_j = |\supp(\CDF_{\givenclassifier_j})|$.
For each prediction value $\prediction_{j, \cellind}\in\reportset_j$, we define its state-wise prediction-probability vector $\reportvec^{j, \cellind}\in[0,1]^{\feanum}$ by
\begin{align}
    \label{eq:state-wise-prediction-prob}
    \reportvec^{j, \cellind}_i
    \triangleq
    \prob{\predrv_{\givenclassifier_j}=\prediction_{j, \cellind}\mid \fea_i}
    \quad
    \text{for every }i\in[\feanum]~.
\end{align}
Then $\sum\nolimits_{a\in[\cellnum_j]}\reportvec^{j, \cellind}=\vone$.
The total probability mass of prediction $\prediction_{j, \cellind}$ is $\feasubsetprob(\reportvec^{j, \cellind})$.
Perfect calibration implies
\begin{align*}
    \feasubsetlabel(\reportvec^{j, \cellind})
    =
    \prediction_{j, \cellind}\cdot
    \feasubsetprob(\reportvec^{j, \cellind})
    \quad
    \text{for every }a\in[\cellnum_j]~.
\end{align*}
Thus, although the Bayes probabilities $(\bayesprob_i)_{i\in[\feanum]}$ are unknown, the outcome mass associated with each prediction value of an input expert is observable from the reported prediction value and its probability mass.

\section{Observable Linear Information}

In this section, we make precise the observable linear information contained in the prior distribution $\feaprob$ and the input experts $(\givenclassifier_j)_{j\in[\classifiernum]}$.
This linear representation is the main object used by the algorithms in the sequel.

For each input expert $\givenclassifier_j$, let its finite prediction support be
$\reportset_j \triangleq \reportset_{\givenclassifier_j}
=\{\prediction_{j,1},\ldots,\prediction_{j,\cellnum_j}\}$.
For each $a\in[\cellnum_j]$, recall that the state-wise prediction-probability vector $\reportvec^{j,a}\in[0,1]^{\feanum}$ is defined by $\reportvec^{j,a}_i=\prob{\predrv_{\givenclassifier_j}=\prediction_{j,a}\mid \fea_i}$.
Since $\feaprob_i>0$ for all $i$, if $\feasubsetprob(\reportvec^{j,a})=0$, then $\reportvec^{j,a}=\vzero$.
As each reported prediction is the expert's Bayesian posterior mean, the component $\reportvec^{j,a}$ satisfies the following posterior-mean equation:
$\feasubsetlabel(\reportvec^{j,a})
=
\prediction_{j,a}\feasubsetprob(\reportvec^{j,a})$.
Let $\compindex
\triangleq
\{(j,a):j\in[\classifiernum],\ a\in[\cellnum_j]\}$
be the set of all prediction components of the input experts, and let $\totalcells\triangleq|\compindex|$.
Define the signal-state matrix $\mymatrix\in\reals_+^{\totalcells\times\feanum}$ by
\begin{align}
    \label{eq:signal-state-matrix}
    \mymatrix_{(j,a),i}
    \triangleq
    \reportvec^{j,a}_i
    \quad
    \text{for every }(j,a)\in\compindex\text{ and }i\in[\feanum]~.
\end{align}
Let $\weightedbayes\in\R^{\feanum}$ be the weighted Bayes vector with coordinates
$\weightedbayes_i\triangleq \feaprob_i\bayesprob_i$.
The observable label-mass vector $\weightedpredic\in\R^{\totalcells}$ is defined by
$\weightedpredic_{j,a}
\triangleq
\prediction_{j,a}\feasubsetprob(\reportvec^{j,a})$.
Then the observed input experts induce the linear system
\begin{align*}
    \mymatrix\weightedbayes
    =
    \weightedpredic~.
\end{align*}
Indeed, the $(j,a)$-th coordinate of $\mymatrix\weightedbayes$ is
$\sum\nolimits_{i\in[\feanum]}\reportvec^{j,a}_i\feaprob_i\bayesprob_i
=\feasubsetlabel(\reportvec^{j,a})$.

\xhdr{Observable vectors and the observable cone}
The linear system 
$\mymatrix\weightedbayes
=
\weightedpredic$ identifies exactly those state-wise linear combinations whose outcome masses can be recovered from the input experts; we therefore collect these identifiable vectors in an observable linear space and its nonnegative cone.
\begin{definition}
\label{defn:observable-vectors}
Fix the prior distribution $\feaprob$, and the input experts $(\givenclassifier_j)_{j\in[\classifiernum]}$, we define their observable linear space
\begin{align*}
    \obsspace
    \triangleq
    \rowspace(\mymatrix)
    =
    \myspan\{\reportvec^{j,a}:(j,a)\in\compindex\}
    \subseteq
    \R^{\feanum},
\end{align*}
where state-wise prediction-probability vectors $(\reportvec^{j,a})_{j ,a}$ are defined as in Eqn.~\eqref{eq:state-wise-prediction-prob}, and the signal-state matrix $\mymatrix$ is defined as in Eqn.~\eqref{eq:signal-state-matrix}.
We also define its observable nonnegative cone
\begin{align*}
    \obscone
    \triangleq
    \obsspace\cap\R^{\feanum}_+~.
\end{align*}
\end{definition}
The cone $\obscone$ consists of nonnegative state-wise routing vectors whose outcome masses are identifiable by linear consequences of the posterior-mean equations induced by the input experts.

We first make precise the sense in which $\obsspace$ exactly characterizes the state-wise vectors whose outcome masses are linearly identifiable from the input experts.

\begin{definition}
\label{def:linearly-observable-vector}
A vector $\signalvec\in\R^{\feanum}$ is called {\em linearly observable} from the input experts if the value
$\signalvec^\top\weightedbayes$ is uniquely determined by the linear observation
$\mymatrix\weightedbayes=\weightedpredic$.
\end{definition}

\begin{proposition}
\label{prop:observable-iff-rowspace}
A vector $\signalvec\in\R^{\feanum}$ is linearly observable if and only if $\signalvec\in\obsspace$.
\end{proposition}

\begin{proof}
Suppose first that $\signalvec\in\obsspace$.
Then there exists $\xi\in\R^{\totalcells}$ such that $\signalvec=\mymatrix^\top\xi$.
For any $\weightedbayes$ satisfying $\mymatrix\weightedbayes=\weightedpredic$, we have
\begin{align*}
    \signalvec^\top\weightedbayes
    =
    \xi^\top \mymatrix\weightedbayes
    =
    \xi^\top\weightedpredic~.
\end{align*}
Thus $\signalvec^\top\weightedbayes$ is uniquely determined by $\weightedpredic$.

Conversely, suppose $\signalvec\notin\obsspace$.
Since $\obsspace=\rowspace(\mymatrix)=(\ker \mymatrix)^\perp$, there exists $z\in\ker\mymatrix$ such that $\signalvec^\top z\neq0$.
Fix any vector $\weightedbayes^{(0)}$ satisfying $\mymatrix\weightedbayes^{(0)}=\weightedpredic$.
For every scalar $\scalescalar\in\R$,
\begin{align*}
    \mymatrix(\weightedbayes^{(0)}+\scalescalar z)
    =
    \mymatrix\weightedbayes^{(0)}+\scalescalar \mymatrix z
    =
    \weightedpredic~.
\end{align*}
However,
$\signalvec^\top(\weightedbayes^{(0)}+\scalescalar z)
=
\signalvec^\top\weightedbayes^{(0)}+\scalescalar\signalvec^\top z$ depends nontrivially on $\scalescalar$.
Therefore $\signalvec^\top\weightedbayes$ is not uniquely determined by the linear observation.
\end{proof}

\xhdr{Observable label functional}
For each input component, its observable outcome mass is
$\obslabel(\reportvec^{j,a})\triangleq\prediction_{j,a}\feasubsetprob(\reportvec^{j,a})$.
We extend this quantity linearly to all of $\obsspace$.
Equivalently, if $\signalvec\in\obsspace$ and $\signalvec=\mymatrix^\top\xi$ for some $\xi\in\R^{\totalcells}$, we define
\begin{align*}
    \obslabel(\signalvec)
    \triangleq
    \xi^\top\weightedpredic~.
\end{align*}

\begin{lemma}
\label{lem:observable-label-functional}
The functional $\obslabel:\obsspace\to\R$ is well-defined.
Moreover, for every $\signalvec\in\obsspace$ and every Bayes-probability vector $(\bayesprob_i)_{i\in[\feanum]}$ under which the input experts report Bayesian posterior means, we have,
$\obslabel(\signalvec)
=
\feasubsetlabel(\signalvec)
=
\sum\nolimits_{i\in[\feanum]}\feaprob_i\bayesprob_i\signalvec_i$.
\end{lemma}

\begin{proof}
Suppose $\mymatrix^\top\xi=\mymatrix^\top\xi'$.
Then $\mymatrix^\top(\xi-\xi')=\vzero$.
Using $\weightedpredic=\mymatrix\weightedbayes$, we obtain
\begin{align*}
    (\xi-\xi')^\top\weightedpredic
    =
    (\xi-\xi')^\top\mymatrix\weightedbayes
    =
    \left(\mymatrix^\top(\xi-\xi')\right)^\top\weightedbayes
    =
    0~.
\end{align*}
Therefore $\xi^\top\weightedpredic=\xi'^\top\weightedpredic$, so $\obslabel(\signalvec)$ is independent of the representation $\signalvec=\mymatrix^\top\xi$.

Now fix any $\signalvec\in\obsspace$ and write $\signalvec=\mymatrix^\top\xi$.
Then
\begin{align*}
    \obslabel(\signalvec)
    =
    \xi^\top\weightedpredic
    =
    \xi^\top\mymatrix\weightedbayes
    =
    \left(\mymatrix^\top\xi\right)^\top\weightedbayes
    =
    \signalvec^\top\weightedbayes
    =
    \sum\nolimits_{i\in[\feanum]}\feaprob_i\bayesprob_i\signalvec_i
    =
    \feasubsetlabel(\signalvec)~.
\end{align*}
The proof completes.
\end{proof}

The next observation ensures that every nonzero vector in the observable cone defines a valid posterior-mean prediction value.
\begin{lemma}
\label{lem:observable-cone-valid-probability}
For every nonzero $\signalvec\in\obscone$, we have $\feasubsetprob(\signalvec)>0$ and
$\frac{\obslabel(\signalvec)}{\feasubsetprob(\signalvec)} \in [0, 1]$.
\end{lemma}

\begin{proof}
Since $\signalvec\in\R_+^{\feanum}$ is nonzero and $\feaprob_i>0$ for every $i\in[\feanum]$, we have
$\feasubsetprob(\signalvec)=\sum\nolimits_i\feaprob_i\signalvec_i>0$.
By \Cref{lem:observable-label-functional}, we have 
$\obslabel(\signalvec)
=
\sum\nolimits_{i\in[\feanum]}\feaprob_i\bayesprob_i\signalvec_i$.
Since $0\le \bayesprob_i\le1$ and $\signalvec_i\ge0$ for every $i$, it follows that
$0\le \obslabel(\signalvec)\le \sum\nolimits_i\feaprob_i\signalvec_i=\feasubsetprob(\signalvec)$.
\end{proof}

\xhdr{Constructible experts}
We now define the class of (possibly randomized) experts that can be constructed from the observable linear information.

\begin{definition}[Constructible experts]
\label{def:linearly-constructible-predictors}
An expert $\classifier$ is constructible from $\feaprob$ and $(\givenclassifier_j)_{j\in[\classifiernum]}$ if there exist nonzero vectors
$\signalvec^{(1)},\ldots,\signalvec^{(\cellnum_{\classifier})}\in\obscone$ where $\cellnum_{\classifier} = |\supp(\CDF_\classifier)|$ such that
$\sum\nolimits_{\outputpreind\in[\cellnum_{\classifier}]}\signalvec^{(\outputpreind)}=\vone$, and, conditional on state $\fea_i$, expert $\classifier$ outputs prediction
$\prediction_\outputpreind\triangleq \obslabel(\signalvec^{(\outputpreind)})/\feasubsetprob(\signalvec^{(\outputpreind)})$ with probability $\signalvec^{(\outputpreind)}_i$.
Equivalently,\footnote{Throughout this work, we use $\delta_{(A)}$ to denote the Dirac point mass at $A$.}
\begin{align*}
    \classifier(\cdot\mid \fea_i)
    =
    \sum\nolimits_{\outputpreind\in[\cellnum_{\classifier}]}
    \signalvec^{(\outputpreind)}_i
    \delta_{(\prediction_\outputpreind)}
    \quad
    \text{for every }i\in[\feanum]~.
\end{align*}
\end{definition}

The requirement $\signalvec^{(\outputpreind)}\in\obscone$ ensures that the outcome mass of each output prediction can be computed from the observed calibration equations of the input experts.
The requirement $\sum\nolimits_\outputpreind\signalvec^{(\outputpreind)}=\vone$ ensures that the routing probabilities define a valid stochastic kernel from states to predictions.

\begin{lemma}
\label{lem:linear-constructible-calibrated}
Every constructible expert $\classifier\in\classifierfam$ generates Bayesian posterior-mean predictions for every Bayes-probability vector under which the input experts report Bayesian posterior-mean predictions.
\end{lemma}

\begin{proof}
Let expert $\classifier\in\classifierfam$ be represented by nonzero vectors
$\signalvec^{(1)},\ldots,\signalvec^{(\cellnum_{\classifier})}\in\obscone$ with $\sum\nolimits_\outputpreind\signalvec^{(\outputpreind)}=\vone$ where $\cellnum_{\classifier} = |\supp(\CDF_\classifier)|$.
For each output prediction index $\outputpreind$, we define
$\prediction_\outputpreind\triangleq\obslabel(\signalvec^{(\outputpreind)})/\feasubsetprob(\signalvec^{(\outputpreind)})$.
By \Cref{lem:observable-cone-valid-probability}, we know $\prediction_\outputpreind\in[0,1]$.
The probability that $\classifier$ reports $\prediction$ is $\sum\nolimits_{\outputpreind\in[\cellnum_{\classifier}]:\prediction_\outputpreind=\prediction}\feasubsetprob(\signalvec^{(\outputpreind)})$.
The corresponding total outcome mass, by \Cref{lem:observable-label-functional}, is, 
\begin{align*}
    \sum\nolimits_{\outputpreind\in[\cellnum_{\classifier}]:\prediction_\outputpreind=\prediction}\obslabel(\signalvec^{(\outputpreind)})
    =
    \sum\nolimits_{\outputpreind\in[\cellnum_{\classifier}]:\prediction_\outputpreind=\prediction}\prediction_\outputpreind\feasubsetprob(\signalvec^{(\outputpreind)})
    =
    p
    \sum\nolimits_{\outputpreind\in[\cellnum_{\classifier}]:\prediction_\outputpreind=\prediction}\feasubsetprob(\signalvec^{(\outputpreind)})~.
\end{align*}
Thus, we have 
$\expect{\labelrea\mid \predrv_\classifier=\prediction}
= \prediction$.
Since this holds for every reported prediction value $\prediction$, every prediction generated by the expert $\classifier$ is a Bayesian posterior mean.
\end{proof}

We conclude this section with the following remark, which discusses the connection between our expert aggregation framework and the prior literature on robust forecast aggregation.
\begin{remark}
\label{rmk:relation-to-robust}
\wtedit{Robust forecast aggregation \citep{ABS-18,GHHKSY-25,FMNW-25} studies numerical aggregation rules that perform well under uncertainty about the experts' joint information structure. 
Our model can be viewed as imposing an additional posterior-consistency requirement on such aggregation rules: the aggregate report must itself be interpretable as a Bayesian posterior belief for every latent information structure consistent with the input experts. 
Under this requirement, from the analysis above, we know that any admissible output expert must be constructed from the observable linear information contained in the input experts alone; otherwise, two compatible information structures could induce the same aggregate report but different conditional outcome means. Thus, the robustness over unknown  (joint) information structures in robust forecast aggregation formulation is absorbed into a constructibility constraint, and our objective becomes to characterize and optimize over the resulting constructible experts.}
\end{remark}



\section{Efficient Algorithm to \texorpdfstring{\predicrefinelink}{Search-Refine}}
\label{subsec:efficient}

\newcommand{\targetexpertset}{\mathcal{T}}

In this section, we focus on the search problem {\predicrefinelink}, i.e., given a collection of $\classifiernum$ input experts $(\givenclassifier_\claind)_{\claind \in [\classifiernum]}$ and a target index $\targetindex \in [\classifiernum]$, we want to find an undominated expert $\classifier \in \strictdomfam$ among all constructible experts such that $\classifier$ weakly dominates the target expert $\targetclassifier$, i.e., $\classifier \blackwelldomeq \targetclassifier$.

\begin{theorem}
\label{thm:efficient}
    There is a polynomial-time algorithm (\Cref{alg:search}) which solves the {\predicrefinelink} problem.
\end{theorem}

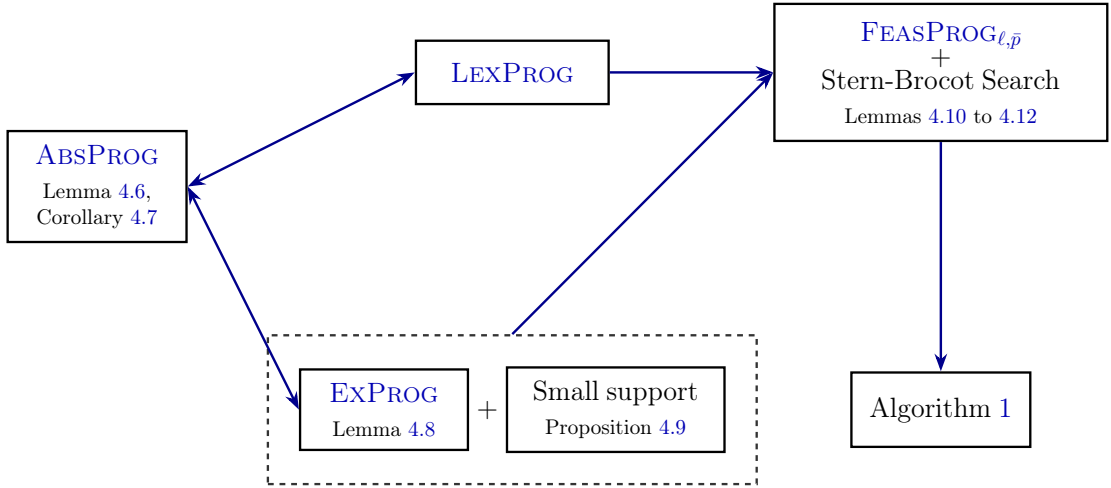
\begin{figure}[t]
    \centering
    \resizebox{0.9\linewidth}{!}{

\begin{tikzpicture}[
    box/.style={
        draw=black,
        very thick,
        rectangle,
        align=center,
        inner sep=7pt,
        minimum height=1.15cm
    },
    bigbox/.style={
        draw=black,
        very thick,
        rectangle,
        align=center,
        inner sep=9pt,
        minimum height=1.8cm
    },
    groupbox/.style={
        draw=black!80,
        very thick,
        dashed,
        rectangle,
        inner sep=16pt
    },
    arrow/.style={
        -{Stealth[length=3.2mm,width=2.3mm]},
        line width=1.25pt,
        draw=blue!55!black
    },
    doublearrow/.style={
        {Stealth[length=3.2mm,width=2.3mm]}-{Stealth[length=3.2mm,width=2.3mm]},
        line width=1.25pt,
        draw=blue!55!black
    }
]

\node[box, text width=2.75cm, minimum height=1.85cm] (abs) at (0,0)
{
    {\Large \hyperref[eq:abstract-program]{$\absprog$}}\\[4pt]
    {\normalsize
    \Cref{lem:lex_obj_monotonicity},\\
    \Cref{cor:lexopt_vs_dom}}
};

\node[box, text width=2.55cm, minimum height=1.25cm] (ex) at (5.20,-4.05)
{
    {\Large \hyperref[eq:extreme-program]{$\exprog$}}\\[4pt]
    {\normalsize \Cref{lem:exprog_properties}}
};

\node[box, text width=3.45cm, minimum height=1.25cm] (support) at (9.40,-4.05)
{
    {\Large Small support}\\[4pt]
    {\normalsize \Cref{cor:small_support}}
};

\begin{scope}[on background layer]
    \node[groupbox, fit=(ex)(support)] (structgroup) {};
\end{scope}

\node[box, text width=3.0cm, minimum height=1.15cm] (lex) at ($(structgroup.north)+(0,4.75)$)
{
    {\Large \hyperref[eq:lexicographical-program]{$\lexprog$}}
};

\node[bigbox, text width=5.4cm, minimum height=1.95cm] (feas) at ($(lex.east)+(6.0,0)$)
{
    {\Large \hyperref[eq:feasibility-program]{\ref{eq:feasibility-program}}}\\[-1pt]
    {\Large $+$}\\[-1pt]
    {\Large Stern-Brocot Search}\\[4pt]
    {\normalsize
    \Cref{lem:numerical_complexity,lem:stern-brocot,lem:feasprog}}
};

\node[box, text width=2.75cm, minimum height=1.35cm] (alg) at ($(feas |- support)$)
{
    {\Large \Cref{alg:search}}
};

\node[
    font=\Large,
    inner sep=2pt
] at ($(ex.east)!0.5!(support.west)$) {$+$};

\draw[doublearrow] (abs.east) -- (lex.west);
\draw[doublearrow] (abs.east) -- (ex.west);

\draw[arrow] (lex.east) -- (feas.west);

\draw[arrow] (structgroup.north) -- (feas.west);

\draw[arrow] (feas.south) -- (alg.north);

\end{tikzpicture}
    }
    \caption{
        Proof flow for \Cref{thm:efficient}.
        The abstract lexicographical program \hyperref[eq:abstract-program]{$\absprog$}
        identifies the canonical undominated output expert.
        The source-labelled prototype \hyperref[eq:lexicographical-program]{$\lexprog$}
        motivates the feasibility oracle, while
        \hyperref[eq:extreme-program]{$\exprog$} and \Cref{cor:small_support}
        provide the structural sparsity of prediction levels needed to make the algorithm polynomial-time.
    }
    \label{fig:search-proof-flow}
\end{figure}

\subsection{Algorithm Idea}
\label{subsec:search-algorithm-idea}
We now describe the high-level proof idea behind \Cref{thm:efficient}.
The main difficulty is that the constructible class $\classifierfam$ is infinite: a constructible expert may use arbitrary nonzero atoms $\signalvec \in \obscone$, and the prediction attached to such an atom is $\feasubsetmean(\signalvec) = \obslabel(\signalvec) / \feasubsetprob(\signalvec)$.
The first step is to reduce this infinite search space to a finite space.
Since $\obscone = \obsspace \cap \reals_+^\feanum$ is a polyhedral cone, every atom $\signalvec \in \obscone$ can be decomposed into a nonnegative combination of extreme rays of $\obscone$.
Let $\extremerays$ denote the set of these extreme rays, normalized to have $\ell_1$-norm one, and define the finite prediction set
$\predictions
\triangleq
\{\feasubsetmean(\ray): \ray \in \extremerays\}$.
The key structural observation is that every constructible expert $\classifier \in \classifierfam$ can be weakly Blackwell-improved by an extreme-ray expert $\erd(\classifier)$ whose prediction values all lie in $\predictions$.
Therefore, in searching for an undominated output expert, it is without loss to restrict attention to experts whose reported predictions are supported on $\predictions$.
This reduction is formalized in \Cref{lem:finite_extreme_rays,lem:extreme_ray_decomposition}.

The second step is to choose a canonical undominated expert through lexicographical optimization.
To this end, let the set $\predictions = \{\finpre_1,\ldots,\finpre_{\numpredictions}\}$ with $\finpre_\preind < \finpre_{\preind+1}$.
For any expert $\classifier$ whose reported predictions are supported on $\predictions$, let $\preprob_\preind(\classifier)=\prob{\prediction_\classifier=\finpre_\preind}$. We define the following lexicographical objective
\begin{align*}
    \lexobj(\classifier)
    =
    \left(
    \preprob_1(\classifier),
    \preprob_2(\classifier),
    \ldots,
    \preprob_{\numpredictions}(\classifier)
    \right)~.
\end{align*}
We then consider the abstract lexicographical program (formulated in \ref{eq:abstract-program}), which maximizes $\lexobj(\classifier)$ over all experts $\classifier \in \finfam$ satisfying $\classifier \blackwelldomeq \targetclassifier$.
The reason this objective is useful is that, on the finite prediction set $\predictions$, the probability vector $(\preprob_\preind(\classifier))_\preind$, the CDF $\CDF_\classifier$, and the integrated CDF $\SCDF_\classifier$ all have a one-to-one correspondence w.r.t.\ each other.
Consequently, if another expert strictly Blackwell-dominates expert $\classifier$, then it strictly improves the lexicographical objective.
Thus any optimal solution to \ref{eq:abstract-program} is undominated and weakly dominates the target expert.
This is made formal in \Cref{lem:lex_obj_monotonicity} and \Cref{cor:lexopt_vs_dom}.

The remaining challenge is computational.
Although $\predictions$ is finite, it may be exponentially large, so \ref{eq:abstract-program} cannot be solved by enumerating all prediction values.
We overcome this using two additional structural properties.
First, there exists an optimal solution whose lexicographical objective has at most $\suppsize$ positive-probability prediction values.
To prove this property, we introduce an auxiliary exponential-size program (see \ref{eq:extreme-program}), which optimizes over extreme-ray decompositions directly. 
This program is not solved by the algorithm; it is only a proof device. 
Since \ref{eq:abstract-program} admits an optimal solution in the extreme-ray class, \ref{eq:extreme-program} and \ref{eq:abstract-program} have the same optimal lexicographical value. Thus, any sparsity property proved for an optimal solution of \ref{eq:extreme-program} transfers back to \ref{eq:abstract-program}.
Applying a vertex-counting argument to $\exprog$: the feasibility equations contribute at most $\feanum$ essential binding constraints, while the Blackwell constraints contribute at most two essential binding constraints on each interval between consecutive prediction values of the target expert $\targetclassifier$.
Thus, after accounting for the nonnegativity constraints, some optimal vertex solution uses at most $\suppsize$ positive extreme-ray variables, and hence 
the optimal lexicographical objective has at most $\suppsize$ positive-probability prediction values.
Second, every prediction value in $\predictions$ is a rational number whose numerator and denominator are bounded by $\premaxint$.
These two facts are proved in \Cref{cor:small_support} and \Cref{lem:numerical_complexity}.

\wtedit{
Before turning these structural properties into an efficient algorithm, it is useful to introduce an intermediate source-labelled formulation of the lexicographical problem (see \ref{eq:lexicographical-program}). 
If the finite prediction set $\predictions$ could be enumerated, one could solve the lexicographical problem through the program \ref{eq:lexicographical-program}. In this program, the variables $\signalvec^{\cellind,\preind}\in\obscone$ represent the portion of the target component $\cellind\in[\cellnum_{\targetindex}]$ that is split into the output prediction value $\finpre_{\preind}\in\predictions$. The source-mass constraint
$\sum\nolimits_{\preind\in[\numpredictions]}
    \feaprob(\signalvec^{\cellind,\preind})
    =
    \feaprob(\reportvec^{\targetindex,\cellind})$
and the source-posterior-mean constraint
$\sum\nolimits_{\preind\in[\numpredictions]}
\obslabel(\signalvec^{\cellind,\preind})
=
\prediction_{\targetindex,\cellind}
\feaprob(\reportvec^{\targetindex,\cellind})$
encode a martingale coupling from the prediction distribution of the target expert $\targetclassifier$ to the output prediction distribution, and therefore enforce the Blackwell constraint $\classifier\blackwelldomeq\targetclassifier$ linearly. The program \ref{eq:lexicographical-program} is still not efficient because it has one level for every prediction value in $\predictions$, which may be exponentially large. Nevertheless, it provides the source-labelled perspective from which our final polynomial-size feasibility programs are derived.}

The algorithm, detailed in \Cref{alg:search}, identifies the prediction values that have positive probability in an optimal lexicographical solution one by one. 
Suppose that the first $\lvlind$ such prediction values and their probabilities, $\optpre_1,\ldots,\optpre_\lvlind$ and $\optprob_1,\ldots,\optprob_\lvlind$, have already been identified. 
For a candidate prediction upper bound $\bound$, we solve a polynomial-size feasibility linear program (see \ref{eq:feasibility-program}). 
This program fixes the already identified lexicographical prefix and asks whether the optimal expert can assign positive additional probability to some prediction value at most $\bound$, while preserving constructibility and the Blackwell dominance constraint $\classifier \blackwelldomeq \targetclassifier$. 
\wtedit{Its optimal value is positive if and only if the next positive-probability prediction value satisfies $\optpre_{\lvlind+1} \le \bound$. }
Thus the program \ref{eq:feasibility-program} serves as a comparison oracle for the unknown rational number $\optpre_{\lvlind+1}$.

Since $\optpre_{\lvlind+1}$ belongs to $\predictions$ and therefore has numerator and denominator bounded by $\premaxint$, we can find it exactly using an accelerated search in the Stern-Brocot tree.
After the next prediction value is found, we solve \ref{eq:feasibility-program} once more with $\bound=\optpre_{\lvlind+1}$ to obtain its optimal probability $\optprob_{\lvlind+1}$.
Repeating this procedure for at most $\suppsize$ iterations would identify all positive levels of an optimal solution to \ref{eq:abstract-program}.
The final linear program then reconstructs the corresponding constructible expert.
By the lexicographical optimality argument above, the recovered expert is undominated and weakly Blackwell-dominates $\targetclassifier$, which proves \Cref{thm:efficient}.
A graph illustration of the proof flow is provided in \Cref{fig:search-proof-flow}.

\subsection{Finite Lexicographic Reduction}

\xhdr{Restricting to finite-support experts}
We first establish a seemingly weak but essential structural property of undominated experts $\strictdomfam$: every constructible expert $\classifier \in \classifierfam$ is weakly dominated by a finite-support expert $\classifier\primed$, where $\classifier\primed$ only reports predictions that belong to a certain finite set, \wtedit{denoted by} $\predictions$.
This is established through the following chain of argument: We already know that each prediction comes from an atom $\signalvec$ in the nonnegative cone $\obscone$, where the latter is the intersection of the subspace $\obsspace$ spanned by $\totalcells$ state-wise prediction-probability vectors, and the nonnegative orthant $\reals_+^\feanum$ \wtedit{(see \Cref{lem:linear-constructible-calibrated})}.
Now view $\obscone$ as the collection of nonnegative linear combinations of extreme rays.
Suppose $\totalcells \le n$ (otherwise we can replace $\totalcells$ with the dimensionality of the subspace $\obsspace$).
Then, each extreme ray is uniquely determined by no more than $\totalcells - 1$ ``walls'' of the nonnegative orthant (where each wall is a hyperplane defined by setting a particular coordinate to $0$).
As a result, the number of such extreme rays is finite.
We denote the collection of extreme rays by $\extremerays$ (each normalized such that the $\ell_1$-norm is $1$), formally specified in the proof of the lemma below.
To make the above formal:
\begin{lemma}
\label{lem:finite_extreme_rays}
    There exists a collection of vectors (namely, the extreme rays of $\obscone$) $\extremerays \subseteq \obscone$ such that (1) any vector $\signalvec \in \obscone$ can be written as a nonnegative linear combination
    of vectors in $\extremerays$, and (2) $|\extremerays| \le 2^\feanum$.
\end{lemma}
\begin{proof}
    Let
    \[
        \polytope = \{\signalvec \in \obsspace \mid \signalvec \ge \vzero \text{ and } \|\signalvec\|_1 \le 1\}~.
    \]
    Note that $\polytope$ is a bounded, finite polytope.
    Let $\extremerays$ be the non-zero vertices of $\polytope$.
    Below we argue that $\extremerays$ satisfies both conditions in the lemma.
    For each $\signalvec \in \obscone$, there exists $\signalvec\primed \in \polytope$ such that $\signalvec = \signalvec\primed \cdot \|\signalvec\|_1$, and $\signalvec\primed$ can be written as a convex combination of vectors in $\extremerays$, so $\signalvec$ can be written as a nonnegative linear
    combination of the same vectors.
    To bound the size of $\extremerays$, observe that $\polytope$ is $d$-dimensional (here, $d$ is the dimensionality of $\obsspace$) where $d \le \feanum$, and is defined by $\feanum + 1$ constraints.
    So, each non-zero vertex $\signalvec$ of $\polytope$ is uniquely determined by $d \le \feanum$ binding constraints: the constraint that $\|\signalvec\|_1 \le 1$, together with $d - 1$ constraints, each specifying a dimension $i$ and requiring that $\signalvec_i \ge 0$. 
    There are at most $\binom{\feanum}{d - 1} \le \binom{\feanum}{\feanum / 2} \le 2^\feanum$ ways to choose these constraints, which means $|\extremerays| \le 2^\feanum$.
\end{proof}

We let $\numrays = |\extremerays| \le 2^\feanum$.
Here, we do not try to optimize the bound on $\numrays$, since anything finite suffices for our purposes.
Now we can formally define the finite set of predictions $\predictions$ that we care about: $\predictions = \{\feasubsetmean(\ray) \mid \ray \in \extremerays\}$.
We immediately have:
\begin{corollary}
    $|\predictions| \le |\extremerays| = \numrays \le 2^\feanum$.
\end{corollary}
\begin{proof}
    Every prediction $\finpre \in \predictions$ can be identified with an extreme ray $\ray \in \extremerays$ (though multiple extreme rays may lead to the same prediction), which means the number of distinct predictions in $\predictions$ is no larger than $|\extremerays|$.
\end{proof}
We let $\numpredictions = |\predictions| \le \numrays$.
For notational simplicity, suppose $\predictions = \{\finpre_1, \dots, \finpre_{\numpredictions}\}$, where, without loss of generality, $\finpre_\preind < \finpre_{\preind + 1}$ for each $\preind \in [\numpredictions - 1]$.
Moreover, for each $\preind \in [\numpredictions]$, let $\numrays_\preind$ be the number of rays $\ray \in \extremerays$ such that $\feasubsetmean(\ray) = \finpre_\preind$, and $\extremerays_\preind = \{\ray_{\preind, 1}, \dots, \ray_{\preind, \numrays_\preind}\}$ be the collection of rays $\ray \in \extremerays$ such that $\feasubsetmean(\ray) = \finpre_\preind$.

Now consider any expert $\givenclassifier$.
To construct a finite-support expert that weakly dominates $\givenclassifier$, we perform an ``extreme ray decomposition'' of $\givenclassifier$.
We visit each possible prediction $\prediction$ reported by $\givenclassifier$, take its corresponding state-wise prediction-probability vector $\signalvec$, and replace it with a nonnegative linear combination of vectors in $\extremerays$ that produces $\signalvec$.
As a result, we obtain a ``weak spread'' of $\givenclassifier$, which by definition weakly dominates $\givenclassifier$.
Moreover, every possible prediction reported by this new expert lies within $\predictions$, making it finite-support.

To make the above formal, we properly define several essential restricted classes of constructible experts:
\begin{itemize}
    \item The extreme class $\rfinfam$: An expert $\classifier = \{\wtedit{\signalvec^{(1)}, \dots, \signalvec^{(\cellnum_\classifier)}}\}$ is in $\rfinfam$ iff for every $\outputpreind \in [\cellnum_\classifier]$, there exists $\preind(\outputpreind) \in [\numpredictions]$ such that $\prediction_{\classifier, \outputpreind} = \feasubsetmean(\wtedit{\signalvec^{(\outputpreind)}}) = \finpre_{\preind(\outputpreind)} \in \predictions$, and moreover, there exist coefficients $(\coeff_{\rayind})_{\rayind \in [\numrays_{\preind(\outputpreind)}]} \ge \vzero$ such that $\wtedit{\signalvec^{(\outputpreind)}} = \sum_{\rayind \in [\numrays_{\preind(\outputpreind)}]} \coeff_\rayind \cdot \ray_{\preind(\outputpreind), \rayind}$.
    In words, $\rfinfam$ contains all experts $\classifier$ where each reported prediction $\prediction_{\classifier, \outputpreind}$ lies within $\predictions$, and corresponds to a nonnegative linear combination of extreme rays all leading to the same prediction.
    \item The finite-support class $\finfam$: An expert $\classifier = \{\wtedit{\signalvec^{(1)}, \dots, \signalvec^{(\cellnum_\classifier)}}\}$ is in $\finfam$ iff for every $\outputpreind \in [\cellnum_\classifier]$, $\feasubsetmean(\wtedit{\signalvec^{(\outputpreind)}}) \in \predictions$.
    In words, $\finfam$ contains all experts whose reported predictions always lie within $\predictions$.
    \item The undominated extreme class $\findomfam$: $\findomfam = \rfinfam \cap \strictdomfam = \{\classifier \in \rfinfam \mid \not\exists \classifier\primed \in \classifierfam:\, \classifier\primed \blackwelldom \classifier\}$.
\end{itemize}

Note that by definition, $\findomfam \subseteq \rfinfam \subseteq \finfam$.
Our algorithm works directly on the relatively clean class of $\finfam$, but we also need to operate with the other two classes to establish essential technical properties required by the algorithm.
To begin with, we will show that $\findomfam$ is non-empty.
To this end, we first establish the following claim through extreme ray decompositions sketched above:
\begin{lemma}
\label{lem:extreme_ray_decomposition}
    For each expert $\classifier \in \classifierfam$, there exists $\classifier\primed \in \rfinfam$ (denoted by $\classifier\primed = \erd(\classifier)$) such that $\classifier\primed \blackwelldomeq \classifier$.
\end{lemma}
\begin{proof}
    Suppose $\classifier = \{\wtedit{\signalvec^{(1)}, \dots, \signalvec^{(\cellnum_\classifier)}}\}$.
    Consider $\classifier\primed$ constructed as follows:
    For each $\outputpreind \in [\cellnum_\classifier]$, let $\{\coeff_{\outputpreind, \preind, \rayind}\}_{\preind \in [\numpredictions], \rayind \in [\numrays_\preind]}$ be nonnegative coefficients such that $\wtedit{\signalvec^{(\outputpreind)}} = \sum_{\preind, \rayind} \coeff_{\outputpreind, \preind, \rayind} \cdot \ray_{\preind, \rayind}$.
    The existence of these coefficients is guaranteed by \Cref{lem:finite_extreme_rays}.
    Let
    \[
        \classifier\primed = \left\{\sum_{\outputpreind \in [\cellnum_\classifier], \rayind \in [\numrays_\preind]} \coeff_{\outputpreind, \preind, \rayind} \cdot \ray_{\preind, \rayind}\right\}_{\preind\wtedit{\in[\numpredictions]}}.
    \]
    Note that $\classifier\primed$ reports at most $\numpredictions$ different prediction values by construction.
    Also, $\classifier\primed \blackwelldomeq \classifier$ because the former is obtained by replacing each reported-prediction atom of the latter with a mean-preserving spread.
    In particular, $\{\coeff_{\outputpreind, \preind, \rayind} \cdot \ray_{\preind, \rayind}\}_{\preind, \rayind}$ spreads $\wtedit{\signalvec^{(\outputpreind)}}$ for each $\outputpreind$.
\end{proof}

As a corollary, we show that there exist undominated experts that are finite-support (and in fact extreme-ray-based):
\begin{corollary}
    $\findomfam = \{\classifier \in \rfinfam \mid \not\exists \classifier\primed \in \rfinfam:\, \classifier\primed \blackwelldom \classifier\}$.
    As a result, $\findomfam \ne \emptyset$.
\end{corollary}
\begin{proof}
    We only need to show that for any $\classifier \in \rfinfam$, $\classifier \in \strictdomfam$ if there exists no $\classifier\primed \in \rfinfam$ that strictly dominates $\classifier$.
    Assume the latter condition, and suppose towards a contradiction that there exists $\classifier\primed \in \classifierfam \setminus \rfinfam$ that dominates $\classifier$.
    Then, by \Cref{lem:extreme_ray_decomposition}, there exists $\classifier\doubleprimed = \erd(\classifier\primed) \in \rfinfam$ such that $\classifier\doubleprimed \blackwelldomeq \classifier\primed \blackwelldom \classifier$, a contradiction.
    Now we can write $\findomfam = \{\classifier \in \rfinfam \mid \not\exists \classifier\primed \in \rfinfam:\, \classifier\primed \blackwelldom \classifier\}$, which by definition is non-empty since $\blackwelldom$ is a partial order and $\rfinfam$ is non-empty.
\end{proof}

Note that in the above, the numbers of extreme rays and possible predictions can be exponential in $\feanum$ and / or $\totalcells$.
As a result, the above property does not imply efficient algorithms in any straightforward way (e.g., one might be tempted to set up a linear program based on extreme rays, which would involve exponentially many variables).
Nonetheless, given this property, we can focus on the sub-family of undominated experts whose reported predictions lie in the finite set $\predictions$, i.e., $\finfam \cap \strictdomfam$ (we do not focus on $\findomfam$ directly for technical reasons that will become clear later).
Our problem thus becomes: find an expert $\classifier \in \finfam \cap \strictdomfam$ such that $\classifier \blackwelldomeq \targetclassifier$.

\xhdr{A lexicographical objective}
In light of the above, we focus our attention to the family of experts whose reported predictions lie in $\predictions$ (i.e., $\finfam$, which subsumes $\findomfam$ as a sub-family).
Restricted to this family, we formulate a constrained lexicographical optimization problem, whose solution necessarily satisfies the conditions of our search problem.
Suppose we want to optimize $\classifier$.
We first maximize the probability that $\classifier$ reports the prediction $\finpre_1$; then, among all candidates that report $\finpre_1 \in \predictions$ with the same probability, we further maximize the probability that $\classifier$ reports the prediction $\finpre_2 \in \predictions$, and so on, so forth.
All this is subject to the constraint that $\classifier \blackwelldomeq \targetclassifier$.
Formally, we define the lexicographical objective $\lexobj$, such that for $\classifier \in \finfam$,
\[
    \lexobj(\classifier) = (\preprob_1(\classifier), \preprob_2(\classifier), \dots, \preprob_{\numpredictions}(\classifier))~.
\]
where, as a shorthand and with a slight abuse of notation, we let $\preprob_\preind(\classifier) = \prob{\prediction_\classifier = \finpre_\preind}$ for each $\preind \in [\numpredictions]$.
We omit the dependency on $\classifier$ whenever it is clear from the context.
We formulate the abstract (in the sense that we will make it more concrete later) program below:
\begin{equation}
\tag{$\absprog$}
\label{eq:abstract-program}
\begin{aligned}
    \max \quad
    &
    \lexobj(\classifier) &&
    \\
    \text{s.t.}\quad
    &
    \classifier \in \finfam \text{ and } \classifier \blackwelldomeq \targetclassifier~.
    && 
    \\
\end{aligned}
\end{equation}

Here, $\max$ denotes lexicographical maximization.
This is a well-defined optimization problem.
Moreover, any optimal solution to this lexicographical optimization problem must also be a solution to our search problem {\predicrefinelink}.
This is because of the following property of $\lexobj$:
\begin{lemma}
\label{lem:lex_obj_monotonicity}
    For any \wtedit{pair of experts} $\classifier$ and $\classifier\primed$ in $\finfam$, if $\classifier\primed \blackwelldomeq \classifier$ (resp.\ $\classifier\primed \blackwelldom \classifier$), then $\lexobj(\classifier\primed) \ge \lexobj(\classifier)$ (resp.\ $\lexobj(\classifier\primed) > \lexobj(\classifier)$).
\end{lemma}
\begin{proof}
    Observe that for each expert $\classifier \in \finfam$, the CDF $\CDF_\classifier$ of its reported-prediction distribution is obtained by taking the prefix sum of $\preprob_\preind(\classifier)$ over $\preind$.
    That is, for each $\preind$, $\CDF_\classifier(\finpre_\preind) = \sum_{\preind' \le \preind} \preprob_{\preind'}(\classifier)$.
    So there is a one-to-one correspondence between 
    \wtedit{the prediction probabilities}$ (\preprob_\preind(\classifier))_\preind)$, 
    \wtedit{the CDF} $\CDF_\classifier$, and 
    \wtedit{the integrated CDF} $\SCDF_\classifier$.
    Fix any $\classifier$ and $\classifier\primed$ in $\finfam$ where $\classifier\primed \blackwelldomeq \classifier$.
    By definition, for each $\preind \in [\numpredictions]$, $\SCDF_{\classifier\primed}(\finpre_\preind) \ge \SCDF_{\classifier}(\finpre_\preind)$.
    If equality holds for all $\numpredictions$ possible predictions, then it must be the case that $\preprob_\preind(\classifier) = \preprob_\preind(\classifier\primed)$.
    Otherwise (i.e., when $\classifier\primed \blackwelldom \classifier$), let $\preind > 1$ be the smallest index such that $\SCDF_{\classifier\primed}(\finpre_\preind) > \SCDF_{\classifier}(\finpre_\preind)$.
    It must be the case that (1) $\CDF_{\classifier\primed}(\finpre_{\preind - 1}) > \CDF_{\classifier}(\finpre_{\preind - 1})$, and (2) for each $\preind' < \preind - 1$, $\CDF_{\classifier\primed}(\finpre_{\preind'}) = \CDF_{\classifier}(\finpre_{\preind'})$.
    As a result, it must be the case that (1) for each $\preind' < \preind - 1$, $\preprob_{\preind'}(\classifier\primed) = \preprob_{\preind'}(\classifier)$, and (2) $\preprob_{\preind - 1}(\classifier\primed) > \preprob_{\preind - 1}(\classifier)$.
    This is equivalent to $\lexobj(\classifier\primed) > \lexobj(\classifier)$.
    The second case alone gives the strict version of the claim, and both cases combined give the weak version.
\end{proof}

As a corollary:
\begin{corollary}
\label{cor:lexopt_vs_dom}
    The following are true:
    \begin{itemize}
        \item Any optimal solution $\classifier^*$ to \ref{eq:abstract-program} satisfies: $\classifier^* \in \strictdomfam \cap \finfam$ and $\classifier^* \blackwelldomeq \targetclassifier$.
        \item There exists an optimal solution $\classifier^*$ to \ref{eq:abstract-program} that lies within $\findomfam$.
    \end{itemize}
\end{corollary}
\begin{proof}
    For the first property, we only need to prove that there is no $\classifier \in \classifierfam$ such that $\classifier \blackwelldom \classifier^*$.
    Suppose towards a contradiction that such a $\classifier$ exists.
    Consider $\classifier\primed = \erd(\classifier) \in \rfinfam \subseteq \finfam$.
    By \Cref{lem:extreme_ray_decomposition}, we have $\classifier\primed \blackwelldomeq \classifier \blackwelldom \classifier^*$.
    In particular, $\lexobj$ is well-defined for $\classifier\primed$.
    By \Cref{lem:lex_obj_monotonicity}, $\lexobj(\classifier\primed) > \lexobj(\classifier^*)$, which contradicts the optimality of $\classifier^*$ with respect to $\lexobj$.

    For the second property, let $\classifier^* \in \finfam$ be an arbitrary optimal solution to \ref{eq:abstract-program}, and consider $\classifier = \erd(\classifier^*) \in \rfinfam$.
    By \Cref{lem:extreme_ray_decomposition}, $\classifier \blackwelldomeq \classifier^* \blackwelldomeq \targetclassifier$.
    Moreover, by \Cref{lem:lex_obj_monotonicity}, $\lexobj(\classifier) \ge \lexobj(\classifier^*)$, which means $\classifier \in \rfinfam$ is also an optimal solution to \ref{eq:abstract-program}.
    Together with the first property, we must have $\classifier \in \rfinfam \cap \strictdomfam = \findomfam$, giving the second property.
\end{proof}

The first property further reduces our search problem to solving \wtedit{the program} \ref{eq:abstract-program}, and we will later need the second property for a technical reason.
In particular, note that the second property is not vacuously true: Superficially, it could be the case that all optimal solutions to \ref{eq:abstract-program} lie in $\finfam \setminus \rfinfam$.
This is ruled out by the second property above.
At a high level, it remains to {\em efficiently} solve the lexicographical optimization problem (whose abstract form is \ref{eq:abstract-program}), which takes a few more steps.

\subsection{Polynomial-Size Supports Suffice}
\xhdr{Optimal solutions have small supports}
Next we argue that optimal solutions to the lexicographical optimization problem make only a very manageable number of predictions.
To this end, we examine a specific representation of the extreme family $\rfinfam$ as an exponentially sized polytope.
The representation does not appear directly in our algorithm.
Instead, it serves as a means to establish the small-support property that is crucial for our algorithm.

Consider the polytope defined by nonnegative weights $(\weight_{\preind, \rayind})_{\preind \in [\numpredictions], \rayind \in [\numrays_\preind]}$ where $\sum_{\preind, \rayind} \weight_{\preind, \rayind} \cdot \ray_{\preind, \rayind} = \vone$.
One can check that this polytope is bounded and gives the family of experts whose state-wise prediction-probability vectors are all scaled extreme rays, namely the extreme family $\rfinfam$.
Recall that $\findomfam \subseteq \rfinfam \subseteq \finfam$, which implies that optimizing over $\finfam$ is equivalent to optimizing over $\rfinfam$, because there exist optimal solutions in $\findomfam$, a common subset of both $\finfam$ and $\rfinfam$.

Restricted to experts $\classifier \in \rfinfam$, we still need to enforce the constraint that $\classifier \blackwelldomeq \targetclassifier$, where $\classifier$ is parametrized by $(\weight_{\preind, \rayind})_{\preind \in [\numpredictions], \rayind \in [\numrays_\preind]}$.
To this end, consider again the CDF and integrated CDF.
We need to ensure that for each $x \in [0, 1]$, $\SCDF_\classifier(x) \ge \SCDF_{\targetclassifier}(x)$.
Note that $\SCDF_\classifier$ and $\SCDF_{\targetclassifier}$ are both piecewise linear and convex.
As a result, we only need to enforce the constraint at possible turning points of $\SCDF_\classifier$, which correspond to possible predictions reported by $\classifier$, i.e., $\predictions$.
Recall that for each $\preind \in [\numpredictions]$, $\preprob_\preind = \prob{\prediction_\classifier = \finpre_\preind} = \sum_{\rayind \in [\numrays_\preind]} \weight_{\preind, \rayind} \cdot \feasubsetprob(\ray_{\preind, \rayind})$.
Note that each $\preprob_\preind$ is linear in $(\weight_{\preind, \rayind})$.
For each $\preind \in [\numpredictions]$, observe that $\SCDF_\targetclassifier(\finpre_\preind)$ is a constant, and $\SCDF_\classifier(\finpre_\preind) = \sum_{\preind' < \preind} \preprob_{\preind'} \cdot (\finpre_\preind - \finpre_{\preind'})$, which is linear in $(\preprob_\preind)$, and therefore in $(\weight_{\preind, \rayind})$.
As such, $\classifier \blackwelldomeq \targetclassifier$ becomes $\numpredictions$ linear constraints: For each $\preind \in [\numpredictions]$, $\SCDF_\classifier(\finpre_\preind) \ge \SCDF_\targetclassifier(\finpre_\preind)$.

Now consider the objective under this representation.
The lexicographical objective becomes: Maximize $\preprob_1$, and then $\preprob_2$, etc.
Note that each level of the objective is linear in $(\weight_{\preind, \rayind})$.
As a result, the objective is optimized at some vertex of the above polytope defined over $(\weight_{\preind, \rayind})$.
We will examine this vertex (or in fact, all vertices of the polytope) more closely to establish the small-support property needed later.

Formally, we consider the following linear program where the decision variables are $(\weight_{\preind, \rayind})$, which corresponds to maximizing $\lexobj$ on $\rfinfam$ rather than $\finfam$:
\begin{equation}
\tag{$\exprog$}
\label{eq:extreme-program}
\begin{aligned}
    \max \quad
    &
    \left(\sum_{\rayind \in [\numrays_1]} \feasubsetprob(\ray_{1, \rayind}) \cdot \weight_{1, \rayind},\ \sum_{\rayind \in [\numrays_2]} \feasubsetprob(\ray_{2, \rayind}) \cdot \weight_{2, \rayind},\ \dots,\ \sum_{\rayind \in [\numrays_{\numpredictions}]} \feasubsetprob(\ray_{\numpredictions, \rayind}) \cdot \weight_{\numpredictions, \rayind}\right) &&
    \\
    \text{s.t.}\quad
    &
    \sum\nolimits_{\preind, \rayind} \weight_{\preind, \rayind} \cdot \ray_{\preind, \rayind} = \vone~,
    && 
    \\
    &
    \sum\nolimits_{\preind' < \preind, \rayind \in [\numrays_{\preind'}]} (\finpre_\preind - \finpre_{\preind'}) \cdot \feasubsetprob(\ray_{\preind', \rayind}) \cdot \weight_{\preind', \rayind} \ge \SCDF_\targetclassifier(\finpre_\preind)
    &&
    \preind \in [\numpredictions]~,
    \\
    &
    \weight_{\preind, \rayind} \ge 0
    &&
    \preind \in [\numpredictions], \rayind \in [\numrays_{\preind}]~.
\end{aligned}
\end{equation}

Note that here, once again, the $\max$ denotes lexicographical maximization.
We prove the following properties regarding \ref{eq:extreme-program}:
\begin{lemma}
\label{lem:exprog_properties}
    The following are true:
    \begin{itemize}
        \item Weights $(\weight_{\preind, \rayind})$ satisfy the constraints of \ref{eq:extreme-program} iff the expert $\classifier_\weight$ defined by
        \[
            \classifier_\weight = \left\{\sum\nolimits_{\rayind \in [\numrays_\preind]} 
            \wtedit{\ray_{\preind, \rayind}}
            \cdot \weight_{\preind, \rayind}\right\}_{\preind \in [\numpredictions]}
        \]
        satisfies $\classifier_\weight \in \rfinfam$.
        \item Any optimal solution $\classifier_{\weight^*}$ (induced by $(\weight^*_{\preind, \rayind})$) to \ref{eq:extreme-program} maximizes $\lexobj$ on $\rfinfam$.
        \item There exists one optimal solution to \ref{eq:extreme-program} where at most $\suppsize$ decision variables are strictly positive.
    \end{itemize}
\end{lemma}
\begin{proof}
    The first two properties are relatively straightforward.
    Observe that any $\classifier_\weight$ defined by feasible $(\weight_{\preind, \rayind})$ must lie within $\rfinfam$, because by construction, (1) it is a legitimate expert since $\sum\nolimits_{\preind, \rayind} \weight_{\preind, \rayind} \cdot \ray_{\preind, \rayind} = \vone$, and (2) each reported prediction is produced by a nonnegative linear combination of extreme rays.
    Conversely, given a $\classifier \in \rfinfam$, one can directly read off the nonnegative linear coefficients, which are a feasible solution to \ref{eq:extreme-program}.
    This gives the first property.
    As for the second, simply observe that given the correspondence between $(\weight_{\preind, \rayind})$ and $\classifier_\weight$, the objective is precisely $\lexobj(\classifier_\weight)$.

    Now consider the third property.
    First observe that, since each level of the objective is linear in the decision variables, the optimal objective value is achieved at some vertex of the polytope defined by the constraints.
    This can be proved by iteratively restraining the space of optimal solutions to the sub-polytope that maximizes the $\preind$-th level of the objective, and observing that at the final level, without loss of generality, the objective is maximized at some vertex of the remaining sub-polytope.
    This vertex must also be a vertex of the original polytope, since each restriction of the polytope is equivalent to adding a linear constraint that is tangent to the existing (sub-)polytope, thereby preserving a lower-dimensional face of it.
    We will skip the details for brevity.

    Now consider any vertex $(\weight_{\preind, \rayind})$ of the polytope.
    As the polytope lives in a $\numrays$-dimensional space, we need $\numrays$ {\em linearly independent binding constraints} to pin down each vertex.
    Below we argue that most decision variables must be $0$ in order for sufficiently many independent constraints to be binding.
    This is done through a counting argument: Given all binding constraints, we prune some of them and make sure the remaining ones still span all binding constraints.
    The number of remaining binding constraints must be at least $\numrays$, and many of them must be nonnegativity constraints on the decision variables, thus demanding the corresponding decision variables to be $0$.

    All constraints can be categorized into $3$ types:
    \begin{itemize}
        \item
        The $\feanum$ constraints induced by $\sum_{\preind, \rayind} \weight_{\preind, \rayind} \cdot \ray_{\preind, \rayind} = \vone$.
        These constraints are always binding.
        We do not prune any of these constraints, so they contribute $\feanum$ constraints to the $\numrays$ needed.

        \item
        The $\numpredictions$ constraints induced by $\classifier \blackwelldomeq \targetclassifier$.
        These constraints can further be grouped into $\cellnum_\targetindex + 1$ groups in the following way.
        For simplicity, let $\prediction_{\targetindex, 0} = 0$ and $\prediction_{\targetindex, \cellnum_\targetindex + 1} = 1$.
        Moreover, without loss of generality, suppose $\prediction_{\targetindex, \cellind} \le \prediction_{\targetindex, \cellind + 1}$ for each $\cellind \in \{0, \dots, \cellnum_\targetindex\}$.
        Then, for each $\cellind \in \{0, \dots, \cellnum_\targetindex\}$, we let the $\cellind$-th group contain predictions in $\predictions$ that lie in the interval of $[\prediction_{\targetindex, \cellind}, \prediction_{\targetindex, \cellind + 1})$.


        \begin{figure}[t!]
        \centering
        \begin{tikzpicture}[
    x=8cm,
    y=5.2cm,
    >=stealth,
    axis/.style={->, thick},
    redline/.style={red, very thick},
    blackline/.style={black!60, line width=2.6pt},
    reddash/.style={red, thick, dashed},
    blackdash/.style={black, thick, dashed},
    tick/.style={thick}
]

\def\pOne{0.10}
\def\pfOne{0.17}
\def\pTwo{0.26}
\def\pfTwo{0.38}
\def\pThree{0.52}
\def\pFour{0.57}
\def\pFive{0.63}
\def\pSix{0.71}
\def\pfThree{0.78}
\def\pSeven{0.86}
\def\one{0.93}

\def\yPthree{0.28}
\def\yPsix{0.48}

\draw[axis] (-0.02,0) -- (1.08,0) node[right] {$\prediction$};
\draw[axis] (0,-0.02) -- (0,1.02);

\node[below left] at (0,0) {$0$};

\node[anchor=south] at (0,1.04)
{\textcolor{red}{$\SCDF_\classifier(\prediction)$} and $\SCDF_{\targetclassifier}(\prediction)$};

\draw[reddash] (\pTwo,0) -- (\pTwo,0.90);
\draw[reddash] (\pThree,0) -- (\pThree,0.90);
\draw[reddash] (\pSix,0) -- (\pSix,0.90);
\draw[reddash] (\pSeven,0) -- (\pSeven,0.97);

\draw[blackdash] (\pfTwo,0) -- (\pfTwo,0.90);
\draw[blackdash] (\pfThree,0) -- (\pfThree,0.90);
\draw[blackdash] (\one,0) -- (\one,1.00);

\draw[red, thick] (\pOne,-0.025) -- (\pOne,0.025);
\draw[black, thick] (\pfOne,-0.025) -- (\pfOne,0.025);
\draw[red, thick] (\pFour,-0.025) -- (\pFour,0.025);
\draw[red, thick] (\pFive,-0.025) -- (\pFive,0.025);

\draw[blackline]
    (\pfOne,0.00)
    -- (\pTwo,0.06)
    -- (\pfTwo,0.12)
    -- (\pThree,\yPthree)
    -- (\pSix,\yPsix)
    -- (\pfThree,0.55)
    -- (\pSeven,0.75)
    -- (\one,0.94);

\draw[redline]
    (\pOne,0.00)
    -- (\pTwo,0.08)
    -- (\pThree,\yPthree)
    -- (\pSix,\yPsix)
    -- (\pSeven,0.75)
    -- (\one,1.02);

\node[red, below] at (\pOne,-0.035) {$\prediction_1$};
\node[below] at (\pfOne,-0.035) {$\prediction_{\targetindex,1}$};
\node[red, below] at (\pTwo,-0.035) {$\prediction_2$};
\node[below] at (\pfTwo,-0.035) {$\prediction_{\targetindex,2}$};
\node[red, below] at (\pThree,-0.035) {$\prediction_3$};
\node[red, below] at (\pFour,-0.035) {$\prediction_4$};
\node[red, below] at (\pFive,-0.035) {$\prediction_5$};
\node[red, below] at (\pSix,-0.035) {$\prediction_6$};
\node[below] at (\pfThree,-0.035) {$\prediction_{\targetindex,3}$};
\node[red, below] at (\pSeven,-0.035) {$\prediction_7$};

\end{tikzpicture}
        \label{fig:independent_constraints}
        \caption{Graphical illustration of type-$2$ constraints in \ref{eq:extreme-program}.
            For readability, we omit the superscript $\mathrm{fin}$ and write $\prediction_\preind$ instead of $\finpre_\preind$.
            The dashed vertical lines mark turning points of $\SCDF_\classifier$ and $\SCDF_\targetclassifier$.
            Note that between $\prediction_{\targetindex, 2}$ and $\prediction_{\targetindex, 3}$, $4$ type-$2$ constraints are binding, but only the two corresponding to $\prediction_3$ and $\prediction_6$ are essential.}
        \end{figure}
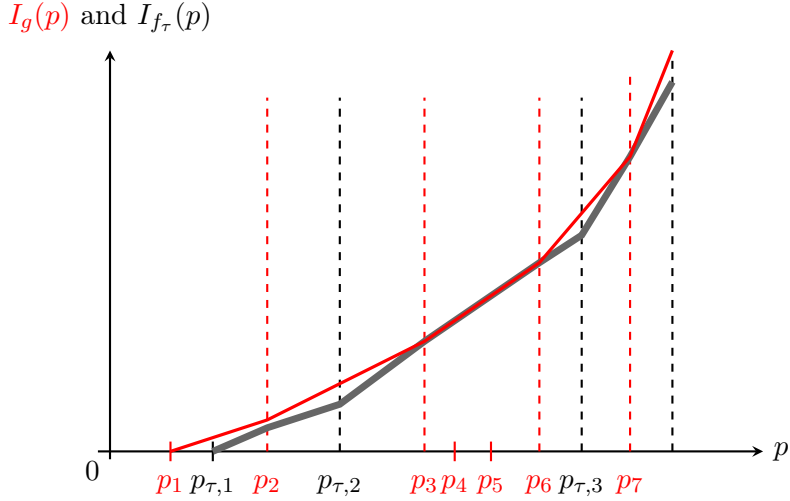

        We can prune binding constraints within each group, such that each group contributes at most $2$ binding constraints after pruning.
        To see why intuitively, refer to \Cref{fig:independent_constraints} (in particular the part between $\prediction_{\targetindex, 2}$ and $\prediction_{\targetindex, 3}$) for a graphical illustration.
        For each $\cellind$, restricted to $[\prediction_{\targetindex, \cellind}, \prediction_{\targetindex, \cellind + 1})$, because $\SCDF_\classifier$ is convex and $\SCDF_\targetclassifier$ is linear, only the first and the last points where $\SCDF_\classifier$ touches $\SCDF_\targetclassifier$ (corresponding to two binding constraints within the group) matter, and all binding constraints in between are implied by these $2$ constraints, together with some binding type-$3$ constraints to be discussed next.
        Within each group, we prune all binding type-$2$ constraints except the first and the last, as detailed below.

        To be more concrete, fix any $\cellind$ and suppose the type-$2$ constraints corresponding to $\finpre_{\preind_1}$ and $\finpre_{\preind_2}$ are binding, where $\prediction_{\targetindex, \cellind} \le \finpre_{\preind_1} < \finpre_{\preind_2} < \prediction_{\targetindex, \cellind + 1}$.
        Also recall that we are examining a vertex, which means all other constraints must hold too.
        So we have,
        \begin{align}
            \SCDF_\targetclassifier(\finpre_{\preind_1}) & = \sum\nolimits_{\preind' < \preind_1, \rayind \in [\numrays_{\preind'}]} (\finpre_{\preind_1} - \finpre_{\preind'}) \cdot \feasubsetprob(\ray_{\preind', \rayind}) \cdot \weight_{\preind', \rayind} \label{eq:binding_1} \\
            \SCDF_\targetclassifier(\finpre_{\preind_1 + 1}) & \le \sum\nolimits_{\preind' < \preind_1 + 1, \rayind \in [\numrays_{\preind'}]} (\finpre_{\preind_1 + 1} - \finpre_{\preind'}) \cdot \feasubsetprob(\ray_{\preind', \rayind}) \cdot \weight_{\preind', \rayind} \label{eq:binding_2} \\
            \SCDF_\targetclassifier(\finpre_{\preind_2}) & = \sum\nolimits_{\preind' < \preind_2, \rayind \in [\numrays_{\preind'}]} (\finpre_{\preind_2} - \finpre_{\preind'}) \cdot \feasubsetprob(\ray_{\preind', \rayind}) \cdot \weight_{\preind', \rayind} \label{eq:binding_3}
        \end{align}
        Taking the difference between Eqn.~\eqref{eq:binding_2} and Eqn.~\eqref{eq:binding_1}:
        \begin{align*}
            \SCDF_\targetclassifier(\finpre_{\preind_1 + 1}) - \SCDF_\targetclassifier(\finpre_{\preind_1}) & \le \CDF_\targetclassifier(\finpre_{\preind_1}) \cdot (\finpre_{\preind_1 + 1} - \finpre_{\preind_1}) \\
            & = \sum\nolimits_{\preind' \le \preind_1, \rayind \in [\numrays_{\preind'}]} (\finpre_{\preind_1 + 1} - \finpre_{\preind_1}) \cdot \feasubsetprob(\ray_{\preind', \rayind}) \cdot \weight_{\preind', \rayind},
        \end{align*}
        which implies
        \begin{align}
            \CDF_\targetclassifier(\finpre_{\preind_1}) \le \sum\nolimits_{\preind' \le \preind_1, \rayind \in [\numrays_{\preind'}]} \feasubsetprob(\ray_{\preind', \rayind}) \cdot \weight_{\preind', \rayind}. \label{eq:diff}
        \end{align}
        Taking the difference between Eqn.~\eqref{eq:binding_3} and Eqn.~\eqref{eq:binding_1}:
        \begin{align*}
            \SCDF_\targetclassifier(\finpre_{\preind_2}) - \SCDF_\targetclassifier(\finpre_{\preind_1}) & = \CDF_\targetclassifier(\finpre_{\preind_1}) \cdot (\finpre_{\preind_2} - \finpre_{\preind_1}) \\
            & = \sum\nolimits_{\preind' \le \preind_1, \rayind \in [\numrays_{\preind'}]} (\finpre_{\preind_2} - \finpre_{\preind_1}) \cdot \feasubsetprob(\ray_{\preind', \rayind}) \cdot \weight_{\preind', \rayind} \\
            & + \sum\nolimits_{\preind_1 < \preind' < \preind_2, \rayind \in [\numrays_{\preind'}]} (\finpre_{\preind_2} - \finpre_{\preind'}) \cdot \feasubsetprob(\ray_{\preind', \rayind}) \cdot \weight_{\preind', \rayind}.
        \end{align*}
        Plugging in Eqn.~\eqref{eq:diff}, the above becomes
        \begin{align*}
            \CDF_\targetclassifier(\finpre_{\preind_1}) \cdot (\finpre_{\preind_2} - \finpre_{\preind_1}) & \ge \CDF_\targetclassifier(\finpre_{\preind_1}) \cdot (\finpre_{\preind_2} - \finpre_{\preind_1}) \\
            & + \sum\nolimits_{\preind_1 < \preind' < \preind_2, \rayind \in [\numrays_{\preind'}]} (\finpre_{\preind_2} - \finpre_{\preind'}) \cdot \feasubsetprob(\ray_{\preind', \rayind}) \cdot \weight_{\preind', \rayind},
        \end{align*}
        which implies
        \begin{align*}
            \sum\nolimits_{\preind_1 \le \preind' < \preind_2, \rayind \in [\numrays_{\preind'}]} (\finpre_{\preind_2} - \finpre_{\preind'}) \cdot \feasubsetprob(\ray_{\preind', \rayind}) \cdot \weight_{\preind', \rayind} \le 0.
        \end{align*}
        Recall that $(\weight_{\preind, \rayind})$ are nonnegative, which means it must be the case that for each $\preind'$ where $\preind_1 \le \preind' < \preind_2$ and $\rayind \in [\numrays_{\preind'}]$, $\weight_{\preind', \rayind} = 0$ (none of these binding constraints will be pruned, as explained in the bullet point below).
        Moreover, equality must hold in Eqn.~\eqref{eq:diff} too.
        As a result, for any $\preind_3$ between $\preind_1 + 1$ and $\preind_2$,
        \begin{align*}
            \sum\nolimits_{\preind' < \preind_3, \rayind \in [\numrays_{\preind'}]} (\finpre_{\preind_3} - \finpre_{\preind'}) \cdot \feasubsetprob(\ray_{\preind', \rayind}) \cdot \weight_{\preind', \rayind} & = \sum\nolimits_{\preind' < \preind_1, \rayind \in [\numrays_{\preind'}]} (\finpre_{\preind_3} - \finpre_{\preind'}) \cdot \feasubsetprob(\ray_{\preind', \rayind}) \cdot \weight_{\preind', \rayind} \\
            & = \sum\nolimits_{\preind' < \preind_1, \rayind \in [\numrays_{\preind'}]} (\finpre_{\preind_1} - \finpre_{\preind'}) \cdot \feasubsetprob(\ray_{\preind', \rayind}) \cdot \weight_{\preind', \rayind} \\
            & + \sum\nolimits_{\preind' < \preind_1, \rayind \in [\numrays_{\preind'}]} (\finpre_{\preind_3} - \finpre_{\preind_1}) \cdot \feasubsetprob(\ray_{\preind', \rayind}) \cdot \weight_{\preind', \rayind} \\
            & = \SCDF_\targetclassifier(\finpre_{\preind_1}) + \CDF_\targetclassifier(\finpre_{\preind_1}) \cdot (\finpre_{\preind_3} - \finpre_{\preind_1}) \\
            & = \SCDF_\targetclassifier(\finpre_{\preind_3}).
        \end{align*}
        That is, the type-$2$ constraint corresponding to $\preind_3$ is binding --- and importantly, it is implied by the first and the last binding type-$2$ constraints within the group, as well as certain binding nonnegativity (type-$3$ to be discussed below) constraints.
        So we can safely prune all such binding constraints, leaving at most $2$ binding constraints within each group.
        In total, type-$2$ constraints can contribute at most $2 (\cellnum_\targetindex + 1) \le 2 (\totalcells + 1)$ binding constraints after pruning to the $\numrays$ needed.
        
        \item
        Finally, we have the $\numrays$ nonnegativity constraints $\weight_{\preind, \rayind} \ge 0$.
        We do not prune any of these constraints.
        Since we need $\numrays$ binding constraints after pruning, and the first $2$ types contribute at most $\feanum + 2(\totalcells + 1)$, we still need at least $\numrays - \feanum - 2(\totalcells + 1)$ type-$3$ constraints to be binding.
        In other words, only $\feanum + 2(\totalcells + 1)$ of $(\weight_{\preind, \rayind})$ can be strictly positive.
    \end{itemize}
    The argument above establishes that each vertex of the polytope involves at most $\suppsize$ strictly positive decision variables.
    To prove the third property, we simply pick the vertex that achieves the optimal objective value, which must exist because of the second property.
\end{proof}

\Cref{lem:exprog_properties}, together with the second property in \Cref{cor:lexopt_vs_dom} proved earlier, immediately implies the small-support property that we need:
\begin{proposition}
\label{cor:small_support}
    There are at most $\suppsize$ strictly positive levels in the optimal objective value of \ref{eq:abstract-program}.
\end{proposition}
\begin{proof}
    By \Cref{lem:exprog_properties}, the optimal objective value of \ref{eq:extreme-program} involves at most $\suppsize$ positive levels.
    This can be seen by examining the optimal solution identified in the third property:
    The expert corresponding to the optimal solution involves at most $\suppsize$ different extreme rays, leading to at most $\suppsize$ different prediction values.
    Each level in $\lexobj$ corresponds to one possible prediction, so at this optimal solution, at most $\suppsize$ levels can have positive values.
    However, all optimal solutions lead to the same optimal objective value, which means the optimal objective value of \ref{eq:extreme-program} involves at most $\suppsize$ positive levels.
    Now by the second property of \Cref{cor:lexopt_vs_dom}, \ref{eq:abstract-program} and \ref{eq:extreme-program} share the same optimal objective value, which again involves at most $\suppsize$ positive values.
\end{proof}

This is one of the main structural properties needed for our algorithm.

\xhdr{Numerical complexity of possible predictions}
Another property needed for our algorithm is that all possible predictions in $\predictions$ are numerically simple, meaning that each of them can be written as the ratio between two positive integers no larger than something exponential in some polynomial function of the input length.
This ensures, among other things, that our output expert has a polynomially sized representation.
To see why the property holds, let us assume that all input numbers are rationals whose numerators and denominators are between $1$ and $\maxint$ (so each input number takes $O(\log \maxint)$ bits).
Recall that each possible prediction is induced by some extreme ray, which can be written as a nonnegative linear combination of the input state-wise prediction-probability vectors.
Moreover, the coefficients $(\coeff_{\claind, \cellind})$ in this combination are a solution to a system with at most $\feanum$ linear equations (again, we can assume $\totalcells \le \feanum$ because we can focus on a span of $\obsspace$) defined by certain input numbers (together with $0$'s and $1$'s).
One can show that each of these coefficients involves only integers bounded by $(\feanum \maxint)^{\mathrm{poly}(\feanum)}$.
Now, each prediction is obtained by computing the weighted mean of at most $\feanum$ input predictions, i.e., $\left(\sum_{\claind, \cellind} \coeff_{\claind, \cellind} \cdot \prediction_{\claind, \cellind}\right) / \left(\sum_{\claind, \cellind} \coeff_{\claind, \cellind}\right)$, where all numbers involved (coefficients and input predictions) have reasonable numerical complexity.
As a result, each possible prediction in $\predictions$ is a rational whose numerator and denominator are no larger than $(\feanum \maxint)^{\mathrm{poly}(\feanum)}$.
This is also crucial for our efficient algorithm.
Formally:
\begin{lemma}
\label{lem:numerical_complexity}
    For each $\preind \in [\numpredictions]$, there exist positive integers $0 < \integernew_1 < \integernew_2 \le \premaxint$ such that $\finpre_\preind = \integernew_1 / \integernew_2$.
\end{lemma}
\begin{proof}
    First consider an arbitrary extreme ray $\ray \in \extremerays$.
    By definition, there exist coefficients $(\coeff_{\claind, \cellind})$ such that $\ray = \sum_{\claind, \cellind} \coeff_{\claind, \cellind} \cdot \wtedit{\reportvec^{\claind, \cellind}}$.
    These coefficients are the solution to a linear system defined by $\|\ray\|_1 = 1$ and at most $\feanum - 1$ linear equations each of the form $\ray_i = 0$ for some $i \in [\feanum]$.\footnote{
        Here, we take an arbitrary basis of the input $(\wtedit{\reportvec^{\claind, \cellind}})$ and set the coefficients corresponding to the other (dependent) input vectors to $0$.
    }
    Observe that all numbers involved in the system are input numbers, $0$, or $1$.
    By classical arguments in numerical linear algebra (see, e.g., \citep{schrijver1998theory}), the solution to such a system (i.e., the coefficients $(\coeff_{\claind, \cellind})$ can each be written as a fraction whose numerator and denominator are bounded by $(\feanum \maxint)^{\feanum(\feanum + 1)}$ (we do not try to optimize the precise bound here, or later).

    Now observe that the prediction $\finpre \in \predictions$ corresponding to the coefficients $(\coeff_{\claind, \cellind})$ can be written as
    \[
        \wtedit{
        \finpre = \frac{\sum\nolimits_{\claind, \cellind} \coeff_{\claind, \cellind} \cdot \prediction_{\claind, \cellind} \cdot \feasubsetprob(\reportvec^{\claind, \cellind})}{\sum\nolimits_{\claind, \cellind} \coeff_{\claind, \cellind} \cdot \feasubsetprob(\reportvec^{\claind, \cellind})}
        }
    \]
    where, without loss of generality, at most $\feanum$ coefficients can be strictly positive.
    All the numbers involved here are of reasonable numerical complexity.
    As a result, $\finpre$ can be written as the ratio between two integers bounded by $\premaxint$.
    Details are omitted in particular because we do not need the bound to be particularly tight.
\end{proof}

\subsection{Efficient Algorithm via Feasibility Search and Proof of \texorpdfstring{\Cref{thm:efficient}}{Theorem 4.1}}
We are now ready to proceed to the actual algorithm.

\xhdr{An inefficient algorithmic prototype}
We first present an inefficient prototype of our algorithm, which we will later turn into an efficient one based on the properties on the support size and the numerical complexity established earlier.
We will temporarily refrain from making formal claims, and focus on describing the overall approach and technical ideas.
The inefficient prototype involves solving a series of linear programs under a different parametrization from \ref{eq:extreme-program}.
The idea is to represent an expert $\classifier$ that dominates the target expert $\targetclassifier$ as a Blackwell spread of the latter.
That is, the predictions reported by $\classifier$ are fractionally packed into $\cellnum_\targetindex$ groups, each corresponding to one prediction reported by $\targetclassifier$, subject to the constraint that if we merge the predictions reported by $\classifier$ in a group, it becomes the corresponding prediction reported by $\targetclassifier$.
So we need $\cellnum_\targetindex \cdot \numpredictions$ variables, where the one indexed by $(\cellind, \preind)$, $\signalvec^{\cellind, \preind}$, corresponds to the fraction of the mass at $\prediction_{\targetindex, \cellind}$ under $\targetclassifier$ that is being spread into $\finpre_\preind \in \predictions$ under $\classifier$.
This is similar to our approach to the {\predicrefineoptlink} problem.
The essential difference lies in the objective, which is now the lexicographical objective that sequentially maximizes $\preprob_1, \dots, \preprob_{\numpredictions}$.
The complete program is as follows:
\begin{equation}
\tag{$\lexprog$}
\label{eq:lexicographical-program}
\begin{aligned}
    \max \quad
    &
    \left(\sum\nolimits_{\cellind\in[\cellnum_\targetindex]} \feasubsetprob(\signalvec^{\cellind, 1}), \ \sum\nolimits_{\cellind\in[\cellnum_\targetindex]} \feasubsetprob(\signalvec^{\cellind, 2}), \ \dots, \ \sum\nolimits_{\cellind\in[\cellnum_\targetindex]} \feasubsetprob(\signalvec^{\cellind, \numpredictions})\right) &&
    \\
    \text{s.t.}\quad
    &
    \signalvec^{\cellind,\preind}\in\obscone
    \quad
    && 
    \preind\in[\numpredictions],\ \cellind\in[\cellnum_\targetindex]~,
    \\
    &
    \sum\nolimits_{\preind\in[\numpredictions]}
    \sum\nolimits_{\cellind\in[\cellnum_\targetindex]}
    \signalvec^{\cellind,\preind}
    =
    \vone~,
    \\
    &
    \finpre_\preind \cdot 
    \feasubsetprob(\signalvec^{\cellind,\preind})
    =
    \obslabel(\signalvec^{\cellind,\preind})
    \quad
    &&
    \preind\in[\numpredictions],\ 
    \cellind \in [\cellnum_\targetindex]~,
    \\
    &
    \sum\nolimits_{\preind\in[\numpredictions]}
    \feasubsetprob(\signalvec^{\cellind,\preind})
    =
    \targetmass_\cellind
    \quad
    &&
    \cellind\in[\cellnum_\targetindex]~,
    \\
    &
    \sum\nolimits_{\preind\in[\numpredictions]}
    \obslabel(\signalvec^{\cellind,\preind})
    =
    \prediction_{\targetindex,\cellind}\targetmass_\cellind
    \quad
    &&
    \cellind\in[\cellnum_\targetindex]~.
\end{aligned}
\end{equation}

In the above, $\targetmass_\cellind = \feasubsetprob(\reportvec^{\targetindex, \cellind})$.
The constraint that $\signalvec^{\cellind, \preind}\in\obscone$ is implemented in the same way as in our approach to $\predicrefineopt$.
The functions $\feasubsetprob(\cdot)$ and $\obslabel(\cdot)$ are also implemented in the same way.
Note that the above program has a lexicographical objective with $\numpredictions$ levels, where each level is linear.
We can find an optimal solution by sequentially solving $\numpredictions$ linear programs, corresponding to the $\numpredictions$ levels of the objective.
We first maximize the level-$1$ objective subject to all the constraints of the original program.
Suppose the maximum level-$1$ objective value is $\optprob_1$.
Then, we enforce the additional linear constraint that $\sum\nolimits_{\cellind\in[\cellnum_\targetindex]} \feasubsetprob(\signalvec^{\cellind, 1}) = \optprob_1$, and maximize the level-$2$ objective.
More generally, after each iteration $\preind$, we add a new linear constraint, which together ensure that the first $\preind$ levels of the objective are lexicographically optimal.
After $\numpredictions$ iterations, we will arrive at an optimal solution of the original  program with the lexicographical objective.
Note again that this is not an efficient algorithm, since $\numpredictions$ can be exponential in $\feanum$ and / or $\totalcells$.

\xhdr{Overview of the optimized efficient algorithm}
Now we discuss how the two properties established earlier help optimize the prototype into an efficient algorithm, and provide an overview of the latter.
The key intuition is the following: Although the objective in \ref{eq:lexicographical-program} has $\numpredictions$ levels, which is generally super-polynomial, the small-support property ensures that only a polynomial number of these levels can be positive in an optimal solution.
All we need is to find these polynomially many levels efficiently.
In fact, we can do this one level at a time: Given (the indices of) the first $\lvlind$ positive levels, we want to find the $(\lvlind + 1)$-th positive level efficiently.
The natural idea is to binary search for the next positive level, i.e., the next prediction value to which the expert $\classifier$ being optimized can assign positive probability mass, subject to all the constraints from \ref{eq:lexicographical-program} and induced by the optimality of the first $\lvlind$ levels.
This is where we need the other property that bounds the numerical complexity of the possible predictions in $\predictions$.

Recall that the property says that each possible prediction in $\predictions$ is a rational number whose numerator and denominator are both bounded by $(\feanum \maxint)^{\mathrm{poly}(\feanum)}$.
Suppose for a moment that these possible predictions are integers rather than rationals, so we could naively binary search for the next prediction.
In that case, the number of iterations needed for the binary search would be $\mathrm{poly}(\feanum) \log(\feanum \maxint)$, which is polynomial.
However, our problem is a bit trickier, because we need to search over {\em rational} numbers.
One solution here is to run an accelerated search algorithm in the {\em Stern-Brocot tree} \citep{Stern-58,Brocot-62}.

\xhdr{A detour: (accelerated) search in the Stern-Brocot tree}
To be minimally self-contained, we provide a quick overview of the (somewhat folklore) Stern-Brocot tree.
The Stern-Brocot tree is an infinite binary search tree.
Each node of the tree corresponds to an upper bound and a lower bound, both defined by rational numbers.
For example, the root node corresponds to the lower bound of $0 / 1$ and the upper bound of $1 / 0$.
Each node defined by the lower bound of $a / b$ and the upper bound of $c / d$ has two children, defined by $(a / b, (a + c) / (b + d))$ and $((a + c) / (b + d), c / d)$ respectively.
When we (naively) search for a particular rational number $x$ in the tree, we start from the root, and iteratively compare $x$ with $(a + c) / (b + d)$ at each node defined by $a / b$ and $c / d$.
We go left if $x \le (a + c) / (b + d)$, and go right otherwise.
We stop as soon as $a + c$ or $b + d$ exceeds our upper bound of the numerical complexity.

The above procedure generally takes time linear in the largest integer involved, while we need something that runs in poly-logarithmic time.
Perhaps the simplest solution is to group all descents in the tree into consecutive left-descents and consecutive right-descents, and within each consecutive group, binary search for the number of consecutive descents.
One can show that this in fact runs in poly-logarithmic time in the largest integer involved.
This suffices for our purposes.
Formally, we will call the following algorithmic result in a blackbox way:
\begin{lemma}[see, e.g., Section 4.5 of \citealp{graham1994concrete}]
\label{lem:stern-brocot}
    There is an algorithm that finds an arbitrary rational number $\integernew_1 / \integernew_2$ where $\integernew_1$ and $\integernew_2$ are integers between $1$ and $L$ in time $O(\log^2 L)$.
    Moreover, the algorithm makes only comparison queries between $\integernew_1 / \integernew_2$ and another number $\bound$ specified by the algorithm, which returns $\mathtt{true}$ if $\integernew_1 / \integernew_2 \le \bound$, and $\mathtt{false}$ otherwise. 
\end{lemma}

\xhdr{The feasibility linear program}
Finally, it remains to efficiently check whether a candidate upper bound $\bound$ of the next reported prediction value is too small or large enough.
This can be done by solving a polynomially sized linear program.
Let $\classifier^*$ be any optimal solution to \ref{eq:lexicographical-program}, whose corresponding optimal objective value satisfies the small-support property of \Cref{cor:small_support}.
Suppose we have already found the first $\lvlind$ predictions reported by $\classifier^*$, $(\optpre_1, \dots, \optpre_\lvlind) = (\prediction_{\classifier^*, 1}, \dots, \prediction_{\classifier^*, \lvlind})$, together with their probabilities, $(\optprob_1, \dots, \optprob_\lvlind) = (\prob{\prediction_{\classifier^*} = \optpre_1}, \dots, \prob{\prediction_{\classifier^*} = \optpre_\lvlind})$.
All these should be enforced as hard constraints in all subsequent steps.
On top of that, we enforce $\prediction \le \bound$ as a linear constraint, where $\prediction$ is the $(\lvlind + 1)$-th prediction reported by $\classifier$ (see \ref{eq:feasibility-program} for how this is implemented).
We then maximize the probability that $\classifier$ reports the prediction value $\prediction$.
The complete linear program is as follows:

\begin{equation}
\tag{$\feasprog$}
\label{eq:feasibility-program}
\begin{aligned}
    \max \quad
    &
    \sum\nolimits_{\cellind\in[\cellnum_\targetindex]}\feasubsetprob(\signalvec^{\cellind, \lvlind + 1})
    &&
    \\
    \text{s.t.}\quad
    &
    \signalvec^{\cellind,\preind}\in\obscone
    \quad
    && 
    \preind\in[\suppsize],\ \cellind\in[\cellnum_\targetindex]~,
    \\
    &
    \sum\nolimits_{\preind\in[\suppsize]}
    \sum\nolimits_{\cellind\in[\cellnum_\targetindex]}
    \signalvec^{\cellind,\preind}
    =
    \vone~,
    \\
    &
    \bound \cdot \feasubsetprob(\signalvec^{\cellind, \lvlind + 1})
    \ge
    \obslabel(\signalvec^{\cellind, \lvlind + 1})
    &&
    \cellind \in [\cellnum_\targetindex]~,
    \\
    &
    \optpre_\preind \cdot 
    \feasubsetprob(\signalvec^{\cellind, \preind})
    =
    \obslabel(\signalvec^{\cellind, \preind})
    \quad
    &&
    \preind\in[\lvlind],\ 
    \cellind\in[\cellnum_\targetindex]~,
    \\
    &
    \sum\nolimits_{\cellind\in[\cellnum_\targetindex]}
    \feasubsetprob(\signalvec^{\cellind, \preind})
    =
    \optprob_\preind
    \quad
    &&
    \preind\in[\lvlind]~,
    \\
    &
    \sum\nolimits_{\preind\in[\suppsize]}
    \feasubsetprob(\signalvec^{\cellind, \preind})
    =
    \targetmass_\cellind
    \quad
    &&
    \cellind\in[\cellnum_\targetindex]~,
    \\
    &
    \sum\nolimits_{\preind\in[\suppsize]}
    \obslabel(\signalvec^{\cellind, \preind})
    =
    \prediction_{\targetindex,\cellind}\targetmass_\cellind
    \quad
    &&
    \cellind\in[\cellnum_\targetindex]~.
\end{aligned}
\end{equation}

Again, $\targetmass_\cellind = \feasubsetprob(\reportvec^{\targetindex, \cellind})$.
The constraint that $\signalvec^{\cellind, \preind}\in\obscone$ is implemented in the same way as in our approach to $\predicrefineopt$.
The functions $\feasubsetprob(\cdot)$ and $\obslabel(\cdot)$ are also implemented in the same way.
In addition to $\lvlind$ and $\bound$, \ref{eq:feasibility-program} is parametrized by $\optpre_1, \dots, \optpre_\lvlind$ and $\optprob_1, \dots, \optprob_\lvlind$, but we omit the dependence for brevity.

Now let us be formal again.
Observe that \ref{eq:feasibility-program} is a polynomially sized linear program.
For any (not necessarily feasible) solution $(\signalvec^{\cellind, \preind})$ to \ref{eq:feasibility-program}, let $\classifier$ (dependency on $(\signalvec^{\cellind, \preind})$ omitted for brevity) be the corresponding expert, where each $\signalvec^{\cellind, \preind}$ by default constitutes its own prediction atom, with different vectors leading to the same prediction value merged.\footnote{
    This is an important conceptual subtlety, since if we pool $\signalvec^{\cellind, \preind}$ with the same $\preind$ together to form a single prediction, the resulting expert may not be a Blackwell spread of $\targetclassifier$.
}
Suppose the predictions reported by $\classifier$ are $\prediction_{\classifier, 1}, \dots, \prediction_{\classifier, \cellnum_\classifier}$ where $\prediction_{\classifier, \cellind} < \prediction_{\classifier, \cellind + 1}$ for each $\cellind \in [\cellnum_\classifier - 1]$.
We need the following property of \ref{eq:feasibility-program}:
\begin{lemma}
\label{lem:feasprog}
    Fix any $\lvlind$ and $\bound$.
    \begin{itemize}
        \item If a solution $(\signalvec^{\cellind, \preind})$ to \ref{eq:feasibility-program} is feasible, and the corresponding objective value is $\objval$, then the corresponding expert $\classifier$ satisfies (1) $\classifier \blackwelldomeq \targetclassifier$, (2) for each $\preind \in [\lvlind]$, $\prob{\prediction_\classifier = \optpre_\preind} \ge \optprob_\preind$, and (3) $\prob{\prediction_\classifier \le \bound} \ge \objval + \sum_{\preind \in [\lvlind]} \optprob_\preind$.
        \item Conversely, if there exists an expert satisfying the $3$ conditions above for some $\objval \ge 0$, then the optimal objective value of \ref{eq:feasibility-program} is at least $\objval$.
    \end{itemize}
\end{lemma}
\begin{proof}
    Consider the first bullet point.
    Suppose $(\signalvec^{\cellind, \preind})$ is feasible.
    Then one can obtain $\targetclassifier$ by contracting the predictions reported by $\classifier$, i.e., for each $\cellind \in [\cellnum_\targetindex]$, 
    \begin{align*}
        & \reportvec^{\targetindex, \cellind} = \sum_{\preind \in [\suppsize]} \signalvec^{\cellind, \preind} \\
        \implies\ & \sum\nolimits_{\preind \in [\suppsize]} \feasubsetprob(\signalvec^{\cellind, \preind}) = \targetmass_\cellind \quad \text{and} \quad \sum\nolimits_{\preind \in [\suppsize]} \obslabel(\signalvec^{\cellind, \preind}) = \targetmass_\cellind \cdot \prediction_{\targetindex, \cellind}.
    \end{align*}
    This implies the first property: $\classifier \blackwelldomeq \targetclassifier$.
    For each $\preind \in [\lvlind]$,
    \[
        \prob{\prediction_\classifier = \optpre_\preind} \ge \sum\nolimits_{\cellind \in [\cellnum_{\targetindex}]} \feasubsetprob(\signalvec^{\cellind, \preind}) = \optprob_\preind.
    \]
    This is the second property.
    Now for the third property, observe that for each $\cellind \in [\cellnum_\targetindex]$, if $\signalvec^{\cellind, \preind} \ne \vzero$, then $\feasubsetmean(\signalvec^{\cellind, \preind}) \le \bound$.
    So, 
    \begin{align*}
        \prob{\prediction_\classifier \le \bound} & \ge \sum\nolimits_{\preind \in [\lvlind + 1], \cellind \in [\cellnum_\targetindex]} \feasubsetprob(\signalvec^{\cellind, \preind}) \\
        & = \sum\nolimits_{\cellind \in [\cellnum_\targetindex]} \feasubsetprob(\signalvec^{\cellind, \lvlind + 1}) + \sum\nolimits_{\preind \in [\lvlind]} \optprob_\preind \\
        & = \objval + \sum\nolimits_{\preind \in [\lvlind]} \optprob_\preind~.
    \end{align*}

    Conversely, suppose $\classifier$ satisfies all $3$ properties for some $\objval$.
    Then, since $\classifier \blackwelldomeq \targetclassifier$, one can construct a solution $(\signalvec^{\cellind, \preind})$ from $\classifier$ that satisfies all constraints other than the ones involving $\bound$ and $\optprob_\preind$.
    Since $\classifier$ satisfies the second and the third property, one can ``permute'' $(\signalvec^{\cellind, \preind})$ over $\preind$ such that the remaining two constraints are also satisfied, and the objective value is at least $\objval$.
    We omit the concrete construction of $(\signalvec^{\cellind, \preind})$ because it is notationally messy and carries little insight.
\end{proof}

In words, our candidate upper bound $\bound$ is feasible iff we end up with a positive objective value when we solve \ref{eq:feasibility-program}.
Now we run the accelerated search algorithm in the Stern-Brocot tree on the candidate upper bound $\bound$ to find the exact value of the next prediction $\optpre_{\lvlind + 1}$.
In the process, we also obtain $\optprob_{\lvlind + 1}$ by solving \ref{eq:feasibility-program} with $\bound$ set to $\optpre_{\lvlind + 1}$.

\begin{algorithm}[t!]
\caption{\algsearch}
\label{alg:search}
\DontPrintSemicolon
\KwIn{Prior distribution $\feaprob$, input experts $(\givenclassifier_j)_{j\in[\classifiernum]}$, target index $\targetindex\in[\classifiernum]$, \wtedit{$\maxint$ defined as the input rational-size bound.}}
\KwOut{A constructible expert $\classifier \in \strictdomfam$ such that $\classifier \blackwelldomeq \targetclassifier$.}

Construct the component vectors $(\reportvec^{j,a})_{(j,a)\in\compindex}$, the observable space $\obsspace$, the cone $\obscone$, and the observable label functional $\obslabel$.\;

Let $\lvlind \gets 0$.

\While{$\lvlind < \suppsize$ and $\sum_{\preind \in [\lvlind]} \optprob_\preind < 1$}{
    Run accelerated search over $\bound$ in the Stern-Brocot tree (as specified in \Cref{lem:stern-brocot}), 
    with the bound on numerical complexity set to $\premaxint$, 
    and the comparison function given by: $\bound$ is large enough iff solving \ref{eq:feasibility-program} gives a positive objective value.

    Let $\wtedit{\optpre_{\lvlind+1} \gets \bound}$, where $\bound$ is the output of the search procedure above.
    Solve \ref{eq:feasibility-program} once again and let $\optprob_{\lvlind+1}$ be the optimal objective value.

    Let $\lvlind \gets \lvlind + 1$.
}

Solve \ref{eq:feasibility-program} once again with $\bound = 1$ (or any number), and let
$(\signalvec^{\cellind, \preind})_{\cellind, \preind}$ be an optimal solution.\;

For every pair $(\cellind, \preind)$ with
$\feasubsetprob(\signalvec^{\cellind, \preind})>0$, set
$\widehat{\prediction}_{\cellind, \preind}
\leftarrow
\obslabel(\signalvec^{\cellind, \preind})/
\feasubsetprob(\signalvec^{\cellind, \preind})$.\;

Output the randomized expert $\classifier$ defined by
$\classifier(\cdot\mid \fea_i)
=
\sum_{\cellind, \preind:\ \feasubsetprob(\signalvec^{\cellind, \preind})>0}
\signalvec^{\cellind, \preind}_i
\delta_{(\widehat{\prediction}_{\cellind, \preind})}(\cdot)$
for every $i\in[\feanum]$.\;
\end{algorithm}

\xhdr{Putting everything together}
Now we are ready to analyze the complete algorithm (\Cref{alg:search}).
\begin{proof}[Proof of \Cref{thm:efficient}]
    First consider the time complexity of \Cref{alg:search}.
    Overall, we need $O(\feanum + \totalcells)$ outermost iterations to find the $O(\feanum + \totalcells)$ predictions in the support of $\classifier^*$; in each of these iterations, we run an accelerated search, which terminates in $O(\mathrm{poly}(\feanum, \log(\feanum \maxint))$ iterations by \Cref{lem:stern-brocot}; within each iteration of the search, we solve a polynomially sized linear program to determine feasibility of the current candidate upper bound.
    The entire algorithm runs in polynomial time.

    By \Cref{cor:lexopt_vs_dom}, we only need to solve the lexicographical optimization problem whose abstract form is \ref{eq:abstract-program}.
    By \Cref{cor:small_support}, the optimal objective value involves at most $\suppsize$ positive levels.
    We inductively argue that before the $\lvlind+1$-th iteration (or after the $\lvlind $-th iteration), the algorithm has correctly found $\optpre_1, \dots, \optpre_{\lvlind}$ and $\optprob_1, \dots, \optprob_{\lvlind}$.
    Consider the $\lvlind$-th iteration for some $\lvlind+1 \ge 1$, and focus on the execution of the previous iteration.
    
    We first argue that the search procedure over $\bound$ correctly finds $\optpre_{\lvlind+1}$.
    Given \Cref{lem:stern-brocot}, the core of the correctness lies in whether the comparison between $\bound$ and $\optpre_{\lvlind+1}$ is done properly.
    Fixing some $\bound$, if \ref{eq:feasibility-program} 
    returns a positive objective value $\objval > 0$, then by \Cref{lem:feasprog},
    there exists an expert $\classifier$ satisfying the $3$ conditions.
    Since $\optpre_1, \dots, \optpre_{\lvlind}$ and $\optprob_1, \dots, \optprob_{\lvlind}$ are a prefix of the optimal objective value, it must be the case that for each $\preind \in [\lvlind]$, $\prob{\prediction_\classifier = \optpre_\preind} = \optprob_\preind$, so $\prob{\optpre_{\lvlind} < \prediction_\classifier \le \bound} = \objval > 0$.
    In other words, $\optpre_{\lvlind+1} \le \bound$.

    On the other hand, suppose \ref{eq:feasibility-program} returns an objective value of $0$.
    Assume towards a contradiction that $\optpre_{\lvlind+1} \le \wtedit{\bound}$.
    Then by \Cref{lem:feasprog}, based on any optimal expert $\classifier^*$ that maximizes $\lexobj$, we could construct a feasible solution to \ref{eq:feasibility-program} whose induced objective value is at least $\prob{\prediction_{\classifier^*} = \optpre_{\lvlind+1}} > 0$, a contradiction.

    It remains to argue that the additional solve of \ref{eq:feasibility-program} 
    with $\bound \gets \optpre_{\lvlind+1}$
    at the end of each iteration correctly gives $\optprob_{\lvlind+1}$, which follows from essentially the same logic.
    We omit the argument to avoid repetition.
\end{proof}

\wtedit{\xhdr{Dominating multiple target experts}}
\label{rmk:multiple-dominance}
\Cref{thm:efficient} extends naturally to a multi-target variant in which the output expert is required to weakly Blackwell dominate several target experts simultaneously. Fix a target set $\targetexpertset\subseteq[\classifiernum]$ and replace the single dominance constraint $\classifier\blackwelldomeq \givenclassifier_{\targetindex}$ by
\begin{align*}
    \classifier \blackwelldomeq \givenclassifier_{\targetindex}
    \quad
    \text{for every } \targetindex\in\targetexpertset~.
\end{align*}
If no constructible expert satisfies these constraints, then the multi-target variant of {\predicrefinelink} reports infeasibility. Otherwise, the same lexicographic optimization argument applies. Indeed, if a feasible expert admits a strict Blackwell improvement within $\classifierfam$, then this strict improvement remains feasible by transitivity of $\blackwelldomeq$. Hence, any lexicographically optimal feasible constructible expert is again undominated within the full constructible class $\classifierfam$.

At the level of the finite-prediction representation used in the proof of \Cref{thm:efficient}, this extension only adds linear integrated-CDF constraints. Specifically, for every $\targetindex\in\targetexpertset$ and every finite prediction value $\finpre_{\preind}\in\predictions$, we add
\begin{align*}
    \SCDF_{\classifier}(\finpre_{\preind})
    \ge
    \SCDF_{\givenclassifier_{\targetindex}}(\finpre_{\preind})~.
\end{align*}
Equivalently, in the extreme-ray program \ref{eq:extreme-program}, these constraints take the linear form
\begin{align*}
    \sum\nolimits_{\preind'<\preind}
    \sum\nolimits_{\rayind\in[R_{\preind'}]}
    \bigl(\finpre_{\preind}-\finpre_{\preind'}\bigr)
    \feaprob(\ray_{\preind',\rayind})
    \weight_{\preind',\rayind}
    \ge
    \SCDF_{\givenclassifier_{\targetindex}}(\finpre_{\preind}),
    \quad
    \targetindex\in\targetexpertset,\ \preind\in[\numpredictions]~.
\end{align*}
The vertex-counting argument is unchanged except that the Blackwell constraints are now grouped separately for each target expert. Thus, the support bound becomes
$\feanum
+
2\sum\nolimits_{\targetindex\in\targetexpertset}
(\cellnum_{\targetindex}+1)$
where $\cellnum_{\targetindex}=|\supp(\CDF_{\givenclassifier_{\targetindex}})|$. Consequently, for any fixed target set $\targetexpertset$, the proof of \Cref{thm:efficient} goes through with the corresponding modified support bound.

The feasibility programs used by the algorithm can be modified in the same way. 
To ensure simultaneous dominance, we replace the single source label $\cellind\in[\cellnum_{\targetindex}]$ by a joint source label
\begin{align*}
    \mathbf{\cellind}
    =
    (\mathbf{\cellind}_{\targetindex})_{\targetindex\in\targetexpertset}
    \in
    \prod\nolimits_{\targetindex\in\targetexpertset}[\cellnum_{\targetindex}]~,
\end{align*}
and introduce atoms $\signalvec^{\mathbf{\cellind},\preind}\in\obscone$. For each target expert $\givenclassifier_{\targetindex}$ and each source component $\cellind\in[\cellnum_{\targetindex}]$, we impose the linear marginal constraints
\begin{align*}
    \sum\nolimits_{\mathbf{\cellind}:\,\mathbf{\cellind}_{\targetindex}=\cellind}
    \sum\nolimits_{\preind}
    \feaprob(\signalvec^{\mathbf{\cellind},\preind})
    =
    \feaprob(\reportvec^{\targetindex,\cellind})~,
    \quad
    \sum\nolimits_{\mathbf{\cellind}:\,\mathbf{\cellind}_{\targetindex}=\cellind}
    \sum\nolimits_{\preind}
    \obslabel(\signalvec^{\mathbf{\cellind},\preind})
    =
    \prediction_{\targetindex,\cellind}
    \feaprob(\reportvec^{\targetindex,\cellind})~.
\end{align*}
These constraints provide, for every $\targetindex\in\targetexpertset$, a martingale coupling from the prediction distribution of $\givenclassifier_{\targetindex}$ to the output prediction distribution. 
Thus, the output expert weakly Blackwell dominates every target expert in $\targetexpertset$. When $|\targetexpertset|$ is fixed, the number of joint source labels is polynomial in the input size, and the algorithm remains polynomial-time.

The additive FPTAS for {\predicrefineoptlink} admits the same fixed-target-set extension: one uses the same joint source labels and adds the corresponding source-mass and source-posterior-mean constraints for every $\targetindex\in\targetexpertset$. Finally, the deterministic-output hardness results continue to hold for the multi-target variant, since the original single-target problem is the special case $|\targetexpertset|=1$.

\section{An FPTAS to \texorpdfstring{\predicrefineoptlink}{OPT-Refine}}
\label{subsec:fptas}

\newcommand{\scalar}{z}

In this section, we give an additive FPTAS for the problem {\predicrefineoptlink}.
Recall that in {\predicrefineoptlink}, the input consists of the prior distribution $\feaprob$, $\classifiernum$ finite-support input experts $(\givenclassifier_j)_{j\in[\classifiernum]}$, a target index $\targetindex\in[\classifiernum]$, and a proper loss $\loss$.
Each input expert may be randomized, and each reported prediction generated by an input expert is assumed to be a Bayesian posterior mean.
Our goal is to output a constructible expert that weakly Blackwell dominates the target expert $\givenclassifier_\targetindex$ while approximately minimizing the expected proper loss.

For the target expert $\givenclassifier_\targetindex$, define the optimal achievable loss by
\begin{align*}
    \lossopt_\targetindex
    \triangleq
    \inf
    \left\{
    \expect[\prediction\sim \CDF_{\classifier}]{\loss(\predrv)}:
    \classifier\in\classifierfam
    \text{ and }
    \classifier\blackwelldomeq\givenclassifier_\targetindex
    \right\}~.
\end{align*}
Our main result is the following additive FPTAS.

\begin{theorem}[Additive FPTAS for {\predicrefineoptlink}]
\label{thm:proper-loss-fptas-general-k}
Given any regular proper loss function $\loss$ in the sense of \Cref{def:regular-proper-loss}, for every target index $\targetindex\in[\classifiernum]$ and every $\eps>0$, \Cref{alg:proper-loss-fptas}, which runs in time polynomial in the input size and $1/\eps$, outputs an expert $\classifier_\eps \in\classifierfam$ such that
$\classifier_\eps\blackwelldomeq\givenclassifier_\targetindex$ and 
\begin{align*}
    \expect[\prediction\sim \CDF_{\classifier_\eps}]{\loss(\predrv)}
    \le
    \lossopt_\targetindex+\eps~.
\end{align*}
\end{theorem}

\subsection{Algorithm Idea}

The main computational difficulty is that the true loss contribution of an output atom $\signalvec\in\obscone$ is
\begin{align*}
    \feasubsetprob(\signalvec)
    \bayesrisk\left(
    \frac{\obslabel(\signalvec)}{\feasubsetprob(\signalvec)}
    \right),
\end{align*}
which is generally nonlinear in $\signalvec$.
Here, $\obslabel(\signalvec)/\feasubsetprob(\signalvec)$ is the prediction attached to the atom $\signalvec$.
The algorithm overcomes this difficulty by replacing the concave Bayes-risk function $\bayesrisk$ with a polynomial-size piecewise-linear upper approximation.

\begin{definition}[Regular proper loss]
\label{def:regular-proper-loss}
A proper loss $\loss$ is called regular if, for every $\eps>0$, 
one can compute a finite collection of affine functions
$q\mapsto \affslope_\affindex q+\affintercept_\affindex$ indexed by $\affindex\in[\affnum_\eps]$ where $\affnum_\eps\in\integers_+$ such that
\begin{align*}
    0
    \le
    \min\nolimits_{\affindex\in[\affnum_\eps]}
    \{\affslope_\affindex\bayesprob+\affintercept_\affindex\}
    -
    \bayesrisk(\bayesprob)
    \le
    \eps
    \quad
    \text{for every }\bayesprob\in[0,1]~.
\end{align*}
Moreover, $\affnum_\eps$ and the bit complexity of the coefficients
$(\affslope_\affindex,\affintercept_\affindex)_{\affindex\in[\affnum_\eps]}$ are polynomial in $1/\eps$.
\end{definition}

We note that the regularity condition is a computational condition on the proper loss.
It ensures that the concave Bayes-risk function can be approximated uniformly from above by a polynomial-size minimum of affine functions.
This allows us to replace the nonlinear loss objective by a polynomial-size linear program.
This condition is satisfied by many common proper losses, including Brier loss, log loss, etc.
\begin{corollary}[Smooth proper losses are regular]
\label{cor:smooth-proper-loss-fptas}
Suppose the Bayes-risk function $\bayesrisk$ is twice continuously differentiable and concave on $[0,1]$, and suppose there exists $B\ge0$ such that
\begin{align*}
    -B
    \le
    \bayesrisk''(\bayesprob)
    \le
    0
    \quad
    \text{for every }\bayesprob\in[0,1]~.
\end{align*}
Then $\loss$ is regular with a uniform grid of size $O(\sqrt{B/\eps}+1)$.
\end{corollary}

Given a regular proper loss, for every fixed affine function $q\mapsto \affslope_\affindex q+\affintercept_\affindex$, the loss contribution of atom $\signalvec$ admits the following linear upper bound
\begin{align*}
    \feasubsetprob(\signalvec)
    \bayesrisk\left(
    \frac{\obslabel(\signalvec)}{\feasubsetprob(\signalvec)}
    \right)
    \le
    \affslope_\affindex\obslabel(\signalvec)
    +
    \affintercept_\affindex\feasubsetprob(\signalvec)~.
\end{align*}
Therefore, after replacing $\bayesrisk$ by a finite upper envelope of affine functions, the loss objective can be minimized by a linear program.
For every report $\prediction\in[0,1]$, the function $\bayesprob\mapsto\exploss(\prediction,\bayesprob)$ is affine in $\bayesprob$, and properness implies
\begin{align*}
    \bayesrisk(\bayesprob)
    \le
    \exploss(\prediction,\bayesprob)
    \quad
    \text{for every }\prediction,\bayesprob\in[0,1],
\end{align*}
with equality when $\prediction=\bayesprob$.
Thus the affine functions induced by reports are canonical affine upper bounds on the concave Bayes-risk function.
For smooth proper losses, a finite grid of such reports yields the regularity condition; see \Cref{cor:smooth-proper-loss-fptas}.

The Blackwell constraint is enforced by source-labelled atoms.
For each source prediction component $\targetcompidx\in[\cellnum_\targetindex]$, define
$\targetmass_\targetcompidx\triangleq\feasubsetprob(\reportvec^{\targetindex,\targetcompidx})$.
For every source prediction component $\targetcompidx\in[\cellnum_\targetindex]$ and every affine piece $\affindex\in[\affnum_\eps]$, the LP creates an atom $\signalvec^{\targetcompidx,\affindex}\in\obscone$.
The source label $\targetcompidx$ records the target prediction value from which this output atom is coupled.
The LP imposes the source-mass and source-posterior-mean constraints
\begin{align*}
    \sum\nolimits_{\affindex\in[\affnum_\eps]}
    \feasubsetprob(\signalvec^{\targetcompidx,\affindex})
    =
    \targetmass_\targetcompidx,
    \quad
    \sum\nolimits_{\affindex\in[\affnum_\eps]}
    \obslabel(\signalvec^{\targetcompidx,\affindex})
    =
    \prediction_{\targetindex,\targetcompidx}\targetmass_\targetcompidx~.
\end{align*}
These constraints define a martingale coupling from the target prediction distribution to the output prediction distribution.
Therefore the output expert weakly Blackwell dominates the target expert $\givenclassifier_\targetindex$.
Importantly, the affine-piece index $\affindex$ is not the prediction reported by the final expert.
It only indexes an affine upper bound used in the LP objective.
The actual prediction attached to atom $\signalvec^{\targetcompidx,\affindex}$ is
$\obslabel(\signalvec^{\targetcompidx,\affindex})/\feasubsetprob(\signalvec^{\targetcompidx,\affindex})$.

\xhdr{The finite-size linear program}
For every source component $\targetcompidx\in[\cellnum_\targetindex]$ and every affine-piece index $\affindex\in[\affnum_\eps]$, introduce a vector variable $\signalvec^{\targetcompidx,\affindex}\in\R^{\feanum}$.
The linear program \ref{eq:discretized-LP} is described as follows:
\begin{equation}
\tag{$\lpprog_\eps$}
\label{eq:discretized-LP}
\begin{aligned}
    \min \quad
    &
    \sum\nolimits_{\targetcompidx\in[\cellnum_\targetindex]}
    \sum\nolimits_{\affindex\in[\affnum_\eps]}
    \left(
    \affslope_\affindex\obslabel(\signalvec^{\targetcompidx,\affindex})
    +
    \affintercept_\affindex\feasubsetprob(\signalvec^{\targetcompidx,\affindex})
    \right) &&
    \\
    \text{s.t.}\quad
    &
    \signalvec^{\targetcompidx,\affindex}\in\obscone
    \quad
    && 
    \targetcompidx\in[\cellnum_\targetindex],\ \affindex\in[\affnum_\eps]~,
    \\
    &
    \sum\nolimits_{\targetcompidx\in[\cellnum_\targetindex]}
    \sum\nolimits_{\affindex\in[\affnum_\eps]}
    \signalvec^{\targetcompidx,\affindex}
    =
    \vone~,
    \\
    &
    \sum\nolimits_{\affindex\in[\affnum_\eps]}
    \feasubsetprob(\signalvec^{\targetcompidx,\affindex})
    =
    \targetmass_\targetcompidx
    \quad
    &&
    \targetcompidx\in[\cellnum_\targetindex]~,
    \\
    &
    \sum\nolimits_{\affindex\in[\affnum_\eps]}
    \obslabel(\signalvec^{\targetcompidx,\affindex})
    =
    \prediction_{\targetindex,\targetcompidx}\targetmass_\targetcompidx
    \quad
    &&
    \targetcompidx\in[\cellnum_\targetindex]~.
\end{aligned}
\end{equation}
Program \ref{eq:discretized-LP} is a finite-dimensional linear program.
Indeed, the constraint $\signalvec^{\targetcompidx,\affindex}\in\obscone$ can be written explicitly by introducing coefficients
$\theta^{\targetcompidx,\affindex}_{j,a'}\in\R$ such that
\begin{align*}
    \signalvec^{\targetcompidx,\affindex}
    =
    \sum\nolimits_{(j,a')\in\compindex}
    \theta^{\targetcompidx,\affindex}_{j,a'}\reportvec^{j,a'},
    \quad
    \signalvec^{\targetcompidx,\affindex}_i\ge0
    \quad
    \text{for every }i\in[\feanum]~.
\end{align*}
Under this representation,
\begin{align*}
    \obslabel(\signalvec^{\targetcompidx,\affindex})
    =
    \sum\nolimits_{(j,a')\in\compindex}
    \theta^{\targetcompidx,\affindex}_{j,a'}
    \prediction_{j,a'}\feasubsetprob(\reportvec^{j,a'})~,
\end{align*}
which is linear in the variables.
The function $\feasubsetprob(\signalvec^{\targetcompidx,\affindex})$ is also linear in $\signalvec^{\targetcompidx,\affindex}$.

\begin{algorithm}[H]
\caption{\algfptas}
\label{alg:proper-loss-fptas}
\DontPrintSemicolon
\KwIn{Prior distribution $\feaprob$, input experts $(\givenclassifier_j)_{j\in[\classifiernum]}$, target index $\targetindex\in[\classifiernum]$, regular proper loss $\loss$, accuracy $\eps>0$.}
\KwOut{A constructible expert $\classifier_\eps$ with $\classifier_\eps\blackwelldomeq\givenclassifier_\targetindex$.}

Construct the component vectors $(\reportvec^{j,a})_{(j,a)\in\compindex}$, the observable space $\obsspace$, the cone $\obscone$, and the observable label functional $\obslabel$.\;

Use the regularity oracle for $\loss$ to compute affine upper bounds
$q\mapsto\affslope_\affindex q+\affintercept_\affindex$ for $\affindex\in[\affnum_\eps]$.\;

Solve the linear program \ref{eq:discretized-LP} and let
$(\signalvec^{\targetcompidx,\affindex})_{\targetcompidx,\affindex}$ be an optimal solution.\;

For every pair $(\targetcompidx,\affindex)$ with
$\feasubsetprob(\signalvec^{\targetcompidx,\affindex})>0$, set
$\widehat{\prediction}_{\targetcompidx,\affindex}
\leftarrow
\obslabel(\signalvec^{\targetcompidx,\affindex})/
\feasubsetprob(\signalvec^{\targetcompidx,\affindex})$.\;

Output the randomized expert $\classifier_\eps$ defined by
$\classifier_\eps(\cdot\mid \fea_i)
=
\sum_{\targetcompidx,\affindex:\ \feasubsetprob(\signalvec^{\targetcompidx,\affindex})>0}
\signalvec^{\targetcompidx,\affindex}_i
\delta_{(\widehat{\prediction}_{\targetcompidx,\affindex})}(\cdot)$
for every $i\in[\feanum]$.\;
\end{algorithm}

We note that the affine-piece index $\affindex$ in the program \ref{eq:discretized-LP} is not the prediction value reported by the final expert.
It only indexes an affine upper bound used in the loss objective.
After the LP is solved, the actual prediction attached to atom $\signalvec^{\targetcompidx,\affindex}$ is
$\obslabel(\signalvec^{\targetcompidx,\affindex})/\feasubsetprob(\signalvec^{\targetcompidx,\affindex})$.
Therefore \Cref{alg:proper-loss-fptas} returns an expert that is exactly constructible, whose reported predictions are exactly Bayesian posterior means, and that exactly weakly Blackwell-dominates the target expert.
In other words, only the expected proper-loss value is approximated in program \ref{eq:discretized-LP}.

\subsection{Proof of \Cref{thm:proper-loss-fptas-general-k}}

We now prove \Cref{thm:proper-loss-fptas-general-k}.
The proof is organized around three ingredients:
a martingale characterization of Blackwell dominance, a source-labelled decomposition lemma for experts, and two LP lemmas showing soundness and approximation.

\begin{lemma}
[Martingale characterization of Blackwell dominance]
\label{lem:finite-martingale-blackwell}
Let
$\nu
=
\sum\nolimits_{a\in[\cellnum]}\predicmass_a\delta_{(\prediction_a)},
\nu\primed
=
\sum\nolimits_{b\in[\cellnum\primed]}\predicmass\primed_b\delta_{(\prediction\primed_b)}$
be two finitely supported probability distributions on $[0,1]$ with the same mean.
Then the following are equivalent:
\begin{enumerate}[label=(\roman*)]
    \item
    $\nu\primed$ Blackwell dominates $\nu$, namely, for every $t\in[0,1]$,
    \begin{align*}
        \sum\nolimits_{b\in[\cellnum\primed]}\predicmass\primed_b(t-\prediction\primed_b)_+
        \ge
        \sum\nolimits_{a\in[\cellnum]}\predicmass_a(t-\prediction_a)_+~.
    \end{align*}

    \item
    There exists a nonnegative matrix $(\pi_{a,b})_{a\in[\cellnum],b\in[\cellnum\primed]}$ such that
    \begin{align*}
        \sum\nolimits_{b\in[\cellnum\primed]}\pi_{a,b}
        =
        \predicmass_a,\;
        a\in[\cellnum]~;
        \quad
        \sum\nolimits_{a\in[\cellnum]}\pi_{a,b}
        =
        \predicmass\primed_b,\;
        b\in[\cellnum\primed]~;
        \quad \sum\nolimits_{b\in[\cellnum\primed]}\pi_{a,b}\prediction\primed_b
        =
        \predicmass_a\prediction_a,\;
        a\in[\cellnum]~.
    \end{align*}
\end{enumerate}
\end{lemma}

\begin{proof}
We first recall why the integrated-CDF inequalities in part (i), together with equality of means, are equivalent to convex order.
That is, they are equivalent to the condition that for every continuous convex function $\chi:[0,1]\to\R$,
\begin{align*}
    \sum\nolimits_{b\in[\cellnum\primed]}\predicmass\primed_b\chi(\prediction\primed_b)
    \ge
    \sum\nolimits_{a\in[\cellnum]}\predicmass_a\chi(\prediction_a)~.
\end{align*}
Indeed, every convex piecewise-linear function on $[0,1]$ can be written as an affine function plus a nonnegative linear combination of hinge functions of the form $x\mapsto(t-x)_+$.
The affine part has the same expectation under $\nu$ and $\nu\primed$ because the two measures have the same mean.
Thus the integrated-CDF inequalities imply the desired inequality for every convex piecewise-linear $\chi$.
Every continuous convex function on $[0,1]$ can be uniformly approximated by convex piecewise-linear functions, so the result extends to all continuous convex $\chi$.
The converse is immediate because $x\mapsto(t-x)_+$ is convex.

We now prove that convex order implies the existence of the matrix $\pi$.
Consider the feasibility system in variables $\pi_{a,b}\ge0$ given by the three groups of linear equations in part (ii).
By Farkas' lemma, if this system is infeasible, then there exist scalars $\scalar_a$, $\scalar_b$, and $\scalar\primed_a$ such that
\begin{align*}
    \scalar_a+\scalar_b+\scalar\primed_a\prediction\primed_b
    & 
    \ge
    0
    \quad
    a\in[\cellnum], b\in[\cellnum\primed]~,\\
    \sum\nolimits_{a\in[\cellnum]}\predicmass_a \scalar_a
    +
    \sum\nolimits_{b\in[\cellnum\primed]}\predicmass\primed_b \scalar_b
    +
    \sum\nolimits_{a\in[\cellnum]}\predicmass_a\prediction_a \scalar\primed_a
    & 
    <
    0~.
\end{align*}
We define
\begin{align*}
    \chi(x)
    \triangleq
    \max\nolimits_{a\in[\cellnum]}\{-\scalar_a-\scalar\primed_a x\}~.
\end{align*}
The function $\chi$ is convex as the maximum of affine functions.
The inequalities $\scalar_a+\scalar_b+\scalar\primed_a\prediction\primed_b\ge0$ imply
$\scalar_b\ge -\scalar_a-\scalar\primed_a\prediction\primed_b$ for every $a,b$, and hence $\scalar_b\ge\chi(\prediction\primed_b)$ for every $b$.
Also, $\chi(\prediction_a)\ge -\scalar_a-\scalar\primed_a\prediction_a$, so $\scalar_a+\scalar\primed_a\prediction_a\ge-\chi(\prediction_a)$ for every $a$.
Thus, the following holds
\begin{align*}
    \sum\nolimits_{a\in[\cellnum]}\predicmass_a(\scalar_a+\scalar\primed_a\prediction_a)
    +
    \sum\nolimits_{b\in[\cellnum\primed]}\predicmass\primed_b\scalar_b
    \ge
    -
    \sum\nolimits_{a\in[\cellnum]}\predicmass_a\chi(\prediction_a)
    +
    \sum\nolimits_{b\in[\cellnum\primed]}\predicmass\primed_b\chi(\prediction\primed_b)
    \ge
    0~,
\end{align*}
where the last inequality follows from convex order.
This contradicts the strict Farkas certificate above.
Thus, the feasibility system admits a solution $\pi$.

Conversely, suppose such a matrix $\pi$ exists.
For every convex function $\chi:[0,1]\to\R$, Jensen's inequality gives
\begin{align*}
    \sum\nolimits_{b\in[\cellnum\primed]}\predicmass\primed_b\chi(\prediction\primed_b)
    &=
    \sum\nolimits_{a\in[\cellnum]}
    \sum\nolimits_{b\in[\cellnum\primed]}
    \pi_{a,b}\chi(\prediction\primed_b)
    \\
    &\ge
    \sum\nolimits_{a\in[\cellnum]}
    \predicmass_a
    \chi\left(
    \frac{1}{\predicmass_a}
    \sum\nolimits_{b\in[\cellnum\primed]}\pi_{a,b}\prediction\primed_b
    \right)
    \\
    &=
    \sum\nolimits_{a\in[\cellnum]}\predicmass_a\chi(\prediction_a)~.
\end{align*}
Thus $\nu\primed$ dominates $\nu$ in convex order, equivalently in the integrated-CDF order.
\end{proof}

\begin{lemma}[Source-labelled splitting of a dominating expert]
\label{lem:source-labelled-decomposition}
Let $\classifier\in\classifierfam$ be represented by nonzero atoms
$\signalvec^{(1)},\ldots,\signalvec^{(\cellnum_\classifier)}\in\obscone$.
Then $\classifier\blackwelldomeq\givenclassifier_\targetindex$ if and only if there exist scalars
$\gamma_{\targetcompidx, \preindnew}\ge0$ for $\targetcompidx\in[\cellnum_\targetindex]$ and $\preindnew\in[\cellnum_\classifier]$ such that
$\sum\nolimits_{\targetcompidx\in[\cellnum_\targetindex]}\gamma_{\targetcompidx, \preindnew}=1$ for every $\preindnew$, and, defining
$\signalvec^{\targetcompidx, \preindnew}\triangleq\gamma_{\targetcompidx, \preindnew}\signalvec^{(\preindnew)}$, the following constraints hold:
\begin{align*}
    \sum\nolimits_{\preindnew\in[\cellnum_\classifier]}
    \feasubsetprob(\signalvec^{\targetcompidx, \preindnew})
    =
    \targetmass_\targetcompidx,
    \quad
    \sum\nolimits_{\preindnew\in[\cellnum_\classifier]}
    \obslabel(\signalvec^{\targetcompidx, \preindnew})
    =
    \prediction_{\targetindex,\targetcompidx}\targetmass_\targetcompidx,
    \quad
    \targetcompidx\in[\cellnum_\targetindex]~.
\end{align*}
Consequently, the split atoms also satisfy $\sum\nolimits_{\targetcompidx, \preindnew}\signalvec^{\targetcompidx, \preindnew}=\vone$.
\end{lemma}

\begin{proof}
For each $\preindnew\in[\cellnum_\classifier]$, define
$\prediction_\preindnew\triangleq\obslabel(\signalvec^{(\preindnew)})/\feasubsetprob(\signalvec^{(\preindnew)})$.
Suppose first that such scalars $\gamma_{\targetcompidx, \preindnew}$ exist.
Define $\pi_{\targetcompidx, \preindnew}\triangleq\feasubsetprob(\signalvec^{\targetcompidx, \preindnew})$.
For every source prediction component $\targetcompidx$,
\begin{align*}
    \sum\nolimits_\preindnew \pi_{\targetcompidx, \preindnew}
    =
    \sum\nolimits_\preindnew\feasubsetprob(\signalvec^{\targetcompidx, \preindnew})
    =
    \targetmass_\targetcompidx~.
\end{align*}
For every output atom $\preindnew$,
\begin{align*}
    \sum\nolimits_\targetcompidx \pi_{\targetcompidx, \preindnew}
    =
    \sum\nolimits_\targetcompidx
    \feasubsetprob(\gamma_{\targetcompidx, \preindnew}\signalvec^{(\preindnew)})
    =
    \left(
    \sum\nolimits_\targetcompidx\gamma_{\targetcompidx, \preindnew}
    \right)
    \feasubsetprob(\signalvec^{(\preindnew)})
    =
    \feasubsetprob(\signalvec^{(\preindnew)})~.
\end{align*}
Moreover, since $\signalvec^{\targetcompidx, \preindnew}=\gamma_{\targetcompidx, \preindnew}\signalvec^{(\preindnew)}$, we have
$\obslabel(\signalvec^{\targetcompidx, \preindnew})=\gamma_{\targetcompidx, \preindnew}\obslabel(\signalvec^{(\preindnew)})
=\pi_{\targetcompidx, \preindnew}\prediction_\preindnew$.
Therefore, for every $\targetcompidx$,
\begin{align*}
    \sum\nolimits_\preindnew \pi_{\targetcompidx, \preindnew}\prediction_\preindnew
    =
    \sum\nolimits_\preindnew \obslabel(\signalvec^{\targetcompidx, \preindnew})
    =
    \prediction_{\targetindex,\targetcompidx}\targetmass_\targetcompidx~.
\end{align*}
Thus $\pi$ is a martingale coupling from the target prediction distribution to the prediction distribution of $\classifier$.
By \Cref{lem:finite-martingale-blackwell}, we obtain $\classifier\blackwelldomeq\givenclassifier_\targetindex$.

Conversely, suppose $\classifier\blackwelldomeq\givenclassifier_\targetindex$.
Since both experts report Bayesian posterior-mean predictions, their prediction distributions have the same mean, equal to $\expect{\labelrea}$.
By \Cref{lem:finite-martingale-blackwell}, there exists a martingale coupling $(\pi_{\targetcompidx, \preindnew})$ between the target prediction distribution and the output prediction distribution of $\classifier$.
Since each atom $\signalvec^{(\preindnew)}$ is nonzero and $\feaprob_i>0$ for every $i$, we have $\feasubsetprob(\signalvec^{(\preindnew)})>0$.
Define
$\gamma_{\targetcompidx, \preindnew}\triangleq\pi_{\targetcompidx, \preindnew}/\feasubsetprob(\signalvec^{(\preindnew)})$ and
$\signalvec^{\targetcompidx, \preindnew}\triangleq\gamma_{\targetcompidx, \preindnew}\signalvec^{(\preindnew)}$.
Then $\gamma_{\targetcompidx, \preindnew}\ge0$ and
\begin{align*}
    \sum\nolimits_\targetcompidx \gamma_{\targetcompidx, \preindnew}
    =
    \frac{1}{\feasubsetprob(\signalvec^{(\preindnew)})}
    \sum\nolimits_\targetcompidx \pi_{\targetcompidx, \preindnew}
    =
    1~.
\end{align*}
Because $\obscone$ is a convex cone, $\signalvec^{\targetcompidx, \preindnew}\in\obscone$.
The source-mass and source-posterior-mean constraints follow from
\begin{align*}
    \sum\nolimits_\preindnew
    \feasubsetprob(\signalvec^{\targetcompidx, \preindnew})
    &=
    \sum\nolimits_\preindnew
    \frac{\pi_{\targetcompidx, \preindnew}}{\feasubsetprob(\signalvec^{(\preindnew)})}
    \feasubsetprob(\signalvec^{(\preindnew)})
    =
    \sum\nolimits_\preindnew\pi_{\targetcompidx, \preindnew}
    =
    \targetmass_\targetcompidx~,
    \\
    \sum\nolimits_\preindnew
    \obslabel(\signalvec^{\targetcompidx, \preindnew})
    &=
    \sum\nolimits_\preindnew
    \frac{\pi_{\targetcompidx, \preindnew}}{\feasubsetprob(\signalvec^{(\preindnew)})}
    \obslabel(\signalvec^{(\preindnew)})
    =
    \sum\nolimits_\preindnew
    \pi_{\targetcompidx, \preindnew}\prediction_\preindnew
    =
    \prediction_{\targetindex,\targetcompidx}\targetmass_\targetcompidx~.
\end{align*}
Finally,
\begin{align*}
    \sum\nolimits_{\targetcompidx, \preindnew}\signalvec^{\targetcompidx, \preindnew}
    =
    \sum\nolimits_\preindnew
    \left(
    \sum\nolimits_\targetcompidx\gamma_{\targetcompidx, \preindnew}
    \right)
    \signalvec^{(\preindnew)}
    =
    \sum\nolimits_\preindnew\signalvec^{(\preindnew)}
    =
    \vone~.
\end{align*}
The proof completes.
\end{proof}

\begin{lemma}[LP soundness]
\label{lem:lp-soundness}
Let $(\signalvec^{\targetcompidx,\affindex})_{\targetcompidx,\affindex}$ be any feasible solution to the program \ref{eq:discretized-LP}, and let $\classifier$ be the expert obtained by assigning to each positive-mass atom $\signalvec^{\targetcompidx,\affindex}$ the prediction
$\widehat{\prediction}_{\targetcompidx,\affindex}
=
\obslabel(\signalvec^{\targetcompidx,\affindex})/
\feasubsetprob(\signalvec^{\targetcompidx,\affindex})$.
Then $\classifier\in\classifierfam$, $\classifier\blackwelldomeq\givenclassifier_\targetindex$, and
\begin{align*}
    \expect{\loss(\predrv_\classifier,\labelrea)}
    \le
    \sum\nolimits_{\targetcompidx,\affindex}
    \left(
    \affslope_\affindex\obslabel(\signalvec^{\targetcompidx,\affindex})
    +
    \affintercept_\affindex\feasubsetprob(\signalvec^{\targetcompidx,\affindex})
    \right)~.
\end{align*}
\end{lemma}

\begin{proof}
The covering constraint of the program \ref{eq:discretized-LP} gives
$\sum\nolimits_{\targetcompidx,\affindex}\signalvec^{\targetcompidx,\affindex}=\vone$.
Each atom lies in $\obscone$.
After discarding zero atoms, this is exactly the representation of a constructible expert in \Cref{def:linearly-constructible-predictors}.
Thus $\classifier\in\classifierfam$, and by \Cref{lem:linear-constructible-calibrated}, every reported prediction generated by $\classifier$ is a Bayesian posterior mean.

We next prove Blackwell dominance.
For every source component $\targetcompidx$, the program \ref{eq:discretized-LP} imposes
\begin{align*}
    \sum\nolimits_{\affindex\in[\affnum_\eps]}
    \feasubsetprob(\signalvec^{\targetcompidx,\affindex})
    =
    \targetmass_\targetcompidx,
    \quad
    \sum\nolimits_{\affindex\in[\affnum_\eps]}
    \obslabel(\signalvec^{\targetcompidx,\affindex})
    =
    \prediction_{\targetindex,\targetcompidx}\targetmass_\targetcompidx~.
\end{align*}
For every positive-mass atom, we have
$\obslabel(\signalvec^{\targetcompidx,\affindex})
=
\feasubsetprob(\signalvec^{\targetcompidx,\affindex})
\widehat{\prediction}_{\targetcompidx,\affindex}$.
Thus, we have 
$\sum\nolimits_{\affindex\in[\affnum_\eps]}
\feasubsetprob(\signalvec^{\targetcompidx,\affindex})
\widehat{\prediction}_{\targetcompidx,\affindex}
=
\prediction_{\targetindex,\targetcompidx}\targetmass_\targetcompidx$.
Thus the LP variables define a martingale coupling from each target prediction value
$\prediction_{\targetindex,\targetcompidx}$ to the output atoms indexed by $(\targetcompidx,\affindex)$.
By \Cref{lem:finite-martingale-blackwell}, $\classifier\blackwelldomeq\givenclassifier_\targetindex$.

It remains to compare the LP objective with the true expected loss.
For every positive-mass atom, \Cref{def:regular-proper-loss} gives
us 
$\bayesrisk(\widehat{\prediction}_{\targetcompidx,\affindex})
\le
\affslope_\affindex\widehat{\prediction}_{\targetcompidx,\affindex}
+
\affintercept_\affindex$.
Multiplying by $\feasubsetprob(\signalvec^{\targetcompidx,\affindex})$ yields
\begin{align*}
    \feasubsetprob(\signalvec^{\targetcompidx,\affindex})
    \bayesrisk(\widehat{\prediction}_{\targetcompidx,\affindex})
    \le
    \affslope_\affindex
    \obslabel(\signalvec^{\targetcompidx,\affindex})
    +
    \affintercept_\affindex
    \feasubsetprob(\signalvec^{\targetcompidx,\affindex})~.
\end{align*}
Because each reported prediction of $\classifier$ is a Bayesian posterior mean, its expected proper loss is the sum of the Bayes-risk contributions of its atoms.
Summing over all atoms proves the claim.
\end{proof}

\begin{lemma}[LP completeness]
\label{lem:lp-completeness}
Let $\LPval_\eps$ be the optimal value of the program \ref{eq:discretized-LP}.
Then
$\LPval_\eps
\le
\lossopt_\targetindex+\eps$.
\end{lemma}

\begin{proof}
Let $\classifier\in\classifierfam$ be any feasible expert satisfying
$\classifier\blackwelldomeq\givenclassifier_\targetindex$.
Represent $\classifier$ by nonzero atoms
$\signalvec^{(1)},\ldots,\signalvec^{(\cellnum_\classifier)}\in\obscone$ with
$\sum\nolimits_\preindnew\signalvec^{(\preindnew)}=\vone$.
For each $\preindnew$, define
$\prediction_\preindnew\triangleq\obslabel(\signalvec^{(\preindnew)})/\feasubsetprob(\signalvec^{(\preindnew)})$.
By \Cref{lem:source-labelled-decomposition}, there exist scalar-split atoms
$\signalvec^{\targetcompidx, \preindnew}=\gamma_{\targetcompidx, \preindnew}\signalvec^{(\preindnew)}\in\obscone$ satisfying
\begin{align*}
    \sum\nolimits_{\targetcompidx, \preindnew}\signalvec^{\targetcompidx, \preindnew}
    =
    \vone,
    \quad
    \sum\nolimits_\preindnew\feasubsetprob(\signalvec^{\targetcompidx, \preindnew})
    =
    \targetmass_\targetcompidx,
    \quad
    \sum\nolimits_\preindnew\obslabel(\signalvec^{\targetcompidx, \preindnew})
    =
    \prediction_{\targetindex,\targetcompidx}\targetmass_\targetcompidx~.
\end{align*}
Because $\signalvec^{\targetcompidx, \preindnew}$ is a nonnegative scalar multiple of $\signalvec^{(\preindnew)}$, every positive-mass atom $\signalvec^{\targetcompidx, \preindnew}$ has the same prediction value $\prediction_\preindnew$.

By \Cref{def:regular-proper-loss}, for every pair $(\targetcompidx, \preindnew)$ with $\signalvec^{\targetcompidx, \preindnew}\neq\vzero$, there exists $\affindex(\targetcompidx, \preindnew)\in[\affnum_\eps]$ such that
\begin{align*}
    \affslope_{\affindex(\targetcompidx, \preindnew)}\prediction_\preindnew
    +
    \affintercept_{\affindex(\targetcompidx, \preindnew)}
    \le
    \bayesrisk(\prediction_\preindnew)+\eps~.
\end{align*}
For pairs with $\signalvec^{\targetcompidx, \preindnew}=\vzero$, choose $\affindex(\targetcompidx, \preindnew)\in[\affnum_\eps]$ arbitrarily.
For each $\targetcompidx$ and each $\affindex\in[\affnum_\eps]$, we define
\begin{align*}
    \widetilde{\signalvec}^{\targetcompidx,\affindex}
    \triangleq
    \sum\nolimits_{\preindnew\in[\cellnum_\classifier]:\ \affindex(\targetcompidx, \preindnew)=\affindex}
    \signalvec^{\targetcompidx, \preindnew}~.
\end{align*}
Because $\obscone$ is a convex cone, $\widetilde{\signalvec}^{\targetcompidx,\affindex}\in\obscone$.
The source-labelled constraints above imply that
$(\widetilde{\signalvec}^{\targetcompidx,\affindex})_{\targetcompidx,\affindex}$ is feasible for \ref{eq:discretized-LP}.

The LP objective value of this feasible solution is
\begin{align*}
    \sum\nolimits_{\targetcompidx,\affindex}
    \left(
    \affslope_\affindex\obslabel(\widetilde{\signalvec}^{\targetcompidx,\affindex})
    +
    \affintercept_\affindex\feasubsetprob(\widetilde{\signalvec}^{\targetcompidx,\affindex})
    \right)
    &=
    \sum\nolimits_{\targetcompidx, \preindnew}
    \left(
    \affslope_{\affindex(\targetcompidx, \preindnew)}
    \obslabel(\signalvec^{\targetcompidx, \preindnew})
    +
    \affintercept_{\affindex(\targetcompidx, \preindnew)}
    \feasubsetprob(\signalvec^{\targetcompidx, \preindnew})
    \right)
    \\
    &=
    \sum\nolimits_{\targetcompidx, \preindnew}
    \feasubsetprob(\signalvec^{\targetcompidx, \preindnew})
    \left(
    \affslope_{\affindex(\targetcompidx, \preindnew)}\prediction_\preindnew
    +
    \affintercept_{\affindex(\targetcompidx, \preindnew)}
    \right)
    \\
    &\le
    \sum\nolimits_{\targetcompidx, \preindnew}
    \feasubsetprob(\signalvec^{\targetcompidx, \preindnew})
    \bayesrisk(\prediction_\preindnew)
    +
    \eps
    \sum\nolimits_{\targetcompidx, \preindnew}
    \feasubsetprob(\signalvec^{\targetcompidx, \preindnew})~.
\end{align*}
Since $\sum\nolimits_{\targetcompidx, \preindnew}\signalvec^{\targetcompidx, \preindnew}=\vone$, the total mass is
$\sum\nolimits_{\targetcompidx, \preindnew}\feasubsetprob(\signalvec^{\targetcompidx, \preindnew})=\feasubsetprob(\vone)=1$.
Also, because $\signalvec^{\targetcompidx, \preindnew}=\gamma_{\targetcompidx, \preindnew}\signalvec^{(\preindnew)}$ and
$\sum\nolimits_\targetcompidx\gamma_{\targetcompidx, \preindnew}=1$, splitting the atoms does not change the prediction distribution of $\classifier$:
\begin{align*}
    \sum\nolimits_{\targetcompidx, \preindnew}
    \feasubsetprob(\signalvec^{\targetcompidx, \preindnew})
    \bayesrisk(\prediction_\preindnew)
    &=
    \sum\nolimits_\preindnew
    \left(
    \sum\nolimits_\targetcompidx\gamma_{\targetcompidx, \preindnew}
    \right)
    \feasubsetprob(\signalvec^{(\preindnew)})
    \bayesrisk(\prediction_\preindnew)
    \\
    &=
    \sum\nolimits_\preindnew
    \feasubsetprob(\signalvec^{(\preindnew)})
    \bayesrisk(\prediction_\preindnew)
    =
    \expect{\loss(\predrv_\classifier,\labelrea)}~.
\end{align*}
Therefore, this feasible LP value is at most
$\expect{\loss(\predrv_\classifier,\labelrea)}+\eps$.
Since $\classifier$ was arbitrary among experts in $\classifierfam$ satisfying
$\classifier\blackwelldomeq\givenclassifier_\targetindex$, we conclude that
$\LPval_\eps\le\lossopt_\targetindex+\eps$.
\end{proof}

\begin{proof}[Proof of \Cref{thm:proper-loss-fptas-general-k}]
We first show that the program \ref{eq:discretized-LP} is feasible.
Fix any $\bar{s}\in[\affnum_\eps]$.
Set $\signalvec^{\targetcompidx,\bar{s}}=\reportvec^{\targetindex,\targetcompidx}$ for every $\targetcompidx\in[\cellnum_\targetindex]$, and set all other variables to zero.
Each source component $\reportvec^{\targetindex,\targetcompidx}$ lies in $\obscone$.
The covering constraint follows from
$\sum\nolimits_{\targetcompidx\in[\cellnum_\targetindex]}\reportvec^{\targetindex,\targetcompidx}=\vone$.
The source-mass constraints hold because
$\feasubsetprob(\reportvec^{\targetindex,\targetcompidx})=\targetmass_\targetcompidx$.
The source-posterior-mean constraints hold because $\givenclassifier_\targetindex$ generates Bayesian posterior-mean predictions:
\begin{align*}
    \obslabel(\reportvec^{\targetindex,\targetcompidx})
    =
    \prediction_{\targetindex,\targetcompidx}
    \feasubsetprob(\reportvec^{\targetindex,\targetcompidx})
    =
    \prediction_{\targetindex,\targetcompidx}
    \targetmass_\targetcompidx~.
\end{align*}
Thus the program \ref{eq:discretized-LP} has a feasible solution.

Let $(\signalvec^{\targetcompidx,\affindex})_{\targetcompidx,\affindex}$ be an optimal solution, and let $\classifier_\eps$ be the expert output by \Cref{alg:proper-loss-fptas}.
By \Cref{lem:lp-soundness}, we know that $\classifier_\eps\in\classifierfam$ and
$\classifier_\eps\blackwelldomeq\givenclassifier_\targetindex$.
Moreover, if $\LPval_\eps$ denotes the optimal value of the program \ref{eq:discretized-LP}, then
\begin{align*}
    \expect{\loss(\predrv_{\classifier_\eps},\labelrea)}
    \le
    \LPval_\eps~.
\end{align*}
Together with \Cref{lem:lp-completeness}, this implies
\begin{align*}
    \expect{\loss(\predrv_{\classifier_\eps},\labelrea)}
    \le
    \lossopt_\targetindex+\eps~.
\end{align*}

It remains to prove the running-time bound.
Let $\totalcells =\sum\nolimits_{j\in[\classifiernum]}\cellnum_j$ be the total number of input prediction components.
For every pair $(\targetcompidx,\affindex)\in[\cellnum_\targetindex]\times[\affnum_\eps]$, the explicit LP representation uses a vector variable in $\R^{\feanum}$ and span coefficients indexed by $\compindex$.
Thus the number of scalar variables is polynomial in $\feanum$, $\totalcells$, $\cellnum_\targetindex$, and $\affnum_\eps$.
The number of constraints is also polynomial in these quantities.
All coefficients in the LP are formed from the input probabilities, input report probabilities, input prediction values, and the affine-approximation coefficients
$(\affslope_\affindex,\affintercept_\affindex)_{\affindex\in[\affnum_\eps]}$.
Under the rational encoding assumption and \Cref{def:regular-proper-loss}, the LP has polynomial bit complexity in the input size and $1/\eps$.
Therefore, standard polynomial-time algorithms for rational linear programming solve \ref{eq:discretized-LP} in time polynomial in the input size and $1/\eps$.

Finally, the algorithm creates at most one output atom for every pair
$(\targetcompidx,\affindex)\in[\cellnum_\targetindex]\times[\affnum_\eps]$.
Thus, the output support size is at most $\cellnum_\targetindex\affnum_\eps$.
\end{proof}

\section{Aggregating to Deterministic Experts}
\label{sec:det-hardness}

In this section, we study the deterministic-output versions of the two aggregation problems {\predicrefinelink} and {\predicrefineoptlink}, namely, {\detpredicrefinelink} and {\detpredicrefineoptlink}.

Our main message is that deterministic outputs introduce a genuine combinatorial obstruction.
A deterministic constructible expert induces a partition of the state space into reported-prediction cells, and every cell of this partition must have an observable binary indicator vector.
This integrality requirement makes deterministic aggregation computationally hard.
We prove two results.
First, {\detpredicrefinelink} is $\NP$-hard even with only two input experts.
Second,  {\detpredicrefineoptlink} has no multiplicative PTAS unless $\Poly=\NP$, even for the Brier loss.
Both results follow from the same amplified subset-sum construction.

\begin{theorem}[$\NP$-hardness of {\detpredicrefinelink}]
\label{thm:det-search-hard}
The problem {\detpredicrefinelink} is $\NP$-hard even when there are exactly two input experts, i.e., $\classifiernum = 2$.
More precisely, the hardness holds even when the prior distribution is uniform and the target expert $\givenclassifier_1$ is the constant deterministic expert that reports the base rate.
\end{theorem}

Our second result shows that the optimization problem {\detpredicrefineoptlink} is also computationally hard.
We state it for the Brier loss, denoted by $\brierloss(\prediction,y)=(\prediction-y)^2$.
For a calibrated expert $\classifier$, we define its expected Brier loss by
\begin{align*}
    \brierlossval(\classifier)
    \triangleq
    \expect{\brierloss(\predrv_\classifier,\labelrea)}
    =
    \expect[\prediction\sim \CDF_\classifier]{\prediction(1-\prediction)}~.
\end{align*}
For the target expert $\givenclassifier_1$, define the deterministic-output Brier-loss optimum by
\begin{align*}
    \detbrieropt
    \triangleq
    \inf\left\{
        \brierlossval(\classifier):
        \classifier\in\classifierfamdet,
        \ \classifier\blackwelldomeq\givenclassifier_1
    \right\}~.
\end{align*}
A multiplicative PTAS for  {\detpredicrefineoptlink} with Brier loss is an algorithm that, for every fixed $\varepsilon>0$, runs in time polynomial in the input size and outputs a deterministic constructible expert $\classifier^\varepsilon\in\classifierfamdet$ such that
$\classifier^\varepsilon\blackwelldomeq\givenclassifier_1$ and
$\brierlossval(\classifier^\varepsilon)\le (1+\varepsilon)\detbrieropt$.

\begin{theorem}[No multiplicative PTAS for  {\detpredicrefineoptlink}]
\label{thm:det-opt-NO-mult-ptas}
Unless $\Poly=\NP$, there is no multiplicative PTAS for  {\detpredicrefineoptlink} with Brier loss, even when there are exactly two input experts, the prior distribution is uniform, and the target expert is $\givenclassifier_1\equiv 1/2$.
\end{theorem}

\paragraph{Proof overview.}
Both hardness results use the same amplified subset-sum construction.
We reduce from the following restricted subset-sum problem.

\phantomsection
\label{prob:det-subsetsum}
\mybox{
$\rss$:\\
\textbf{Input:} positive integers $\subsuminteger_1,\ldots,\subsuminteger_\harditemnum$ and a target integer $\hardtarget$ such that
$0<\hardtarget<\hardtotal\triangleq\sum\nolimits_{i\in[\harditemnum]}\subsuminteger_i$.\\
\textbf{Question:} does there exist $\hardset\subseteq[\harditemnum]$ such that $\sum\nolimits_{i\in\hardset}\subsuminteger_i=\hardtarget$?
}

Given such an instance, we construct an aggregation instance with two input experts.
Let the block size be $\hardblocksize\triangleq\harditemnum$.
The state space consists of three parts:
\begin{align}
    \label{eq:state-space}
    \hardground
    \triangleq
    [\harditemnum]\cup\hardposblock\cup\hardnegblock~,
    \quad
    \hardposblock
    \triangleq
    \{\hardplus_1,\ldots,\hardplus_\hardblocksize\}~,
    \quad
    \hardnegblock
    \triangleq
    \{\hardminus_1,\ldots,\hardminus_\hardblocksize\}~.
\end{align}
The states in $[\harditemnum]$ are the item states, while $\hardposblock$ and $\hardnegblock$ are two special blocks used for amplification.
For a vector $\signalvec\in\R^{\hardground}$, we write $\signalvec_i$ for the coordinate of item state $i\in[\harditemnum]$, $\signalvec_{\hardplus_r}$ for the coordinate of state $\hardplus_r\in\hardposblock$, and $\signalvec_{\hardminus_r}$ for the coordinate of state $\hardminus_r\in\hardnegblock$.
The target expert is the constant deterministic expert $\givenclassifier_1\equiv1/2$.
The auxiliary randomized expert $\givenclassifier_2$ is designed so that the observable linear space defined in \Cref{defn:observable-vectors} and generated by the two input experts is exactly
\begin{align*}
    \obsspace
    =
    \hardspace
    \triangleq
    \Big\{
        \signalvec\in\R^{\hardground}:
        &\ \signalvec_{\hardplus_1}=\cdots=\signalvec_{\hardplus_\hardblocksize}~,
        \quad
        \signalvec_{\hardminus_1}=\cdots=\signalvec_{\hardminus_\hardblocksize}~,\\
        &\ \sum\nolimits_{i\in[\harditemnum]}\subsuminteger_i\signalvec_i
        -\hardtarget\signalvec_{\hardplus_1}
        -(\hardtotal-\hardtarget)\signalvec_{\hardminus_1}
        =
        0
    \Big\}~.
\end{align*}
This identity is the key structural property of the construction, and it is verified later as Property~(iv) of \Cref{lem:det-hard-instance-properties}.

The definition of $\hardspace$ forces every observable binary vector to be constant on each special block.
Thus, if $\signalvec\in\{0,1\}^{\hardground}\cap\hardspace$, then either all states in $\hardposblock$ are selected or none of them are selected, and similarly either all states in $\hardnegblock$ are selected or none of them are selected.
Once these two block choices are fixed, the remaining linear equation becomes a subset-sum equation.
More precisely, the construction gives the equivalence
\begin{align*}
    \exists \hardset\subseteq[\harditemnum]
    \text{ such that }
    \sum\nolimits_{i\in\hardset}\subsuminteger_i=\hardtarget
    \quad\Longleftrightarrow\quad
    \exists \signalvec\in\{0,1\}^{\hardground}\cap\hardspace
    \text{ such that }
    \signalvec\notin\{\vzero,\vone\}~.
\end{align*}
This equivalence is proved in \Cref{lem:det-hard-binary-vectors}.

The deterministic-output restriction is what makes this binary equivalence relevant.
Any deterministic constructible expert partitions the state space into reported-prediction level sets.
\Cref{lem:det-cells-observable} shows that the indicator vector of every such level set must lie in the observable space $\obsspace$.
Therefore, a nonconstant deterministic constructible expert exists only if $\obsspace$ contains a nontrivial observable binary vector.
By the equivalence above, this happens exactly when the {\rsslink} instance is a $\yesinst$ instance.
Together with the fact that any nonconstant calibrated expert with mean report $1/2$ strictly Blackwell-dominates the constant expert $\givenclassifier_1\equiv1/2$, this yields the hardness of {\detpredicrefinelink}.

For {\detpredicrefineoptlink}, we use the same construction and exploit the amplified blocks $\hardposblock$ and $\hardnegblock$ to create a constant Brier-loss gap.
On $\noinst$-instances, the only observable binary vectors are $\vzero$ and $\vone$, so every deterministic constructible expert must be constant.
Since the base rate is $1/2$, every feasible deterministic expert has Brier loss $1/4$.
On $\yesinst$-instances, a subset-sum solution $\hardset$ gives a nontrivial observable binary vector, and hence a two-cell deterministic constructible expert.
The amplification through $\hardposblock$ and $\hardnegblock$ ensures that this expert has Brier loss at most $2/9$.
Thus, a multiplicative PTAS with sufficiently small fixed accuracy could distinguish $\yesinst$-instances from $\noinst$-instances, which would imply $\Poly=\NP$.

\subsection{Hard Instance with Two Input Experts}
\label{subsec:det-hard-instance}

We first prove the hardness of the source problem.

\begin{lemma}
\label{lem:det-rss-np-hard}
The problem {\rsslink} is $\NP$-hard.
\end{lemma}

\begin{proof}
It is known that the usual \textsc{SubsetSum} problem with positive item sizes is $\NP$-hard.
Given an instance $(\subsuminteger_1,\ldots,\subsuminteger_\harditemnum;\hardtarget)$ of \textsc{SubsetSum}, let $\hardtotal=\sum\nolimits_{i\in[\harditemnum]}\subsuminteger_i$.
If $\hardtarget=0$, the instance is trivially a $\yesinst$ instance.
If $\hardtarget>\hardtotal$, the instance is trivially a $\noinst$ instance.
If $\hardtarget=\hardtotal$, the instance is trivially a $\yesinst$ instance, witnessed by the full set $[\harditemnum]$.
In the remaining case, $0<\hardtarget<\hardtotal$, the instance is exactly an instance of {\rsslink}.
Thus a polynomial-time algorithm for {\rsslink}, together with these elementary checks, would solve \textsc{SubsetSum} in polynomial time.
Therefore {\rsslink} is $\NP$-hard.
\end{proof}

\begin{example}[Hard instance]
\label{ex:hard-instance}
Fix an instance of {\rsslink}, with positive integers $\subsuminteger_1,\ldots,\subsuminteger_\harditemnum$, target $\hardtarget$, and total $\hardtotal=\sum\nolimits_{i\in[\harditemnum]}\subsuminteger_i$, where $0<\hardtarget<\hardtotal$.
We construct an aggregation instance with exactly two input experts.
The construction has three parts.
\begin{itemize}
    \item \textbf{State space and prior distribution:}
    Let the block size be $\hardblocksize=\harditemnum$.
    Let the state space be defined as in Eqn.~\eqref{eq:state-space}.
    Thus $\feanum=|\hardground|=\harditemnum+2\hardblocksize=3\harditemnum$.
    We use the uniform prior distribution, $\feaprob_\fea=1/\feanum$ for every $\fea\in\hardground$.

    \item \textbf{Bayes probabilities and target expert:}
    Define Bayes probabilities by
    \begin{align*}
        \bayesprob_i=\frac12~, \quad i\in[\harditemnum]~;
        \quad
        \bayesprob_{\hardplus_r}=\frac34~, \quad r\in[\hardblocksize]~;
        \quad
        \bayesprob_{\hardminus_r}=\frac14~, \quad r\in[\hardblocksize]~.
    \end{align*}
    The base rate is
    \begin{align*}
        \sum\nolimits_{\fea\in\hardground}\feaprob_\fea\bayesprob_\fea
        &=
        \frac{\harditemnum}{\feanum}\cdot\frac12
        +
        \frac{\hardblocksize}{\feanum}\cdot\frac34
        +
        \frac{\hardblocksize}{\feanum}\cdot\frac14
        =
        \frac{\harditemnum}{2\feanum}+\frac{\hardblocksize}{\feanum}
        =
        \frac{\harditemnum/2+\harditemnum}{3\harditemnum}
        =
        \frac12~.
    \end{align*}
    The target expert is the constant deterministic expert
    $\givenclassifier_1\equiv1/2$.
    \item \textbf{One randomized auxiliary input expert:}
    We now define the auxiliary randomized expert $\givenclassifier_2$.
    Let $\hardeta=1/16$.
    For each item $i\in[\harditemnum]$, we define the item perturbation $\harditemeps_i=\hardeta^i/\subsuminteger_i$.

    We first define an auxiliary vector (indexed by the superscript $i$) $\hardbasis^{(i)}\in\R^{\hardground}$ by
    $\hardbasis^{(i)}_i=1$, $\hardbasis^{(i)}_{\hardplus_r}=\subsuminteger_i/\hardtarget$ for every $r\in[\hardblocksize]$, and $\hardbasis^{(i)}_\fea=0$ for all other states $\fea\in\hardground$.
    For each item $i\in[\harditemnum]$, we define two routing vectors
    \begin{align*}
        \reportvec^{2, (i,\hardplus)}
        \triangleq
        \frac1{\harditemnum}\left(\frac12\vone+\harditemeps_i\hardbasis^{(i)}\right)~,
        \quad
        \reportvec^{2, (i,\hardminus)}
        \triangleq
        \frac1{\harditemnum}\left(\frac12\vone-\harditemeps_i\hardbasis^{(i)}\right)~.
    \end{align*}
    The expert $\givenclassifier_2$ reports $\prediction_{2, (i,\hardplus)}$ with state-wise probabilities $\reportvec^{2, (i,\hardplus)}$, and reports $\prediction_{2, (i,\hardminus)}$ with state-wise probabilities $\reportvec^{2, (i,\hardminus)}$.
    The prediction values are defined as the corresponding conditional Bayes means:
    \begin{align*}
        \prediction_{2, (i,\hardplus)}
        \triangleq
        \frac{\feasubsetlabel(\reportvec^{2, (i,\hardplus)})}{\feasubsetprob(\reportvec^{2, (i,\hardplus)})}~,
        \quad
        \prediction_{2, (i,\hardminus)}
        \triangleq
        \frac{\feasubsetlabel(\reportvec^{2, (i,\hardminus)})}{\feasubsetprob(\reportvec^{2, (i,\hardminus)})}~.
    \end{align*}
\end{itemize}
\end{example}
From the above construction, the two input experts are precisely $\givenclassifier_1$ and $\givenclassifier_2$, so $\classifiernum=2$.

The next lemma collects the basic properties of the constructed instance. These properties
verify that the construction is well-defined, calibrated, has the intended observable linear space,
and has polynomial encoding size.
\begin{lemma}[Properties of the hard instance]
\label{lem:det-hard-instance-properties}
The constructed two-expert instance satisfies the following properties.
\begin{enumerate}[label=(\roman*)]
    \item The routing vectors $\reportvec^{2, (i,\hardplus)}$ and $\reportvec^{2, (i,\hardminus)}$, for $i\in[\harditemnum]$, define a valid finite-support randomized expert $\givenclassifier_2$.
    \item The $2\harditemnum$ prediction values of $\givenclassifier_2$ are all distinct. Moreover, for every $i\in[\harditemnum]$, we have $\prediction_{2, (i,\hardplus)}>1/2>\prediction_{2, (i,\hardminus)}$.
    \item The experts $\givenclassifier_1$ and $\givenclassifier_2$ are calibrated under the Bayes probabilities defined above.
    \item The observable linear space of the constructed input experts is $\obsspace=\hardspace$.
    \item The constructed instance has encoding size polynomial in the encoding length of the {\rsslink} instance.
\end{enumerate}
\end{lemma}

\begin{proof}
We verify each property one by one.

\xhdr{Property (i)} 
We first verify that the routing vectors define a valid randomized expert. 
For each item $i\in[\harditemnum]$, we have
\begin{align*}
    0<\harditemeps_i\hardbasis^{(i)}_i
    =
    \frac{\hardeta^i}{\subsuminteger_i}
    \le
    \frac1{16}~,
    \quad
    0<\harditemeps_i\hardbasis^{(i)}_{\hardplus_r}
    =
    \frac{\hardeta^i}{\hardtarget}
    \le
    \frac1{16}
    \quad\text{for every }r\in[\hardblocksize]~.
\end{align*}
All other coordinates of $\hardbasis^{(i)}$ are zero.
Thus, every coordinate of the unscaled vectors $\frac12\vone+\harditemeps_i\hardbasis^{(i)}$ and $\frac12\vone-\harditemeps_i\hardbasis^{(i)}$ are coordinate-wise nonnegative.
Moreover, $\reportvec^{2, (i,\hardplus)}+\reportvec^{2, (i,\hardminus)}=1/\harditemnum \cdot \vone$ for every $i$, and hence
\begin{align*}
    \sum\nolimits_{i \in [\harditemnum]}\left(\reportvec^{2, (i,\hardplus)}+\reportvec^{2, (i,\hardminus)}\right)
    =
    \vone~.
\end{align*}
Thus the routing probabilities are nonnegative and sum to one at every state, so they define a valid randomized expert $\givenclassifier_2$.

\xhdr{Property (ii)}
We next prove that the reports of $\givenclassifier_2$ have distinct prediction values.
For notational convenience, we define the unscaled routing vectors
\begin{align*}
    \widetilde{\reportvec}^{2, (i,\hardplus)}
    \triangleq
    \frac12\vone+\harditemeps_i\hardbasis^{(i)}~,
    \quad
    \widetilde{\reportvec}^{2, (i,\hardminus)}
    \triangleq
    \frac12\vone-\harditemeps_i\hardbasis^{(i)}~.
\end{align*}
The common factor $1/\harditemnum$ cancels in the conditional-mean ratios, so
\begin{align*}
    \prediction_{2, (i,\hardplus)}
    =
    \frac{\feasubsetlabel(\widetilde{\reportvec}^{2, (i,\hardplus)})}
    {\feasubsetprob(\widetilde{\reportvec}^{2, (i,\hardplus)})}~,
    \quad
    \prediction_{2, (i,\hardminus)}
    =
    \frac{\feasubsetlabel(\widetilde{\reportvec}^{2, (i,\hardminus)})}
    {\feasubsetprob(\widetilde{\reportvec}^{2, (i,\hardminus)})}~.
\end{align*}
For every vector $\signalvec\in\R^{\hardground}$, the label-mass functional satisfies
\begin{align}
    \feasubsetlabel(\signalvec)
    =
    \frac12\feasubsetprob(\signalvec)
    +
    \frac{1}{4\feanum}
    \left(
        \sum\nolimits_{r\in[\hardblocksize] }\signalvec_{\hardplus_r}
        -
        \sum\nolimits_{r\in[\hardblocksize] }\signalvec_{\hardminus_r}
    \right)~.
    \label{eq:det-hard-label-identity}
\end{align}
Indeed, this follows from $\bayesprob_i=1/2$, $\bayesprob_{\hardplus_r}=3/4$, $\bayesprob_{\hardminus_r}=1/4$, and $\feaprob_\fea=1/\feanum$.
For the positive report,
\begin{align*}
    \sum\nolimits_{r\in[\hardblocksize] }\widetilde{\reportvec}^{2, (i,\hardplus)}_{\hardplus_r}
    -
    \sum\nolimits_{r\in[\hardblocksize] }\widetilde{\reportvec}^{2, (i,\hardplus)}_{\hardminus_r}
    =
    \hardblocksize\harditemeps_i\frac{\subsuminteger_i}{\hardtarget}
    =
    \frac{\hardblocksize\hardeta^i}{\hardtarget}~,
\end{align*}
whereas the same block difference for $\widetilde{\reportvec}^{2, (i,\hardminus)}$ is $-\hardblocksize\hardeta^i/\hardtarget$.
Therefore
\begin{align*}
    \prediction_{2, (i,\hardplus)}-\frac12
    =
    \frac{\hardblocksize\hardeta^i}
    {4\feanum\hardtarget\feasubsetprob(\widetilde{\reportvec}^{2, (i,\hardplus)})}~,
    \quad
    \frac12-\prediction_{2, (i,\hardminus)}
    =
    \frac{\hardblocksize\hardeta^i}
    {4\feanum\hardtarget\feasubsetprob(\widetilde{\reportvec}^{2, (i,\hardminus)})}~.
\end{align*}
It remains to bound the denominators.
Since 
$\feasubsetprob(\hardbasis^{(i)})
=
\frac1\feanum\left(1+\hardblocksize\frac{\subsuminteger_i}{\hardtarget}\right)$, we have
\begin{align*}
    \harditemeps_i\feasubsetprob(\hardbasis^{(i)})
    =
    \frac{\hardeta^i}{\subsuminteger_i}\cdot
    \frac1\feanum\left(1+\hardblocksize\frac{\subsuminteger_i}{\hardtarget}\right)
    =
    \hardeta^i\left(\frac1{\feanum \subsuminteger_i}+\frac{\hardblocksize}{\feanum\hardtarget}\right) 
    \le
    \frac1{16}\left(\frac13+\frac13\right)
    =
    \frac1{24}~,
\end{align*}
where we used $\subsuminteger_i\ge1$, $\hardtarget\ge1$, $\hardblocksize=\harditemnum$, and $\feanum=3\harditemnum$.
Thus $\feasubsetprob(\widetilde{\reportvec}^{2, (i,\hardplus)})$ and $\feasubsetprob(\widetilde{\reportvec}^{2, (i,\hardminus)})$ both lie in $[1/4,3/4]$.
Consequently, we have
\begin{align*}
    \frac{\hardblocksize\hardeta^i}{3\feanum\hardtarget}
    \le
    \prediction_{2, (i,\hardplus)}-\frac12
    \le
    \frac{\hardblocksize\hardeta^i}{\feanum\hardtarget}~,
    \quad
    \frac{\hardblocksize\hardeta^i}{3\feanum\hardtarget}
    \le
    \frac12-\prediction_{2, (i,\hardminus)}
    \le
    \frac{\hardblocksize\hardeta^i}{\feanum\hardtarget}~.
\end{align*}
If $i<j$, then $j\ge i+1$, and since $\hardeta=1/16<1/3$,
\begin{align*}
    \frac{\hardblocksize\hardeta^i}{3\feanum\hardtarget}
    >
    \frac{\hardblocksize\hardeta^{i+1}}{\feanum\hardtarget}
    \ge
    \frac{\hardblocksize\hardeta^j}{\feanum\hardtarget}~.
\end{align*}
Thus the positive deviations $\prediction_{2, (i,\hardplus)}-1/2$ are pairwise distinct, and the negative deviations $1/2-\prediction_{2, (i,\hardminus)}$ are pairwise distinct.
Moreover, each $\prediction_{2, (i,\hardplus)}$ is above $1/2$, while each $\prediction_{2, (i,\hardminus)}$ is below $1/2$.
Thus, all $2\harditemnum$ prediction values are distinct.

\xhdr{Property (iii)} Part (iii) follows immediately.
The expert $\givenclassifier_1\equiv1/2$ is calibrated because the base rate is $1/2$.

For the expert $\givenclassifier_2$, each reported prediction value is defined as the conditional Bayes mean of its routing component, and part (ii) shows that all prediction values are distinct.
Thus conditioning on a prediction value is the same as conditioning on the corresponding routing component, so $\givenclassifier_2$ is calibrated.

\xhdr{Property (iv)}
For part (iv), since $\givenclassifier_1$ is constant, it contributes the routing vector $\vone$.
Since all reports of $\givenclassifier_2$ are distinct, the observable linear space also contains each routing vector $\reportvec^{2, (i,\hardplus)}$ and $\reportvec^{2, (i,\hardminus)}$.
Using the unscaled vectors above, we obtain
\begin{align*}
    \obsspace
    =
    \operatorname{span}\left\{\vone,\hardbasis^{(1)},\ldots,\hardbasis^{(\harditemnum)}\right\}~.
\end{align*}
We claim that this span equals $\hardspace$.
First, $\vone\in\hardspace$, because $\sum\nolimits_{i \in [\harditemnum]}\subsuminteger_i-\hardtarget-(\hardtotal-\hardtarget)=0$.
For each $i\in[\harditemnum]$, the vector $\hardbasis^{(i)}$ belongs to $\hardspace$: its positive-block coordinates are all $\subsuminteger_i/\hardtarget$, its negative-block coordinates are all $0$, and
\begin{align*}
    \subsuminteger_i
    -
    \hardtarget\cdot\frac{\subsuminteger_i}{\hardtarget}
    -
    (\hardtotal-\hardtarget)\cdot0
    =
    0~.
\end{align*}
Thus, we have $\obsspace\subseteq\hardspace$.
The vectors $\vone,\hardbasis^{(1)},\ldots,\hardbasis^{(\harditemnum)}$ are linearly independent: if $\alpha_0\vone+\sum\nolimits_{i \in [\harditemnum]}\alpha_i\hardbasis^{(i)}=0$, then looking at any negative-block coordinate gives $\alpha_0=0$, and then looking at item coordinate $i$ gives $\alpha_i=0$ for every $i$.
Thus $\dim(\obsspace)=\harditemnum+1$.
On the other hand, a vector in $\hardspace$ is determined by the $\harditemnum$ item coordinates, one common positive-block coordinate, and one common negative-block coordinate, subject to one nontrivial linear equation.
Therefore $\dim(\hardspace)=\harditemnum+1$.
It follows that $\obsspace=\hardspace$.

\xhdr{Property (v)}
Part (v) follows because all probabilities are rational with polynomial bit complexity.
The numbers $\hardeta^i=16^{-i}$ have binary encoding length $O(i)$, and $i\le\harditemnum$.
Each routing probability is obtained from $\harditemnum$, $\subsuminteger_i$, $\hardtarget$, and $16^i$ by a constant number of rational arithmetic operations.
Each prediction value is a ratio of two sums over $\feanum=3\harditemnum$ rational terms of polynomial bit complexity.
Thus every prediction value has polynomial bit complexity, and the full constructed instance has polynomial encoding size.
\end{proof}

\subsection{Technical Lemmas}
\label{subsec:det-observable-binary-cells}

We next prove the technical lemmas used in both deterministic hardness results.
The first and third lemmas are general facts about deterministic constructible experts and Blackwell dominance.
The second lemma specializes to the hard instance from \Cref{subsec:det-hard-instance}; there we repeatedly use Property~(iv) of \Cref{lem:det-hard-instance-properties}, namely, the identity $\obsspace=\hardspace$.

The first lemma is the key characterization where deterministic outputs become an integrality constraint. It says that, in the constructed hard instance, every reported-prediction level set of any deterministic constructible expert must have an observable indicator vector.

\begin{lemma}[Deterministic cells are observable]
\label{lem:det-cells-observable}
Let $\classifier\in\classifierfamdet$ be a deterministic constructible expert.
For every prediction value $\prediction\in\supp(\CDF_\classifier)$, define the level set
$\cellnew_\prediction\triangleq\{\fea\in\feaspace:\classifier(\fea)=\prediction\}$.
Then $\subsetindi_{\cellnew_\prediction}\in\obscone\subseteq\obsspace$, and the prediction value on this level set satisfies
$\prediction
=
\frac{\obslabel(\subsetindi_{\cellnew_\prediction})}
{\feasubsetprob(\subsetindi_{\cellnew_\prediction})}$.
\end{lemma}

\begin{proof}
Because $\classifier$ is constructible, it admits a representation by nonzero atoms
$\signalvec^{(1)},\ldots,\signalvec^{(\cellnum_\classifier)}\in\obscone$ such that
$\sum\nolimits_{\preindnew\in[\cellnum_\classifier]} \signalvec^{(\preindnew)}=\vone$.
For each atom $\preindnew $, let
$\prediction_\preindnew \triangleq \obslabel(\signalvec^{(\preindnew)})/\feasubsetprob(\signalvec^{(\preindnew)})$
be its reported prediction value.
Thus, for every state $\fea\in\feaspace$, the expert can be written as
$\classifier(\cdot \mid \fea)
=
\sum\nolimits_{\preindnew\in[\cellnum_\classifier]} \signalvec^{(\preindnew)}_\fea\delta_{(\prediction_\preindnew )}(\cdot)$.

Fix a prediction value $\prediction\in\supp(\CDF_\classifier)$.
Let $H_\prediction\triangleq\{\preindnew\in[\cellnum_\classifier]:\prediction_\preindnew =\prediction\}$, and define the pooled routing vector
$\signalvec_\prediction
    \triangleq
    \sum\nolimits_{\preindnew\in H_\prediction}\signalvec^{(\preindnew)}$.
Since $\obscone$ is a cone, it is closed under finite nonnegative sums, and therefore $\signalvec_\prediction\in\obscone$.
For each state $\fea$, the coordinate $(\signalvec_\prediction)_\fea$ is exactly the probability that $\classifier$ reports $\prediction$ conditional on state $\fea$:
\begin{align*}
    (\signalvec_\prediction)_\fea
    =
    \sum\nolimits_{\preindnew\in H_\prediction}\signalvec^{(\preindnew)}_\fea
    =
    \prob{\predrv_\classifier=\prediction\mid \fea}~.
\end{align*}
Because $\classifier$ is deterministic, the conditional distribution $\classifier(\cdot\mid \fea)$ is a point mass at the unique value $\classifier(\fea)$.
Thus, 
\begin{align*}
    (\signalvec_\prediction)_\fea
    =
    \begin{cases}
        1, & \text{if }\classifier(\fea)=\prediction,\\
        0, & \text{if }\classifier(\fea)\ne\prediction.
    \end{cases}
\end{align*}
Therefore $\signalvec_\prediction=\subsetindi_{\cellnew_\prediction}$, and so
$\subsetindi_{\cellnew_\prediction}\in\obscone\subseteq\obsspace$.

It remains to identify the prediction value on this level set.
By linearity of $\obslabel$ and by the definition of $\prediction_\preindnew $, we have
\begin{align*}
    \obslabel(\subsetindi_{\cellnew_\prediction})
    =
    \obslabel(\signalvec_\prediction)
    =
    \sum\nolimits_{\preindnew\in H_\prediction}\obslabel(\signalvec^{(\preindnew)}) 
    =
    \sum\nolimits_{\preindnew\in H_\prediction}\prediction_\preindnew \feasubsetprob(\signalvec^{(\preindnew)})
    & =
    \sum\nolimits_{\preindnew\in H_\prediction}\prediction\feasubsetprob(\signalvec^{(\preindnew)}) \\
    &=
    \prediction\feasubsetprob(\signalvec_\prediction)
    =
    \prediction\feasubsetprob(\subsetindi_{\cellnew_\prediction})~.
\end{align*}
Since $\prediction\in\supp(\CDF_\classifier)$, the level set $\cellnew_\prediction$ is nonempty.
Under the standing assumption that all prior probabilities are positive, this implies
$\feasubsetprob(\subsetindi_{\cellnew_\prediction})>0$.
Dividing by this mass proves the desired formula for $\prediction$.
\end{proof}

The next lemma shows that, in the hard instance, nontrivial observable binary vectors are exactly subset-sum certificates.

\begin{lemma}[Binary vectors in the hard observable space]
\label{lem:det-hard-binary-vectors}
For the hard instance constructed in \Cref{subsec:det-hard-instance}, the following statements are equivalent:
\begin{enumerate}[label=(\roman*)]
    \item There exists $\hardset\subseteq[\harditemnum]$ such that $\sum\nolimits_{i\in\hardset}\subsuminteger_i=\hardtarget$.
    \item There exists $\signalvec\in\{0,1\}^{\hardground}\cap\obsspace$ such that $\signalvec\notin\{\vzero,\vone\}$.
\end{enumerate}
Moreover, every nontrivial binary vector $\signalvec\in\{0,1\}^{\hardground}\cap\obsspace$ has different common values on the positive and negative blocks.
\end{lemma}

\begin{proof}
By Property~(iv) of \Cref{lem:det-hard-instance-properties}, the observable linear space of the hard instance satisfies $\obsspace=\hardspace$.
We prove the two directions separately.

\xhdr{The $(ii) \Rightarrow (i)$ direction}
Let $\signalvec\in\{0,1\}^{\hardground}\cap\obsspace$.
Since $\obsspace=\hardspace$, the definition of $\hardspace$ implies that all positive-block coordinates of $\signalvec$ are equal, and all negative-block coordinates of $\signalvec$ are equal.
Denote these two common values by $s_+\in\{0,1\}$ and $s_-\in\{0,1\}$, respectively:
\begin{align*}
    \signalvec_{\hardplus_1}=\cdots=\signalvec_{\hardplus_\hardblocksize}=s_+~,
    \quad
    \signalvec_{\hardminus_1}=\cdots=\signalvec_{\hardminus_\hardblocksize}=s_-~.
\end{align*}
The defining equation of $\hardspace$ then becomes
\begin{align}
    \sum\nolimits_{i\in[\harditemnum]}\subsuminteger_i\signalvec_i
    =
    \hardtarget s_+ + (\hardtotal-\hardtarget)s_-~.
    \label{eq:det-binary-hard-equation}
\end{align}
We analyze the four possible values of $(s_+,s_-)$.
\begin{itemize}
    \item If $s_+=0$ and $s_-=0$, then Eqn.~\eqref{eq:det-binary-hard-equation} gives
    $\sum\nolimits_{i\in[\harditemnum]}\subsuminteger_i\signalvec_i=0$.
    Since every $\subsuminteger_i$ is positive and every $\signalvec_i$ is binary, we must have $\signalvec_i=0$ for every item $i\in[\harditemnum]$.
    Together with $s_+=s_-=0$, this implies $\signalvec=\vzero$.

    \item If $s_+=1$ and $s_-=1$, then Eqn.~\eqref{eq:det-binary-hard-equation} gives
    $\sum\nolimits_{i\in[\harditemnum]}\subsuminteger_i\signalvec_i=\hardtotal$.
    Since $\hardtotal=\sum\nolimits_{i\in[\harditemnum]}\subsuminteger_i$, every $\subsuminteger_i$ is positive, and every $\signalvec_i$ is binary, we must have $\signalvec_i=1$ for every item $i\in[\harditemnum]$.
    Together with $s_+=s_-=1$, this implies $\signalvec=\vone$.

    \item If $s_+=1$ and $s_-=0$, then Eqn.~\eqref{eq:det-binary-hard-equation} gives
    $\sum\nolimits_{i\in[\harditemnum]}\subsuminteger_i\signalvec_i=\hardtarget$.
    Thus, the item set $\hardset\triangleq\{i\in[\harditemnum]:\signalvec_i=1\}$ satisfies
    $\sum\nolimits_{i\in\hardset}\subsuminteger_i=\hardtarget$.

    \item If $s_+=0$ and $s_-=1$, then Eqn.~\eqref{eq:det-binary-hard-equation} gives
    $\sum\nolimits_{i\in[\harditemnum]}\subsuminteger_i\signalvec_i=\hardtotal-\hardtarget$.
    Thus, the complementary item set $\hardset\triangleq\{i\in[\harditemnum]:\signalvec_i=0\}$ satisfies
    \begin{align*}
        \sum\nolimits_{i\in\hardset}\subsuminteger_i
        =
        \hardtotal-\sum\nolimits_{i\in[\harditemnum]}\subsuminteger_i\signalvec_i
        =
        \hardtotal-(\hardtotal-\hardtarget)
        =
        \hardtarget~.
    \end{align*}
\end{itemize}
The case analysis proves that every nonzero, non-all-ones binary vector in $\obsspace$ yields a feasible solution to the {\rsslink} instance.
It also shows that every nontrivial binary vector must satisfy $s_+\ne s_-$, because the two cases with $s_+=s_-$ give exactly $\vzero$ and $\vone$.

\xhdr{The $(i) \Rightarrow (ii)$ direction}
Conversely, suppose there exists $\hardset\subseteq[\harditemnum]$ such that
$\sum\nolimits_{i\in\hardset}\subsuminteger_i=\hardtarget$.
Define $\signalvec\in\{0,1\}^{\hardground}$ by
\begin{align*}
    \signalvec_i=1 \text{ if and only if } i\in\hardset~,
    \quad
    \signalvec_{\hardplus_r}=1 \text{ for every } r\in[\hardblocksize]~,
    \quad
    \signalvec_{\hardminus_r}=0 \text{ for every } r\in[\hardblocksize]~.
\end{align*}
Then $\signalvec$ satisfies the defining equation of $\hardspace$, because
\begin{align*}
    \sum\nolimits_{i\in[\harditemnum]}\subsuminteger_i\signalvec_i
    -
    \hardtarget\signalvec_{\hardplus_1}
    -
    (\hardtotal-\hardtarget)\signalvec_{\hardminus_1}
    =
    \sum\nolimits_{i\in\hardset}\subsuminteger_i-\hardtarget
    =
    0~.
\end{align*}
It also satisfies the positive-block and negative-block equality constraints in the definition of $\hardspace$.
Thus $\signalvec\in\hardspace=\obsspace$, where the last equality again follows from Property~(iv) of \Cref{lem:det-hard-instance-properties}.
Finally, $\signalvec\ne\vzero$ because it includes all positive-block states, and $\signalvec\ne\vone$ because it excludes all negative-block states.
This proves the equivalence.
\end{proof}

We also use the following Blackwell-dominance fact.

\begin{lemma}[A nonconstant calibrated expert dominates the constant expert]
\label{lem:det-nonconstant-dominates-constant}
Let $\classifier$ be a calibrated expert whose prediction random variable $\predrv_\classifier$ is nonconstant and has mean $\baserate\in(0,1)$.
Let $\givenclassifier^{\baserate}$ be the constant expert that reports $\baserate$.
Then $\classifier\blackwelldom \givenclassifier^{\baserate}$.
\end{lemma}

\begin{proof}
For every $t\in[0,1]$, the function $\prediction\mapsto(t-\prediction)_+$ is convex.
By Jensen's inequality,
\begin{align*}
    \SCDF_\classifier(t)
    =
    \expect{(t-\predrv_\classifier)_+}
    \ge
    \left(t-\expect{\predrv_\classifier}\right)_+
    =
    (t-\baserate)_+
    =
    \SCDF_{\givenclassifier^{\baserate}}(t)~.
\end{align*}
Thus $\classifier\blackwelldomeq \givenclassifier^{\baserate}$.

It remains to prove strictness.
Since $\predrv_\classifier$ is nonconstant and $\expect{\predrv_\classifier}=\baserate$, it must place positive probability both below and above $\baserate$.
Indeed, if $\prob{\predrv_\classifier<\baserate}=0$, then $\predrv_\classifier\ge\baserate$ almost surely and $\expect{\predrv_\classifier}=\baserate$, which would force $\predrv_\classifier=\baserate$ almost surely.
This contradicts nonconstancy.
The same argument rules out $\prob{\predrv_\classifier>\baserate}=0$.
Therefore
\begin{align*}
    \SCDF_\classifier(\baserate)
    =
    \expect{(\baserate-\predrv_\classifier)_+}
    >
    0
    =
    \SCDF_{\givenclassifier^{\baserate}}(\baserate)~.
\end{align*}
Because $\baserate\in(0,1)$, this is strict dominance at an interior point.
Hence $\classifier\blackwelldom \givenclassifier^{\baserate}$.
\end{proof}

\subsection{Hardness of \texorpdfstring{\detpredicrefinelink}{Search-DetRefine}}
\label{subsec:det-search-hardness}

We first prove the core equivalence for the hard instance.

\begin{lemma}
\label{lem:det-strict-improvement-subset-sum}
For the hard two-expert instance in \Cref{ex:hard-instance}, constructed from a {\rsslink} instance, the following statements are equivalent:
\begin{enumerate}[label=(\roman*)]
    \item The {\rsslink} instance is a $\yesinst$ instance.
    \item There exists $\classifier\in\classifierfamdet$ such that $\classifier\blackwelldom\givenclassifier_1$.
\end{enumerate}
\end{lemma}

\begin{proof}
We prove the two directions one by one.

\xhdr{The $(i) \Rightarrow (ii)$ direction}
Suppose first that the subset-sum instance is a $\yesinst$ instance, and let $\hardset\subseteq[\harditemnum]$ satisfy $\sum\nolimits_{i\in\hardset}\subsuminteger_i=\hardtarget$.
Define $\signalvec\in\{0,1\}^{\hardground}$ as in \Cref{lem:det-hard-binary-vectors}: it includes exactly the item states in $\hardset$, includes all states in $\hardposblock$, and excludes all states in $\hardnegblock$.
Then $\signalvec\in\obsspace$ and $\signalvec\notin\{\vzero,\vone\}$.
Since $\vone\in\obsspace$, also $\vone-\signalvec\in\obsspace$.
Both $\signalvec$ and $\vone-\signalvec$ are nonzero and coordinate-wise nonnegative, so they belong to $\obscone$.
They therefore define a two-cell deterministic constructible expert $\classifier$: on states with $\signalvec_\fea=1$, it reports $\obslabel(\signalvec)/\feasubsetprob(\signalvec)$, and on states with $\signalvec_\fea=0$, it reports $\obslabel(\vone-\signalvec)/\feasubsetprob(\vone-\signalvec)$.

Since $\signalvec,\vone-\signalvec\in\obsspace$, \Cref{lem:observable-label-functional} gives $\obslabel(\signalvec)=\feasubsetlabel(\signalvec)$ and $\obslabel(\vone-\signalvec)=\feasubsetlabel(\vone-\signalvec)$.
Using Eqn.~\eqref{eq:det-hard-label-identity}, and using that $\signalvec$ contains all positive-block states and no negative-block states, while $\vone-\signalvec$ contains all negative-block states and no positive-block states, we get
\begin{align*}
    \frac{\obslabel(\signalvec)}{\feasubsetprob(\signalvec)}
    &=
    \frac{\feasubsetlabel(\signalvec)}{\feasubsetprob(\signalvec)}
    =
    \frac12+\frac{\hardblocksize}{4\feanum\feasubsetprob(\signalvec)}
    >
    \frac12,
    \\
    \frac{\obslabel(\vone-\signalvec)}{\feasubsetprob(\vone-\signalvec)}
    &=
    \frac{\feasubsetlabel(\vone-\signalvec)}{\feasubsetprob(\vone-\signalvec)}
    =
    \frac12-\frac{\hardblocksize}{4\feanum\feasubsetprob(\vone-\signalvec)}
    <
    \frac12~.
\end{align*}
Thus the expert $\classifier$ is nonconstant.
Because $\classifier$ is constructible, \Cref{lem:linear-constructible-calibrated} implies that it is calibrated.
Its mean report is the base rate, which is $1/2$.
By \Cref{lem:det-nonconstant-dominates-constant}, $\classifier\blackwelldom\givenclassifier_1$.

\xhdr{The $(ii) \Rightarrow (i)$ direction}
Conversely, suppose there exists $\classifier\in\classifierfamdet$ such that $\classifier\blackwelldom\givenclassifier_1$.
Because $\classifier$ is constructible, it is calibrated by \Cref{lem:linear-constructible-calibrated}, so $\expect{\predrv_\classifier}=1/2$.
If $\classifier$ were constant, calibration would force $\classifier\equiv1/2$, so it could not strictly Blackwell dominate $\givenclassifier_1$.
Thus, the expert $\classifier$ is nonconstant.
Therefore, it has some nonempty proper level set $\cellnew\subsetneq\hardground$.
By \Cref{lem:det-cells-observable}, $\subsetindi_\cellnew\in\obsspace$.
Since $\cellnew$ is nonempty and proper, $\subsetindi_\cellnew\notin\{\vzero,\vone\}$.
By \Cref{lem:det-hard-binary-vectors}, the original {\rsslink} instance is a $\yesinst$ instance.
\end{proof}

The search problem requires an undominated strict improvement. The preceding lemma only concerns the existence of some strict deterministic improvement. The next lemma connects the two notions.

\begin{lemma}
\label{lem:det-undominated-extension}
For the hard instance, there exists $\classifier\in\classifierfamdet$ such that $\classifier\blackwelldom\givenclassifier_1$ if and only if there exists $\classifier^\star\in\strictdomfamdet$ such that $\classifier^\star\blackwelldom\givenclassifier_1$.
\end{lemma}

\begin{proof}
The reverse direction is immediate because $\strictdomfamdet\subseteq\classifierfamdet$.
For the forward direction, suppose there exists $\classifier\in\classifierfamdet$ such that $\classifier\blackwelldom\givenclassifier_1$.
The state space $\hardground$ is finite.
A deterministic expert is determined by a partition of $\hardground$ into nonempty level sets and a reported prediction value on each level set.
By \Cref{lem:det-cells-observable}, for any deterministic constructible expert, the prediction value on a level set $\cellnew$ is forced to be $\obslabel(\subsetindi_\cellnew)/\feasubsetprob(\subsetindi_\cellnew)$.
Thus each partition induces at most one deterministic constructible expert, and $\classifierfamdet$ is finite.

Let $\classifierfamdet_1\triangleq\{\classifiernew\in\classifierfamdet:\classifiernew\blackwelldom\givenclassifier_1\}$.
By assumption, the set $\classifierfamdet_1$ is nonempty and finite.
Strict Blackwell dominance is acyclic: a cycle $\classifier^{(0)}\blackwelldom\classifier^{(1)}\blackwelldom\cdots\blackwelldom\classifier^{(r)}=\classifier^{(0)}$ would force all integrated CDFs in the cycle to be identical, contradicting strictness of each edge.
Therefore $\classifierfamdet_1$ has an element $\classifier^\star$ that is undominated within $\classifierfamdet_1$.
We claim that $\classifier^\star\in\strictdomfamdet$.
If not, then there exists $\classifiernew\in\classifierfamdet$ such that $\classifiernew\blackwelldom\classifier^\star$.
Since $\classifier^\star\blackwelldom\givenclassifier_1$, transitivity gives $\classifiernew\blackwelldom\givenclassifier_1$, so $\classifiernew\in \classifierfamdet_1$, contradicting the choice of $\classifier^\star$ as undominated within $\classifierfamdet_1$.
Thus, we have $\classifier^\star\in\strictdomfamdet$ and $\classifier^\star\blackwelldom\givenclassifier_1$.
\end{proof}

We are now ready to prove \Cref{thm:det-search-hard}.
\begin{proof}[Proof of \Cref{thm:det-search-hard}]
Suppose there exists a polynomial-time algorithm for deterministic-output target-wise {\predicrefinelink} with two input experts.
Given an instance of {\rsslink}, construct the two-expert instance from \Cref{subsec:det-hard-instance}.
The construction has polynomial encoding size by \Cref{lem:det-hard-instance-properties}.
Run the assumed {\detpredicrefinelink} algorithm on the target expert $\givenclassifier_1$.

By \Cref{lem:det-strict-improvement-subset-sum,lem:det-undominated-extension}, the original subset-sum instance is a $\yesinst$ instance if and only if there exists $\classifier^\star\in\strictdomfamdet$ such that $\classifier^\star\blackwelldom\givenclassifier_1$.
By the specification of {\detpredicrefinelink}, the algorithm outputs such a strict improvement in the $\yesinst$ case and outputs $\givenclassifier_1$ in the $\noinst$ case.
Therefore we can decide {\rsslink} by answering $\yesinst$ if and only if the algorithm's output is not $\givenclassifier_1$.
This test is correct because a strict Blackwell improvement cannot equal $\givenclassifier_1$.
Since deterministic experts on the finite state space are explicitly represented, checking whether the output equals the constant expert $\givenclassifier_1\equiv1/2$ takes polynomial time.
Thus a polynomial-time {\detpredicrefinelink} algorithm would solve {\rsslink} in polynomial time.
By \Cref{lem:det-rss-np-hard}, this would imply $\Poly=\NP$.
\end{proof}

\subsection{No multiplicative PTAS for  \texorpdfstring{\detpredicrefineoptlink}{OPT-DetRefine}}
\label{subsec:det-opt-hardness}

We now prove \Cref{thm:det-opt-NO-mult-ptas}.
The proof uses the same hard instance and adds a constant loss gap.
For the Brier loss, the Bayes risk is $\brierloss(\prediction)=p(1-\prediction)$.
Thus, for every calibrated expert $\classifier$, $\brierlossval(\classifier)=\expect{\predrv_\classifier(1-\predrv_\classifier)}$.
In the hard instance, every calibrated expert has mean report $1/2$, and thus, we have
\begin{align*}
    \brierlossval(\classifier)
    =
    \expect{\predrv_\classifier-\predrv_\classifier^2}
    =
    \frac12-
    \expect{\predrv_\classifier^2}~.
\end{align*}

\begin{lemma}[Constant Brier-loss gap]
\label{lem:det-brier-gap}
For the hard two-expert instance constructed from a {\rsslink} instance, the deterministic-output Brier-loss optimum satisfies the following two claims.
\begin{enumerate}[label=(\roman*)]
    \item If the subset-sum instance is a $\noinst$ instance, then $\detbrieropt=1/4$.
    \item If the subset-sum instance is a $\yesinst$ instance, then $\detbrieropt\le2/9$.
\end{enumerate}
\end{lemma}

\begin{proof}
First suppose the subset-sum instance is a $\noinst$ instance.
By \Cref{lem:det-hard-binary-vectors}, the only binary vectors in $\{0,1\}^{\hardground}\cap\obsspace$ are $\vzero$ and $\vone$.
Let $\classifier\in\classifierfamdet$.
By \Cref{lem:det-cells-observable}, every level-set indicator of $\classifier$ lies in $\obsspace$.
Thus, every nonempty level set has indicator $\vone$.
Therefore $\classifier$ is constant.
Since $\classifier$ is constructible, it is calibrated by \Cref{lem:linear-constructible-calibrated}; since the base rate is $1/2$, the constant value must be $1/2$.
Thus every deterministic constructible expert is equal to $\givenclassifier_1\equiv1/2$.
In particular, every feasible expert for  {\detpredicrefineoptlink} with target $\givenclassifier_1$ has Brier loss $1/4$, and $\givenclassifier_1$ itself is feasible.
Therefore $\detbrieropt=1/4$.

Now suppose the subset-sum instance is a $\yesinst$ instance, and let $\hardset\subseteq[\harditemnum]$ satisfy $\sum\nolimits_{i\in\hardset}\subsuminteger_i=\hardtarget$.
Let $\hardcell\triangleq \hardset\cup\hardposblock$ and $\signalvec\triangleq\subsetindi_\hardcell$.
By \Cref{lem:det-hard-binary-vectors}, $\signalvec\in\obsspace$.
Since $\vone\in\obsspace$, also $\vone-\signalvec\in\obsspace$.
Both vectors are nonzero and coordinate-wise nonnegative, so they belong to $\obscone$.
They define a two-cell deterministic constructible expert, denoted by $\classifier_\hardset$.
Let $\alpha\triangleq\feasubsetprob(\hardcell)=\feasubsetprob(\signalvec)$ and define $d\triangleq\hardblocksize/(4\feanum)$.
Since $\hardblocksize=\harditemnum$ and $\feanum=3\harditemnum$, we have $d=1/12$.
Using \Cref{lem:observable-label-functional} and Eqn.~\eqref{eq:det-hard-label-identity}, the two prediction values of $\classifier_\hardset$ are
\begin{align*}
    \prediction_\hardcell
    =
    \frac{\obslabel(\signalvec)}{\feasubsetprob(\signalvec)}
    =
    \frac12+\frac{d}{\alpha}~,
    \quad
    \prediction_{\hardground\setminus\hardcell}
    =
    \frac{\obslabel(\vone-\signalvec)}{\feasubsetprob(\vone-\signalvec)}
    =
    \frac12-\frac{d}{1-\alpha}~.
\end{align*}
The quadratic moment of its prediction distribution is
\begin{align*}
    \expect{\predrv_{\classifier_\hardset}^2}
    &=
    \alpha\left(\frac12+\frac d\alpha\right)^2
    +(1-\alpha)\left(\frac12-\frac d{1-\alpha}\right)^2 \\
    &=
    \frac14+d^2\left(\frac1\alpha+\frac1{1-\alpha}\right)\\
    & \ge
    \frac14+4d^2
    =
    \frac14+\frac1{36}
    =
    \frac5{18}~. \tag{By $1/\alpha+1/(1-\alpha)\ge4$ for $\alpha\in(0,1)$}
\end{align*}
The expert $\classifier_\hardset$ is constructible, and hence it is calibrated; it is also nonconstant.
By \Cref{lem:det-nonconstant-dominates-constant}, $\classifier_\hardset\blackwelldom\givenclassifier_1$, so it is feasible for  {\detpredicrefineoptlink}.
Using $\expect{\predrv_{\classifier_\hardset}}=1/2$, its Brier loss satisfies
\begin{align*}
    \brierlossval(\classifier_\hardset)
    =
    \frac12-
    \expect{\predrv_{\classifier_\hardset}^2}
    \le
    \frac12-\frac5{18}
    =
    \frac29~.
\end{align*}
Thus, we have $\detbrieropt\le \brierlossval(\classifier_\hardset)\le2/9$.
\end{proof}

\begin{proof}[Proof of \Cref{thm:det-opt-NO-mult-ptas}]
Suppose, for contradiction, that there exists a multiplicative PTAS for  {\detpredicrefineoptlink} with Brier loss.
Run this algorithm with the fixed accuracy $\varepsilon_0=1/16$.
Since $\varepsilon_0$ is fixed, the running time is polynomial in the input size.
Given an instance of {\rsslink}, construct the hard two-expert instance from \Cref{subsec:det-hard-instance} and run the assumed PTAS.

If the subset-sum instance is a $\noinst$ instance, then by \Cref{lem:det-brier-gap}, every feasible deterministic constructible expert has Brier loss exactly $1/4$.
Thus, the PTAS output $\classifier^{\varepsilon_0}$ satisfies $\brierlossval(\classifier^{\varepsilon_0})=1/4$.
If the subset-sum instance is a $\yesinst$ instance, then \Cref{lem:det-brier-gap} gives $\detbrieropt\le2/9$.
The multiplicative guarantee gives
\begin{align*}
    \brierlossval(\classifier^{\varepsilon_0})
    \le
    \left(1+\frac1{16}\right)\detbrieropt
    \le
    \frac{17}{72}
    <
    \frac14~.
\end{align*}
Therefore the following polynomial-time test decides the original {\rsslink} instance: answer $\yesinst$ if and only if $\brierlossval(\classifier^{\varepsilon_0})<35/144$.
On NO-instances, the output loss is exactly $1/4=36/144$, so the test answers $\noinst$.
On Yes-instances, the output loss is at most $17/72=34/144$, so the test answers $\yesinst$.
The constructed instance has polynomial encoding size by \Cref{lem:det-hard-instance-properties}.
Under the standard explicit rational representation of finite deterministic experts, the Brier loss of the output expert can be computed and compared with the rational threshold $35/144$ in polynomial time.
Thus the assumed multiplicative PTAS would solve {\rsslink} in polynomial time, contradicting \Cref{lem:det-rss-np-hard} unless $\Poly=\NP$.
\end{proof}

\newpage
\bibliography{mybib}
\appendix
\section{Missing Proofs}

\begin{proof}[Proof of \Cref{cor:smooth-proper-loss-fptas}]
If $B=0$, then $\bayesrisk$ is affine on $[0,1]$, and one affine upper bound is exact everywhere.
Thus one affine piece suffices.
Assume now that $B>0$.

For every report $\prediction\in[0,1]$, define the affine function
\begin{align*}
    \minorant_\prediction(\bayesprob)
    \triangleq
    \exploss(\prediction,\bayesprob)
    =
    \bigl(\loss(\prediction,1)-\loss(\prediction,0)\bigr)\bayesprob
    +
    \loss(\prediction,0)~.
\end{align*}
By properness, $\bayesrisk(\bayesprob)\le \minorant_\prediction(\bayesprob)$ for every $\prediction,\bayesprob\in[0,1]$, and equality holds when $\prediction=\bayesprob$.
When $\bayesrisk$ is differentiable, $\minorant_\prediction$ is the tangent affine upper bound to the concave function $\bayesrisk$ at $\prediction$.

Let $\grid_\eps$ be a uniform grid on $[0,1]$ with mesh size $\eta$.
For every $\bayesprob\in[0,1]$, choose $\prediction\in\grid_\eps$ such that $|\bayesprob-\prediction|\le\eta/2$.
Taylor's theorem and the curvature bound give
\begin{align*}
    0
    \le
    \minorant_\prediction(\bayesprob)-\bayesrisk(\bayesprob)
    \le
    \frac{B}{2}(\bayesprob-\prediction)^2
    \le
    \frac{B\eta^2}{8}~.
\end{align*}
Taking $\eta=\sqrt{8\eps/B}$ yields
\begin{align*}
    \min\nolimits_{\prediction\in\grid_\eps}
    \minorant_\prediction(\bayesprob)
    \le
    \bayesrisk(\bayesprob)+\eps
    \quad
    \text{for every }\bayesprob\in[0,1]~.
\end{align*}
Since each $\minorant_\prediction$ is an affine upper bound on $\bayesrisk$, we also have
$\min_{\prediction\in\grid_\eps}\minorant_\prediction(\bayesprob)\ge\bayesrisk(\bayesprob)$.
Thus $\bayesrisk$ admits an $\eps$-accurate affine upper approximation using
$O(\sqrt{B/\eps}+1)$ affine pieces.
This proves regularity.
\end{proof}

\end{document}